\newif\ifsubmission
  \submissionfalse

\ifsubmission
\documentclass{llncs}

\else
\documentclass[11pt]{article}
\usepackage{fullpage}
\fi

\ifsubmission \else
\usepackage{palatino}
\fi
\usepackage[normalem]{ulem}
\usepackage{framed}
\usepackage[utf8]{inputenc} 
\usepackage[pdfstartview=FitH,pdfpagemode=UseNone,colorlinks,linkcolor=blue,filecolor=blue,citecolor=Violet,urlcolor=red]{hyperref}
\usepackage{cmap} 
\usepackage[T1]{fontenc} 
\usepackage{multirow}
\usepackage{float}
\usepackage[section]{placeins}         
\usepackage{mathtools}  
\usepackage[dvipsnames]{xcolor}
\usepackage{amsmath,amsthm,amssymb,algpseudocode,algorithm,cryptocode}
\usepackage{mathrsfs} 
\usepackage[capitalize]{cleveref}
\usepackage{dashbox} 
\usepackage{braket}
\usepackage{physics}
\usepackage{xspace}  

\usepackage{tikz}
\usetikzlibrary{quantikz} 

\usepackage{soul}
\usepackage{enumitem}
\usepackage{breakcites}

\newif\ifnotes
\notestrue
\newcommand{\authnote}[3]{\textcolor{#3}{[{\footnotesize {\bf #1:} { {#2}}}]}}
\newcommand{\james}[1]{\ifnotes \authnote{J}{#1}{red} \fi}
\newcommand{\nishant}[1]{\ifnotes \authnote{N}{#1}{mafootgenta} \fi}
\newcommand{\dakshita}[1]{\ifnotes \authnote{D}{#1}{orange} \fi}
\newcommand{\amit}[1]{\ifnotes \authnote{A}{#1}{Violet} \fi}

\newif\ifforlater
\forlatertrue
\newcommand{\fflater}[1]{\ifforlater \else {#1} \fi}

\newcommand{\gray}[1]{\textcolor{gray}{\mathsf{#1}}}
\newcommand{\nblue}[1]{\textcolor{blue}{{#1}}}

\ifsubmission

\newtheorem{axiom}[theorem]{Axiom}
\newtheorem{physicsaxiom}[theorem]{Physics Axiom}
\newtheorem{importedtheorem}[theorem]{Imported Theorem}
\newtheorem{importedlemma}[theorem]{Imported Lemma}
\newtheorem{informaltheorem}[theorem]{Informal Theorem}
\newtheorem{physicstheorem}[theorem]{Physical Theorem}

\newtheorem{claim}[theorem]{Claim}
\newtheorem{subclaim}[theorem]{SubClaim}

\newtheorem{fact}[theorem]{Fact}

\newtheorem{construction}[theorem]{Construction}
\Crefname{importedtheorem}{Imported Theorem}{Imported Theorems}
\Crefname{importedlemma}{Imported Lemma}{Imported Lemma}
\Crefname{theorem}{Theorem}{Theorems}
\Crefname{proposition}{Proposition}{Propositions}
\Crefname{claim}{Claim}{Claims}
\Crefname{subclaim}{SubClaim}{SubClaims}
\Crefname{subsubclaim}{SubSubClaim}{SubSubClaims}
\Crefname{lemma}{Lemma}{Lemmas}
\Crefname{conjecture}{Conjecture}{Conjectures}
\Crefname{corollary}{Corollary}{Corollaries}
\Crefname{construction}{Construction}{Constructions}
\Crefname{property}{Property}{Properties}

\theoremstyle{definition}

\newtheorem{assumption}[theorem]{Assumption}
\newtheorem{notation}[theorem]{Notation}
\Crefname{definition}{Definition}{Definitions}
\Crefname{assumption}{Assumption}{Assumptions}
\Crefname{notation}{Notation}{Notations}

\theoremstyle{remark}

\newtheorem{comment}[theorem]{Comment}

\Crefname{question}{Question}{Questions}
\Crefname{remark}{Remark}{Remarks}
\Crefname{comment}{Comment}{Comments}
\Crefname{fact}{Fact}{Facts}
\Crefname{step}{Step}{Steps}

\else

\newtheorem{theorem}{Theorem}[section]

\newtheorem{importedtheorem}[theorem]{Imported Theorem}

\newtheorem{claim}[theorem]{Claim}
\newtheorem{subclaim}[theorem]{SubClaim}

\newtheorem{lemma}[theorem]{Lemma}

\newtheorem{corollary}[theorem]{Corollary}

\newtheorem{definition}[theorem]{Definition}

\Crefname{importedtheorem}{Imported Theorem}{Imported Theorems}
\Crefname{importedlemma}{Imported Lemma}{Imported Lemmas}
\Crefname{theorem}{Theorem}{Theorems}
\Crefname{proposition}{Proposition}{Propositions}
\Crefname{claim}{Claim}{Claims}
\Crefname{subclaim}{SubClaim}{SubClaims}
\Crefname{subsubclaim}{SubSubClaim}{SubSubClaims}
\Crefname{lemma}{Lemma}{Lemmas}
\Crefname{conjecture}{Conjecture}{Conjectures}
\Crefname{corollary}{Corollary}{Corollaries}
\Crefname{construction}{Construction}{Constructions}
\Crefname{property}{Property}{Properties}

\theoremstyle{definition}

\Crefname{definition}{Definition}{Definitions}
\Crefname{assumption}{Assumption}{Assumptions}
\Crefname{notation}{Notation}{Notations}

\theoremstyle{remark}

\Crefname{question}{Question}{Questions}
\Crefname{comment}{Comment}{Comments}
\Crefname{fact}{Fact}{Facts}
\Crefname{step}{Step}{Steps}

\fi

\newcommand{\secp}{\lambda}

\def\cA{{\cal A}}
\def\cB{{\cal B}}
\def\cC{{\cal C}}
\def\cD{{\cal D}}
\def\cE{{\cal E}}
\def\cF{{\cal F}}

\def\cH{{\cal H}}

\def\cK{{\cal K}}

\def\cP{{\cal P}}

\def\cR{{\cal R}}
\def\cS{{\cal S}}
\def\cT{{\cal T}}

\def\cW{{\cal W}}
\def\cX{{\cal X}}
\def\cY{{\cal Y}}
\def\cZ{{\cal Z}}

\def\bbC{{\mathbb C}}

\def\bbI{{\mathbb I}}

\def\bbN{{\mathbb N}}

\def\bbR{{\mathbb R}}

\newcommand{\bb}{\mathbf{b}}
\newcommand{\bc}{\mathbf{c}}

\newcommand{\bq}{\mathbf{q}}
\newcommand{\br}{\mathbf{r}}

\newcommand{\bv}{\mathbf{v}}

\newcommand{\bx}{\mathbf{x}}
\newcommand{\by}{\mathbf{y}}

\newcommand{\bR}{\mathbf{R}}

\newcommand{\brho}{\boldsymbol{\rho}}
\newcommand{\bsigma}{\boldsymbol{\sigma}}

\newcommand{\btheta}{\boldsymbol{\theta}}

\def\poly{{\rm poly}}

\def\negl{{\rm negl}}

\newcommand{\Commit}{\mathsf{Commit}}
\newcommand{\com}{\mathsf{com}}
\newcommand{\Com}{\mathsf{Com}}

\newcommand{\Sim}{\mathsf{Sim}}

\DeclareMathOperator*{\expectation}{\mathbb{E}}
\newcommand{\E}{\expectation}

\newcommand{\union}{\cup}

\newcommand{\Hyb}{\mathsf{Hyb}}

\newcommand{\proref}[1]{Protocol~\protect\ref{#1}}

\newenvironment{boxfig}[2]{\begin{figure}[#1]\fbox{
    \begin{minipage}{\linewidth}
    \vspace{0.2em}\makebox[0.025\linewidth]{}    \begin{minipage}{0.95\linewidth}{{#2 }}
    \end{minipage}\vspace{0.2em}\end{minipage}}}{\end{figure}}

\newcommand{\pprotocol}[4]{
\begin{boxfig}{h!}{
\begin{center}
\textbf{#1}
\end{center}
    #4
\vspace{0.2em} } \caption{\label{#3} #2}
\end{boxfig}
}

\newcommand{\protocol}[4]{
\pprotocol{#1}{#2}{#3}{#4} }

\newcommand{\out}{\mathsf{out}}
\newcommand{\state}{\mathsf{st}}

\newcommand{\zo}{\{0,1\}}

\newcommand{\fot}{\mathcal{F}_{\mathsf{OT}}}
\newcommand{\simu}{\mathsf{Sim}}

\newcommand{\Open}{\mathsf{Open}}
\newcommand{\Rec}{\mathsf{Rec}}

\renewcommand{\partial}{\mathsf{partial}}

\newcommand{\abort}{\mathsf{abort}}

\newcommand{\RO}{\mathsf{RO}}

\newcommand{\Ext}{\mathsf{Ext}}
\newcommand{\Equ}{\mathsf{Equ}}

\newcommand{\SimExt}{\mathsf{SimExt}}
\newcommand{\SimEqu}{\mathsf{SimEqu}}

\newcommand{\TD}{\mathsf{TD}}
\newcommand{\ptrace}{\mathsf{Tr}}

\newcommand{\simuequ}{\simu_{\Equ}}

\newcommand{\comro}{{H_{C}}}
\newcommand{\fsro}{H_{FS}}
\newcommand{\extro}{H_{Ext}}

\newcommand{\puncrom}[1]{{#1}^\bot}

\newcommand{\malrecvlamb}{\cR^*_\secp}

\newcommand{\simurecvlamb}{{\simu_{\secp,\malrecvlamb}}}

\newcommand{\wt}{\mathsf{wt}}

\newcommand{\RCommit}{\mathsf{RCommit}}
\newcommand{\ROpen}{\mathsf{ROpen}}

\newcommand{\simro}{\Sim_{\mathsf{RO}}}

\newcommand{\rcv}{\mathsf{R}}
\newcommand{\Adv}{\mathsf{Adv}}

\newcommand{\sendr}{\mathsf{S}}
\newcommand{\recv}{\mathsf{R}}
\newcommand{\OT}{\mathsf{OT}}
\newcommand{\EndSROT}{\mathsf{End}-\mathsf{S}-\mathsf{ROT}}
\newcommand{\EndRROT}{\mathsf{End}-\mathsf{R}-\mathsf{ROT}}
\newcommand{\SROT}{\mathsf{S}-\mathsf{ROT}}
\newcommand{\RROT}{\mathsf{R}-\mathsf{ROT}}

\newcommand{\maj}{\mathsf{maj}}
\newcommand{\simroco}{\Sim_{\mathsf{RO}}^{\mathsf{CO}}}

\newcommand{\sA}{\mathsf{A}}
\newcommand{\sB}{\mathsf{B}}
\newcommand{\sC}{\mathsf{C}}
\newcommand{\sD}{\mathsf{D}}
\newcommand{\sE}{\mathsf{E}}

\newcommand{\sR}{\mathsf{R}}
\newcommand{\sS}{\mathsf{S}}

\newcommand{\strq}{\bq}
\begin{document}
\ifsubmission
\title{A Fixed Basis Framework for Quantum Oblivious Transfer
\thanks{A full version of this paper is also attached as supplementary material.}}
\author{}
\institute{}
\else
\title{A New Framework for Quantum Oblivious Transfer}
\author{Amit Agarwal\thanks{UIUC. Email: \texttt{amita2@illinois.edu}} \and James Bartusek\thanks{UC Berkeley. Email: \texttt{bartusek.james@gmail.com}} \and Dakshita Khurana\thanks{UIUC. Email: \texttt{dakshita@illinois.edu}} \and Nishant Kumar
\thanks{UIUC.}}
\date{}
\fi
\maketitle


\begin{abstract}

We present a new template for building oblivious transfer from quantum information that we call the ``fixed basis'' framework. 
Our framework departs from prior work (eg., Crepeau and Kilian, FOCS '88) by fixing the {\em correct} choice of measurement basis used by each player,
except for some hidden {\em trap} qubits that are intentionally measured in a conjugate basis.
%
We instantiate this template in the quantum random oracle model (QROM) to obtain simple protocols that implement, with security against malicious adversaries: 
\begin{itemize}
    \item {\em Non-interactive} random-input bit OT in a model where parties share EPR pairs a priori.
    \item Two-round random-input bit OT without setup, obtained by showing that the protocol above remains secure even if the (potentially malicious) OT receiver sets up the EPR pairs.
    \item Three-round chosen-input string OT from BB84 states without entanglement or setup.
    This improves upon natural variations of the CK88 template that require at least five rounds.
\end{itemize}
Along the way, we develop technical tools that may be of independent interest.
We prove that natural functions like XOR enable \emph{seedless} randomness extraction from certain quantum sources of entropy. 
We also use idealized (i.e. extractable and equivocal) bit commitments, which we obtain by proving security of simple and efficient constructions in the QROM. 
\end{abstract}

\ifsubmission
\else
\newpage
\thispagestyle{empty}
\phantom{.}
\vfill
\begin{center}
{\em In loving memory of Nishant (December 2, 1994 - April 10, 2022),\\ who led this research and was deeply passionate about cryptography.}
\end{center}
\vspace{5em}
\vfill

\newpage
{
  \hypersetup{linkcolor=Violet}
  \setcounter{tocdepth}{2}
  \tableofcontents
}
\newpage
\fi

\section{Introduction}

Stephen Wiesner’s celebrated paper~\cite{Wiesner1983ConjugateC} that kickstarted the field of quantum cryptography suggested a way to use quantum information in order to achieve {\em a means for transmitting two messages either but not both of which may be received}. Later, it was shown that this powerful primitive -- named oblivious transfer (OT)~\cite{Rabin81,EGL} -- serves as the foundation for secure computation~\cite{STOC:GolMicWig87,STOC:Kilian88}, which is a central goal of modern crytography. 

Wiesner's original proposal only required uni-directional communication, from the sender to the receiver. However, it was not proven secure, and succesful attacks on the proposal (given the ability for the receiver to perform multi-qubit measurements) where even discussed in the paper. Later, \cite{FOCS:CreKil88} suggested a way to use both \emph{interaction} and \emph{bit commitments} (which for example can be instantiated using cryptographic hash functions) to obtain a secure protocol. In this work, we investigate how much interaction is really required to obtain oblivious transfer from quantum information (and hash functions). In particular, we ask

\begin{center}
{\em Can a sender non-interactively transmit two bits to a receiver \\ such that the receiver will be able to recover one but not both of the bits?}
\end{center}

We obtain a positive answer to this question \emph{if the sender and receiver share prior entanglement}, and we analyze the (malicious, simulation-based) security of our protocol in the quantum random oracle model (QROM).

Specifically, we consider the EPR setup model, where a sender and receiver each begin with halves of EPR pairs, which are maximally entangled two-qubit states $\frac{\ket{00} + \ket{11}}{\sqrt{2}}$. Such simple entangled states are likely to be a common shared setup in quantum networks (see e.g. \cite{10.1145/3387514.3405853} and references therein), and have attracted much interest as a quantum analogue of the classical common reference string (CRS) model \cite{10.1007/978-3-540-24587-2_20,CVZ,https://doi.org/10.48550/arxiv.2102.09149,https://doi.org/10.48550/arxiv.2204.02265}.
They have already been shown to be useful for many two-party tasks such as quantum communication via teleportation~\cite{teleport}, entanglement-assisted quantum error correction~\cite{BDH}, and even cryptographic tasks like key distribution~\cite{Ekert91} and non-interactive zero-knowledge~\cite{CVZ,https://doi.org/10.48550/arxiv.2102.09149}. 

\paragraph{Non-interactive Bit OT in the EPR Setup Model.}
We show that once Alice and Bob share a certain (fixed) number of EPR pairs between them, they can realize a \emph{one-shot}\footnote{We use the terms "one-shot", "one-message", and "non-interactive" interchangably in this work, all referring to a protocol between two parties Alice and Bob that consists only of a single message from Alice to Bob.} bit OT protocol, {\em securely} implementing an ideal functionality that takes two \emph{bits} $m_0,m_1$ from Alice and delivers $m_b$ for a uniformly random $b \gets \{0,1\}$ to Bob. We provide an unconditionally secure protocol in the QROM, and view this as a first step towards protocols that rely on concrete properties of hash functions together with entanglement setup.
    
Furthermore, our result helps understand the power of entanglement as a cryptographic resource. Indeed, non-interactive oblivious transfer is impossible to achieve classically, under any computational assumption, even in the common reference string and/or random oracle model. Thus, the only viable one-message solution is to assume the parties already start with so-called \emph{OT correlations}, where the sender gets random bits $x_0,x_1$ from a trusted dealer, and the receier gets $x_b$ for a random bit $b.$ On the other hand, our result shows that OT can be acheived in a one-shot manner just given shared EPR pairs.

We note that an ``OT correlations setup'' is fundementally different than an EPR pair setup. First of all, OT correlations are \emph{specific to OT}, while, as desribed above, shared EPR pairs are already known to be broadly useful, and have been widely studied independent of OT. Moreover, an OT correlations setup requires \emph{private} (hidden) randomness, while generating EPR pairs is a deterministic process. In particular, any (even semi-honest) dealer that sets up OT correlations can learn the parties' private inputs by observing the resulting transcript of communication, while this is not necesarily true of an EPR setup by monogamy of entanglement. 
Furthermore, as we describe next, our OT protocol remains secure even if a {\em potentially malicious receiver} dishonestly sets up the entanglement.
\paragraph{Two-Message Bit OT without Setup.}
The notion of two-message oblivious transfer has been extensively studied in the classical setting \cite{EC:AieIshRei01,NP01,C:PeiVaiWat08,HalKal,EC:DGHMW20} and is of particular theoretical and practical interest. We show that the above protocol remains secure even if the receiver were the one performing the EPR pair setup (as opposed to a trusted dealer / network administrator). That is, we consider a two-message protocol where the receiver first sets up EPR pairs and sends one half of every pair to the sender, following which the sender sends a message to the receiver as before. We show that this protocol also realizes the same bit OT functionality with random receiver choice bit.

This results in the first two-message maliciously-secure variant of OT, without setup, that does not (necessarily) make use of public-key cryptography. 
However, we remark that we still only obtain the random receiver input functionality in this setting, and leave a construction of two-message chosen-input string OT without public-key cryptography as an intriguing open problem.

\paragraph{Another Perspective: OT Correlations from Entanglement via 1-out-of-2 Deletion.}
It is well-known that shared halves of EPR pairs can be used to generate shared randomness by having each player measure their halves of EPR pairs in a common basis. 
But can they also be used to generate OT correlations, where one of the players (say Alice) outputs a random pair of bits, while the other (say Bob) learns only {\em one} of these (depending on a hidden choice bit), and cannot guess the other bit?\footnote{While this framing of the problem is different from the previous page, the two turn out to be equivalent thanks to OT reversal and reorientation methods~\cite{C:IKNP03}.}

At first, it may seem like the following basic property of EPR pairs gives a candidate solution that requires \emph{no} communication: if Alice and Bob measure their halves in the same basis (say, both computational, hereafter referred to as the $+$ basis), then they will obtain the same random bit $r$, while if Alice and Bob measure their halves in conjugate bases (say, Alice in the $+$ basis and Bob in the Hadamard basis, hereafter referred to as the $\times$ basis), then they will obtain random and \emph{independent} bits $r_A,r_B$. Indeed, if Alice and Bob share two EPR pairs, they could agree that Alice measures both of her halves in either the $+$ basis or the $\times$ basis depending on whether her choice bit is $0$ or $1$, while Bob always measures his first half in the $+$ basis and his second half in the $\times$ basis. Thus, Bob obtains $(r_0,r_1)$, and, depending on her choice $b$, Alice obtains $r_b$, while {\em deleting} information about $r_{1-b}$ by measuring the corresponding register in a conjugate basis. 
%

Of course, there is nothing preventing Alice from simply measuring her first half in the $+$ basis and her second half in the $\times$ basis, obtaining both $r_0,r_1$ and rendering this initial candidate completely insecure. However, what if Alice could \emph{prove} to Bob that she indeed measured both qubits in the same basis, \emph{without} revealing to Bob which basis she chose? Then, Bob would be convinced that one of his bits is independent of Alice's view, while the privacy of Alice's choice $b$ would remain intact. We rely on the Random Oracle to implement a cut-and-choose based proof that helps us obtain secure bit OT.

We emphasize that this problem is also interesting in the plain model under computational assumptions. We leave this as an open problem for future work, and discuss it (together with other open problems arising from this work) in Section~\ref{subsec:future-work}.


\paragraph{Other Technical Contributions.}
We make additional technical contributions along the way, that may be of independent interest.
\begin{itemize}
\item {\bf Seedless Extraction from Quantum Sources of Entropy.} Randomness extraction has been a crucial component in all quantum OT protocols, and \emph{seeded} randomness extraction from the quantum sources of entropy that arise in such protocols has been extensively studied (see e.g. \cite{TCC:RenKon05,C:BouFeh10}). In our non-interactive and two-message settings, it becomes necessary to extract entropy without relying on the existence of a random seed. 
As such, we prove the security of \emph{seedless} randomness extractors in this context, which may be of independent interest. In particular, we show that either the XOR function or a random oracle (for better rate) can be used in place of the seeded universal hashing used in prior works. The XOR extractor has been used in subsequent work~\cite{cryptoeprint:2022/1178} as a crucial tool in building cryptosystems with certified deletion.
\item {\bf Extractable and Equivocal Commitments in the QROM.}
We abstract out a notion of (non-interactive) extractable and equivocal bit commitments in the quantum random oracle model, that we make use of in our OT protocols. We provide a simple construction based on prior work ~\cite{C:AmbHamUnr19,C:Zhandry19,DFMS21}. 

\item {\bf Three-Message String OT without Entanglement or Setup.}
We show that our fixed basis framework makes it possible to eliminate the need for both entanglement and setup with just three messages.  The resulting protocol realizes string OT with no entanglement, and only requires one quantum message containing BB84 states followed by two classical messages. Furthermore, it allows both the sender and the receiver to {\em choose} their inputs to the OT (as opposed to sampling a random input to one of the parties).

On the other hand, we find that using prior templates~\cite{FOCS:CreKil88} necessitates a multi-stage protocol where players have to first exchange basis information in order to establish two channels, resulting in protocols that require at least an extra round of interaction.

\item {\bf Concrete Parameter Estimates.} We also 
estimate the number of EPR pairs/BB84 states required for each of our protocols, and derive concrete security losses incurred by our protocols. This is discussed in \cref{subsec:concrete-parameters}, where we also provide a table of our estimates. 
We expect that future work will be able to further study and optimize the concrete efficiency of quantum OT in the QROM, and our work provides a useful starting point.
\end{itemize}

\subsection{Open problems and directions for future research.}\label{subsec:future-work}
Our new frameworks for oblivious transfer raise several fundamental questions of both theoretical and practical interest. 
\paragraph{Strengthening Functionality.}
It would be interesting to obtain non-interactive or two-message variants of non-trivial quantum OT realizing stronger functionality than we obtain in this work\footnote{
Here {\em non-trivial} quantum OT refers to OT that is based on assumptions (such as symmetric-key cryptography) or ideal models that are not known to imply classical OT.}. Our work leaves open the following natural questions.

\begin{itemize}
    \item Does there exist two-message non-trivial quantum {\em chosen-input} bit OT, that allows both parties to choose inputs?
    \item Does there exist one- or two-message non-trivial quantum chosen-sender-input \emph{string} OT, with chosen sender strings and random receiver choice bit? Such a string OT may be sufficient to construct non-interactive secure computation (NISC)~\cite{C:IKOPSW11} with chosen sender input and random receiver input.
    \item Does there exist two-message non-trivial quantum OT without entanglement?
    \item Can our quantum OT protocols serve as building blocks for other non-interactive functionalities, eg., by relying on techniques in~\cite{C:GIKOS15} for one-way secure computation, or~\cite{EC:BitVai17} for obfuscation?
\end{itemize}


\paragraph{Strengthening Security.} While analyses in this work are restricted to the QROM, our frameworks are of conceptual interest even beyond this specific model. In particular, one could ask the following question.
\begin{itemize}
    \item Does there exist non-interactive OT with shared EPR pair setup from {\em any concrete computational hardness assumption}?
\end{itemize}
One possible direction towards acheiving this would be to instantiate our template with post-quantum extractable and equivocal commitments in the CRS model, and then attempt to instantiate the Fiat-Shamir paradigm in this setting based on a concrete hash function (e.g.~\cite{CGH04,C:KalRotRot17,STOC:CCHLRRW19} and numerous followups).
Going further, one could even try to instantiate our templates from weak computational hardness including one-way functions (or even pseudorandom states). We imagine that such an OT would find useful applications even beyond MPC, given how two-message classical OT~\cite{EC:AieIshRei01,NP01} has been a key building block in several useful protocols including two-message proof systems, non-malleable commitments, and beyond~\cite{C:OstPasPas14,AC:BGISW17,C:JKKR17,FOCS:KhuSah17,TCC:BGJKS17,EC:KalKhuSah18,C:BGJKKS18}.

Finally, we note that any cryptographic protocol in a broader context typically requires the protocol to satisfy strong composability properties. It would be useful to develop a formal model for UC security with a (global) quantum random oracle, and prove UC security for our OT protocols in this model. Another question is whether one can achieve composably (UC) secure protocols with minimal interaction by building on our frameworks in the CRS model.



\paragraph{Practical Considerations.}
Our concrete quantum resource requirements and security bounds are computed assuming no transmission errors.
On the other hand, actual quantum systems, even those that do not rely on entanglement, are often prone to errors. One approach to reconcile these differences is to employ techniques to first improve fidelity, eg. of our EPR pair setup via entanglement purification; and then execute our protocol on the resulting states. Another natural approach (following eg.,~\cite{C:BBCS91}) could involve directly building error-resilient versions of our protocols that tolerate low fidelity and/or coherence.
Another question is whether our games can be improved to reduce resource consumption and security loss, both in the idealized/error-free and error-prone models.

\subsection{Related Work}
Wiesner~\cite{Wiesner1983ConjugateC} suggested the first template for quantum OT, but his work did not contain a security proof (and even discussed some potential attacks). Crepeau and Kilian~\cite{FOCS:CreKil88} made progress by demonstrating an approach for basing oblivious transfer on properties of quantum information \emph{plus} a secure "bit commitment" scheme. This led to interest in building bit commitment from quantum information. Unfortunately, it was eventually shown by Mayers, Lo, and Chau~\cite{Mayers97,LoChau97} that bit commitment (and thus oblivious transfer) is \emph{impossible} to build by relying solely on the properties of quantum information. 

This is indeed a strong negative result, and rules out the possibility of basing secure computation on quantum information alone. However, it was still apparent to researchers that quantum information must offer \emph{some} advantage in building secure computation systems. One could interpret the Mayers, Lo, Chau impossibility result as indicating that in order to hone in and understand this advantage, it will be necessary to make additional physical, computational, or modeling assumptions beyond the correctness of quantum mechanics. Indeed, much research has been performed in order to tease out the answer to this question, with three lines of work being particularly prominent and relevant to this work\footnote{Another line of work studies (unconditional) oblivious transfer with \emph{imperfect} security \cite{10.5555/2481591.2481600,Chailloux2016OptimalBF,Kundu2020ADP}, which we view as largely orthogonal to our work.}.

\begin{itemize}
    \item \textbf{Quantum OT from bit commitment.} Although unconditionally-secure bit commitment cannot be constructed using quantum information, \cite{FOCS:CreKil88}'s protocol is still meaningful and points to a fundamental difference between the quantum and classical setting, where bit commitment is not known to imply OT. A long line of work has been devoted to understanding the security of \cite{FOCS:CreKil88}'s proposal: e.g. \cite{C:BBCS91,MayersSalvail94,STOC:Yao95,C:DFLSS09,EC:Unruh10,C:BouFeh10}.
    \item \textbf{Quantum OT in the bounded storage model.} One can also impose physical assumptions in order to recover quantum OT with unconditional security. \cite{DFSS} introduced the \emph{quantum bounded-storage model}, and \cite{noisy-storage} introduced the more general \emph{quantum noisy-storage model}, and showed how to construct unconditionally-secure quantum OT in these idealized models. There has also been much followup work focused on implementation and efficiency \cite{Wehner2010,Erven2014,Ito2017,Furrer2018}.
    \item \textbf{Quantum OT from "minicrypt" assumptions.} While \cite{FOCS:CreKil88}'s proposal for obtaining OT from bit commitment scheme suggests that public-key cryptography is not required for building OT in a quantum world, a recent line of work has been interested in identifying the \emph{weakest} concrete assumptions required for quantum OT, with \cite{BCKM2021,GLSV} showing that the existence of one-way functions suffices and \cite{cryptoeprint:2021:1691,cryptoeprint:2021:1663} showing that the existence of pseudo-random quantum states suffices.
\end{itemize}

Our work initiates the explicit study of quantum oblivious transfer in the \emph{quantum random oracle model}, a natural model in which to study \emph{unconditionally-secure} quantum oblivious transfer.
Any protocol proven secure in the idealized random oracle model immediately gives rise to a natural "real-world" protocol where the oracle is replaced by a cryptographic hash function, such as SHA-256. As long as there continue to exist candidate hash functions with good security against quantum attackers, our protocols remain useful and relevant. On the other hand, the bounded storage model assumes an upper bound on the adversary's quantum storage while noisy storage model assumes that any qubit placed in quantum memory undergoes a certain amount of noise. The quantum communication complexity of these protocols increases with the bounds on storage/noise.
It is clear that advances in quantum storage and computing technology will steadily degrade the security and increase the cost of such protocols, whereas protocols in the QROM do not suffer from these drawbacks.
\section{Technical overview}

\paragraph{Notation.} 
We will consider the following types of OT protocols.
\begin{itemize}
    \item $\cF_{\OT[k]}$: the \emph{chosen-input string} OT functionality takes as input a bit $b$ from the receiver and two strings $m_0,m_1 \in \{0,1\}^k$ from the sender. It delivers $m_b$ to the receiver.
    \item $\cF_{\RROT[1]}$: the \emph{random-receiver-input bit} OT functionality takes as input $\top$ from the receiver and two bits $m_0,m_1 \in \{0,1\}$ from the sender. It samples $b \gets \{0,1\}$ and delivers $(b,m_b)$ to the receiver.
    \item $\cF_{\SROT[k]}$: the \emph{random-sender-input string} OT functionality takes as input $\top$ from the sender and $(b,m)$ from the receiver for $b \in \{0,1\}, m \in \{0,1\}^k$. It set $m_b = m$, samples $m_{1-b} \gets \{0,1\}^k$ and delivers  $(m_0,m_1)$ to the sender.
\end{itemize}

\subsection{Non-Interactive OT in the shared EPR pair model}
\label{subsec:overniot}
As discussed in the introduction, there is a skeleton candidate OT protocol that requires no communication in the shared EPR model that we describe in Figure~\ref{fig:basicot}.

\begin{figure}[ht!]
\begin{framed}
\begin{itemize}
    \item \textbf{Setup}: $2$ EPR pairs on registers $(\cA_0,\cB_0)$ and $(\cA_1,\cB_1)$, where Alice has registers $(\cA_0,\cA_1)$ and Bob has registers $(\cB_0,\cB_1)$.
    \item \textbf{Alice's output}: Input $b \in \{0,1\}$.
    \begin{enumerate}
        \item If $b = 0$, measure both of $\cA_{0},\cA_{1}$ in basis $+$ to obtain $r'_0,r'_1$. Output $r'_0$
        \item If $b = 1$, measure both of $\cA_{0},\cA_{1}$ in basis $\times$ to obtain $r'_0,r'_1$. Output $r'_1$.
    \end{enumerate}
    \item \textbf{Bob's output}: 
Measure $\cB_0$ in basis $+$ to obtain $r_0$ and $\cB_1$ in basis $\times$ to obtain $r_1$. Output $(r_0, r_1)$.
\end{itemize}
\end{framed}
\caption{An (insecure) skeleton OT candidate.}
\label{fig:basicot}
\end{figure}

The next step is for Alice to prove that she measured both her qubits in the same basis, without revealing what basis she chose.
While it is unclear how Alice could directly prove this, we could hope to rely on the cut-and-choose paradigm to check that she measured ``most'' out of a \emph{set} of pairs of qubits in the same basis. 
Indeed, a cut-and-choose strategy implementing a type of ``measurement check'' protocol has appeared in the original quantum OT proposal of~\cite{FOCS:CreKil88} and many followups. Inspired by these works, we develop such a strategy for our protocol as follows.

\paragraph{Non-interactive Measurement Check.} 
To achieve security, we first modify the protocol so that Alice and Bob use $2n$ EPR pairs, where Alice has one half of every pair and Bob has the other half.

Alice samples a set of $n$ bases $\theta_1,\dots,\theta_n \gets \{+,\times\}^n$. For each $i \in [n]$, she must measure the $i^{th}$ pair of qubits (each qubit corresponding to a half of an EPR pair) in basis $\theta_i$, obtaining measurement outcomes $(r_{i,0},r_{i,1})$. Then, she must commit to her bases and outcomes $\com(\theta_1,r_{1,0},r_{1,1}),\dots,$
$\com(\theta_n,r_{n,0},r_{n,1})$. 
Once committed, she must {\em open} commitments corresponding to a randomly chosen (by Bob) $T \subset [n]$ of size $k$, revealing $\{\theta_i,r_{i,0},r_{i,1}\}_{i \in T}$. 
Given these openings, for every $i \in T$, Bob will measure his halves of EPR pairs in bases $(\theta_i,\theta_i)$ to obtain $(r'_{i,0},r'_{i,1})$.
Bob aborts if his outcomes $(r'_{i,0},r'_{i,1})$ do not match Alice's claimed outcomes $(r_{i,0},r_{i,1})$ for any $i \in T$. If outcomes on all $i \in T$ match, we will say that Bob accepts the measurement check.

Now, suppose Alice passes Bob's check with noticeable probability.
Because she did not know the check subset $T$ at the time of committing to her measurement outcomes, we can conjecture that for ``most'' $i \in [n]\setminus T$, Alice also correctly committed to results of measuring her qubits in bases $(\theta_i, \theta_i)$. Moreover we can conjecture that the act of committing and passing Bob's check removed from Alice's view information about at least one out of $(r_{i,0}, r_{i,1})$ for most $i \in [n] \setminus T$. We build on techniques for analyzing quantum ``cut-and-choose'' protocols~\cite{C:DFLSS09,C:BouFeh10} to prove that this is the case.


In fact, we obtain a \emph{non-interactive} instantiation of such a measurement-check by leveraging the random oracle to perform the Fiat-Shamir transform.
That is, Alice applies a hash function, modeled as a random oracle, to her set of commitments in order to derive the ``check set'' $T$ of size $k$. Then, she can compute openings to the commitments in the set $T$, and finally send all of her $n$ commitments together with $k$ openings in a single message to Bob. 
Finally, the unopened positions will be used to derive two strings $(t_0,t_1)$ of $n-k$ bits each, with the guarantee that -- as long as Alice passes Bob's check -- there exists $b$ such that Alice only has partial information about the string $t_{1-b}$. 
We point out that to realize OT, it is not enough for Alice to only have partial information about $t_{1-b}$, we must in fact ensure that she obtains {\em no information} about $t_{1-b}$. We achieve this by developing techniques for \emph{seedless randomness extraction} in this setting, which we discuss later in this overview. 
The resulting protocol is described in \cref{fig:non-interactive-OT}.\footnote{Our actual protocol involves an additional step that allows Alice to program any input $m_b$ of her choice, but we suppress this detail in this overview.}


\begin{figure}[ht!]
\begin{framed}
\begin{itemize}
    \item \textbf{Setup}: Random oracle $\RO$ and $2n$ EPR pairs on registers $\{\cA_{i,b},\cB_{i,b}\}_{i \in [n], b \in \{0,1\}}$, where Alice has register $\cA \coloneqq \{\cA_{i,b}\}_{i \in [n], b \in \{0,1\}}$ and Bob has register $\cB \coloneqq \{\cB_{i,b}\}_{i \in [n], b \in \{0,1\}}$.
    \item \textbf{Alice's message}: Input $b \in \{0,1\}$.
    \begin{enumerate}
        \item Sample $\theta_1,\dots,\theta_n \gets \{+,\times\}^n$ and measure each $\cA_{i,0},\cA_{i,1}$ in basis $\theta_i$ to obtain $r_{i,0},r_{i,1}$.
        \item Compute commitments $\com_1,\dots,\com_n$ to $(\theta_1,r_{1,0},r_{1,1}),\dots,(\theta_n,r_{n,0},r_{n,1})$.
        \item Compute $T = \RO(\com_1,\dots,\com_n)$, where $T$ is parsed as a subset of $[n]$ of size $k$.
        \item Compute openings $\{u_i\}_{i \in T}$ for $\{\com_i\}_{i \in T}$.
        \item Let $\overline{T} = [n] \setminus T$, and for all $i \in \overline{T}$, set $d_i = b \oplus \theta_i$ (interpreting $+$ as 0 and $\times$ as 1).
        \item Send $\{\com_i\}_{i \in [n]}, T, \{r_{i,0},r_{i,1},\theta_i,u_i\}_{i \in T}, \{d_i\}_{i \in \overline{T}}$ to Bob.
    \end{enumerate}
    \item \textbf{Alice's output}: $m_b \coloneqq \mathsf{Extract}(t_b \coloneqq \{r_{i,\theta_i}\}_{i \in \overline{T}})$.
    \item \textbf{Bob's computation}: 
    \begin{enumerate}
        \item Abort if $T \neq \RO(\com_1,\dots,\com_n)$ or if verifying any commitment in the set $T$ fails.
        \item For each $i \in T$, measure registers $\cB_{i,0},\cB_{i,1}$ in basis $\theta_i$ to obtain $r_{i,0}',r_{i,1}'$, and abort if $r_{i,0} \neq r_{i,0}'$ or $r_{i,1} \neq r_{i,1}'$.
        \item For each $i \in \overline{T}$, measure register $\cB_{i,0}$ in the $+$ basis and register $\cB_{i,1}$ in the $\times$ basis to obtain $r_{i,0}',r_{i,1}'$.
    \end{enumerate}
    \item \textbf{Bob's output}: $m_0 \coloneqq \mathsf{Extract}(t_0 \coloneqq \{r_{i,d_i}\}_{i \in \overline{T}}), m_1 \coloneqq \mathsf{Extract}(t_1 \coloneqq \{r_{i,d_i \oplus 1}\}_{i \in \overline{T}})$.
\end{itemize}
\end{framed}
\caption{Non-interactive OT in the shared EPR pair model. $\mathsf{Extract}$ is an (unspecified) seedless hash function used for randomness extraction.}
\label{fig:non-interactive-OT}
\end{figure}

To prove security, we build on several recently developed quantum random oracle techniques \cite{C:Zhandry19,C:DFMS19,DFMS21} as well as techniques for analyzing ``quantum cut-and-choose'' protocols \cite{C:DFLSS09,C:BouFeh10}. In particular, we require the random oracle based commitments to be \emph{extractable}, and then argue that Bob's state on registers $\{\cB_{i,0},\cB_{i,1}\}_{i \in \overline{T}}$ is in some sense close to the state $\ket{\psi}$ described by the information $\{\theta_i,r_{i,0},r_{i,1}\}_{i \in \overline{T}}$ in Alice's unopened commitments. To do so, we use the Fiat-Shamir result of \cite{C:DFMS19,DFMS21} and the quantum sampling formalism of \cite{C:BouFeh10} to bound the trace distance between Bob's state and a state that is in a ``small'' superposition of vectors close to $\ket{\psi}$.

\paragraph{New Techniques for Randomness Extraction.} 
We also note that the arguments above have not yet established a fully secure OT correlation. In particular, Alice might have \emph{some} information about $t_{1-b}$, whereas OT security would require one of Bob's strings to be completely uniform and independent of Alice's view. 


This situation also arises in prior work on quantum OT, and is usually solved via \emph{seeded randomness extraction}. Using this approach, a seed $s$ would be sampled by Bob, and the final OT strings would be defined as $m_0 = \mathsf{Extract}(s,t_0)$ and $m_1 = \mathsf{Extract}(s,t_1)$, where $\mathsf{Extract}$ is a universal hash function. Indeed, quantum privacy amplication \cite{TCC:RenKon05} states that even given $s$, $\mathsf{Extract}(s,t_{1-b})$ is uniformly random from Alice's perspective as long as $t_{1-b}$ has sufficient (quantum) min-entropy conditioned on Alice's state.

Unfortunately, this approach would require Bob to transmit the seed $s$ to Alice in order for Alice to obtain her output $m_b = \mathsf{Extract}(s,t_b)$, making the protocol no longer non-interactive. Instead, we develop techniques for \emph{seedless} randomness extraction that work in our setting, allowing us to make the full description of the hash function used to derive the final OT strings \emph{public} at the beginning of the protocol. 

We provide two instantiations of seedless randomness extraction that work in a setting where the entropy source comes from measuring a state supported on a small superposition of basis vectors in the conjugate basis. More concretely, given a state on two registers $\cA,\cB$, where the state on $\cB$ is supported on standard basis vectors with small Hamming weight, consider measuring $\cB$ in the Hadamard basis to produce $x$. For what unseeded hash functions $\mathsf{Extract}$ does $\mathsf{Extract}(x)$ look uniformly random, even given the state on register $\cA$?

\begin{itemize}
    \item \textbf{XOR extractor.} First, we observe that one can obtain a \emph{single} bit of uniform randomness by XORing all of the bits of $x$ together, as long as the superposition on register $\cB$ only contains vectors with relative Hamming weight $< 1/2$. This can be used to obtain a \emph{bit} OT protocol, where the OT messages $m_0,m_1$ consist of a single bit. In fact, by adjusting the parameters of the quantum cut-and-choose, the XOR extractor could be used bit-by-bit to extract any number of $\secp$ bits. However, this setting of parameters would require a number of EPR pairs that grows with $\secp^3$, resulting in a very inefficient protocol.
    \item \textbf{RO extractor.} To obtain an efficient method of extracting $\secp$ bits, we turn to the random oracle model, which has proven to be a useful seedless extractor in the classical setting. Since an adversarial Alice in our protocol has some control over the state on registers $\cA,\cB$, arguing that $\RO(x)$ looks uniformly random from her perspective requires some notion of \emph{adaptive} re-programming in the QROM. While some adaptive re-programming theorems have been shown before (e.g. \cite{EC:Unruh15,10.1007/978-3-030-92062-3_22}), they have all {\em only considered $x$ sampled from a classical probability distribution}. This is for good reason, since counterexamples in the quantum setting exist, even when $x$ has high min-entropy given the state on register $\cA$.\footnote{For example, consider an adversary that, via a single superposition query to the random oracle, sets register $\cB$ to be a superposition over all $x$ such that the first bit of $\RO(x)$ is 0. Then, measuring $\cB$ in the computational basis will result in an $x$ with high min-entropy, but where $\RO(x)$ is distinguishable from a uniformly random $r$.} In this work, we show that in the special case of $x$ being sampled via measurement in a conjugate basis, one \emph{can} argue that $\RO(x)$ can be replaced with a uniformly random $r$, without detection by the adversary. Our proof relies on the superposition oracle of \cite{C:Zhandry19} and builds on proof techniques in \cite{10.1007/978-3-030-92062-3_22}. We leverage our RO extractor to obtain non-interactive $\secp$-bit string OT with a number of EPR pairs that only grows \emph{linearly} in $\secp$.
\end{itemize}

\paragraph{Differences from the CK88 template.} As mentioned earlier, the original quantum OT proposal \cite{FOCS:CreKil88} and its followups also incorporate a commit-challenge-response measurement-check protocol to enforce honest behavior. However, we point out one key difference in our approach that enables us to completely get rid of interaction. In CK88, each party measures their set of qubits\footnote{More accurately, since the protocol only uses BB84 states, one party prepares and the other party measures.} using a \emph{uniformly random} set of basis choices. Then, in order to set up the two channels required for OT, they need to exchange their basis choices with each other (after the measurement check commitments have been prepared and sent). This requires multiple rounds of interaction. In our setting, it is crucial that one of the parties measures (or prepares) qubits in a \emph{fixed} set of bases known to the other party, removing the need for a two-way exchange of basis information. In the case of \cref{fig:non-interactive-OT}, this party is Bob. Hereafter, we refer to the CK88 template as the \emph{random basis framework}, and our template as the \emph{fixed basis framework}. 

\paragraph{Non-interactive OT reversal.} So far, our techniques have shown that, given shared EPR pairs, Alice can send a single classical message to Bob that results in the following correlations: Alice outputs a bit $b$ and string $m_b$, while Bob outputs strings $m_0,m_1$, thus implementing the $\cF_{\SROT}$ functionality treating Bob as the ``sender''.

However, an arguably more natural functionality would treat Alice as the sender, with some chosen inputs $m_0,m_1$, and Bob as the receiver, who can recover $b,m_b$ from Alice's message. In fact, for the case that $m_0,m_1$ are single bits, a ``reversed'' version of the protocol can already be used to acheive this due to the  non-interactive OT reversal of \cite{C:IKNP03}. Let $(b,r_b)$ and $(r_0,r_1)$ be Alice and Bob's output from our protocol, where Alice has chosen $b$ uniformly at random. Then Alice can define $\ell_0 = m_0 \oplus r_b, \ell_1 = m_1 \oplus r_b \oplus b$ and send $(\ell_0,\ell_1)$ along with her message to Bob. Bob can then use $r_0$ to recover $m_c$ from $\ell_c$ for his ``choice bit'' $c = r_0 \oplus r_1$. Moreover, since in our protocol the bits $r_0,r_1$ can be sampled uniformly at random by the functionality, this implies that $c$ is a uniformly random choice bit, unknown to Alice, but unable to be tampered with by Bob. This results in a protocol that satisfies the $\cF_{\RROT[1]}$ functionality, and we have referred to it as our one-shot bit OT protocol in the introduction.

\subsection{Two-message OT without trusted setup} 
Next, say that we don't want to assume a {\em trusted} EPR pair setup. In particular, what if we allow Bob to set up the EPR pairs? In this case, a malicious Bob may send any state of his choice to Alice. However, observe that in \cref{fig:non-interactive-OT}, Alice's bit $b$ is masked by her random choices of $\theta_i$. These choices remain hidden from Bob due to the hiding of the commitment scheme, plus the fact that they are only used to measure Alice's registers. Regardless of the state that a malicious Bob may send, he will not be able to detect which basis Alice measures her registers in, and thus will not learn any information about $b$.
As a result, we obtain a \emph{two-message} quantum OT protocol in the QROM. As we show in \ifsubmission \cref{sec:two-round} \else \cref{subsec:two-round-OT} \fi, this protocol satisfies the $\cF_{\SROT}$ OT ideal functionality that allows Alice to choose her inputs $(b,m)$, and sends Bob random outputs $(m_0,m_1)$ subject to $m_b = m$.

Moreover, adding another reorientation message at the end from Bob to Alice -- where Bob uses $m_0, m_1$ as keys to encode his chosen inputs -- results in a three-round chosen input string OT protocol realizing the $\cF_{\OT[k]}$ functionality. However, as we will see in the next section, with three messages, we can {\em remove the need for entanglement} while still realizing $\cF_{\OT[k]}$.

Finally, in the case that $m_0,m_1$ are bits, we can apply the same non-interactive \cite{C:IKNP03} reversal described above to the two-round protocol, resulting in a two-round secure realization of the $\cF_{\RROT[1]}$ ideal functionality. This results in our two-round bit OT protocol as referenced in the introduction.



%

\subsection{Three-message chosen-input OT}
We now develop a three-message protocol that realizes the chosen-input string OT functionality $\cF_{\OT}$, which takes two strings $m_0,m_1$ from the sender and a bit $b$ from the receiver, and delivers $m_b$ to the receiver. This protocol will not require entanglement, but still uses the {\em fixed basis framework}, just like the one discussed in \cref{subsec:overniot}. 

Recall that in the EPR-based protocol, Bob would obtain $(r_0, r_1)$ by measuring his halves of two EPR pairs in basis $(+, \times)$,
while Alice would obtain $(r_0,r'_1)$ or $(r_0',r_1)$ respectively by measuring her halves in basis $(+,+)$ or $(\times,\times)$, where $(r_0', r_1')$ are uniform and independent of $(r_0, r_1)$.

Our first observation is that a similar effect is achieved by having Bob send BB84 states polarized in {\em a fixed basis} instead of sending EPR pairs. That is, Bob samples uniform $(r_0, r_1)$ and sends to Alice the states $\ket{r_0}_{+}, \ket{r_1}_{\times}$. 
Alice would obtain $(r_0,r_1')$ or $(r_0',r_1)$ respectively by measuring these states in basis $(+,+)$ or $(\times,\times)$ respectively, where $(r_0', r_1')$ are uniform and independent of $(r_0, r_1)$.
The skeleton protocol is sketched in Figure \ref{fig:basicotbb84}.

\begin{figure}[ht!]
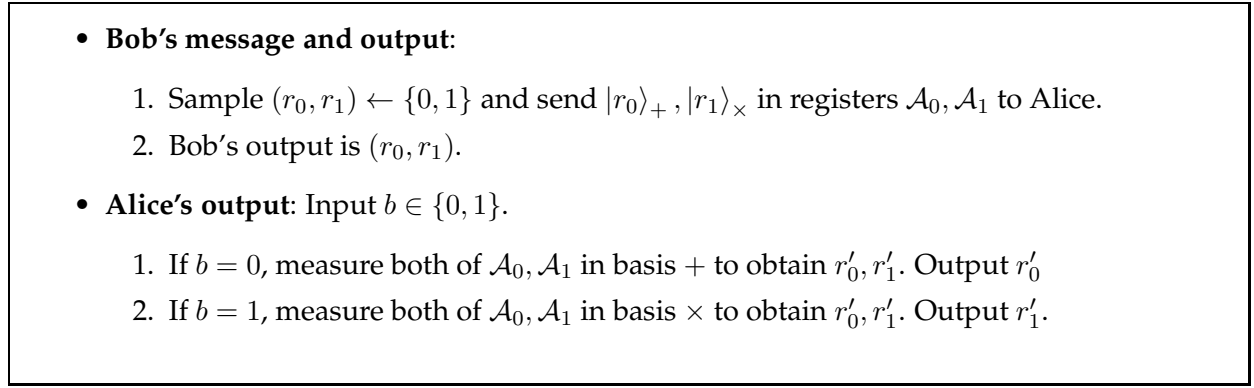

\begin{framed}
\begin{itemize}
    \item \textbf{Bob's message and output}: 
    \begin{enumerate}
    \item Sample $(r_0, r_1) \leftarrow \{0,1\}$ and send $\ket{r_0}_{+}, \ket{r_1}_{\times}$ in registers $\cA_0, \cA_1$ to Alice. 
    \item Bob's output is $(r_0, r_1)$.
    \end{enumerate}
    \item \textbf{Alice's output}: Input $b \in \{0,1\}$.
    \begin{enumerate}
        \item If $b = 0$, measure both of $\cA_{0},\cA_{1}$ in basis $+$ to obtain $r'_0,r'_1$. Output $r'_0$
        \item If $b = 1$, measure both of $\cA_{0},\cA_{1}$ in basis $\times$ to obtain $r'_0,r'_1$. Output $r'_1$.
    \end{enumerate}
\end{itemize}
\end{framed}
\caption{Another (insecure) skeleton OT candidate.}
\label{fig:basicotbb84}
\end{figure}

As before, though, there is nothing preventing Alice from retrieving both $(r_0, r_1)$ by measuring the states she obtains in basis $(+,\times)$. Thus, as before, we need a {\em measurement check} to ensure that Alice measures ``most'' out of a {\em set} of pairs of qubits in the same basis. But implementing such a check with BB84 states turns out to be more involved than in the EPR pair protocol. 

\paragraph{Non-interactive measurement check without entanglement.}
Towards building a measurement check, we first modify the skeleton protocol so that Bob sends $2n$ BB84 qubits $\{\ket{r_{i,0}}_{+}, \ket{r_{i,1}}_{\times}\}_{i \in [n]}$ on registers $\{\cA_{i,b}\}_{i \in [n], b \in \{0,1\}}$ to Alice (instead of just two qubits).

Now Alice is required to sample a set of $n$ bases $\theta_1,\dots,\theta_n \gets \{+,\times\}^n$. For each $i \in [n]$, she must measure the $i^{th}$ pair of qubits in basis $\theta_i$, obtaining measurement outcomes $(r_{i,0}',r_{i,1}')$. Then, she will commit to her bases and outcomes $\com(\theta_1,r_{1,0}',r_{1,1}'),\dots,\com(\theta_n,r_{n,0}',r_{n,1}')$. 
Once committed, she will {\em open} commitments corresponding to a randomly chosen (by Bob) $T \subset [n]$ of size $k$, revealing $\{\theta_i,r_{i,0}',r_{i,1}'\}_{i \in T}$. 

But Bob cannot check these openings the same way as in the EPR-based protocol. Recall that in the EPR protocol, for every $i \in T$, Bob would measure his halves of EPR pairs in bases $(\theta_i,\theta_i)$ to obtain $(r_{i,0},r_{i,1})$, and compare the results against Alice's response.
On the other hand, once Bob has sent registers $\{\cA_{i,b}\}_{i \in [n], b \in \{0,1\}}$ containing $\{\ket{r_{i,0}}_{+}, \ket{r_{i,1}}_{\times}\}_{i \in [n]}$ to Alice, there is no way for him to recover the result of measuring any pair of registers $(\cA_{i,0}, \cA_{i,1})$ in basis $(\theta_i, \theta_i)$.

To fix this, we modify the protocol to allow for a (randomly chosen and hidden) set $U$ of ``trap'' positions. For all $i \in U$, Bob outputs registers $(\cA_{i, 0}, \cA_{i,1})$ containing $\ket{r_{i,0}}_{\vartheta_i}, \ket{r_{i,1}}_{\vartheta_i}$, that is, both qubits are polarized in the same basis $\vartheta_i \leftarrow \{+,\times\}$. 
All other qubits are sampled the same way as before, i.e. as $\ket{r_{i,0}}_{+}, \ket{r_{i,1}}_{\times}$.
Alice commits to her measurement outcomes $\{\theta_i,r_{i,0}',r_{i,1}'\}_{i \in [n]}$, and then reveals commitment openings $\{\theta_i,r_{i,0}',r_{i,1}'\}_{i \in T}$ for a randomly chosen subset of size $T$, as before.
But Bob can now check Alice on all positions $i$ in the intersection $T \cap U$ where $\vartheta_i = \theta_i$.
Specifically, Bob aborts if for any $i \in T \cap U$, $\vartheta_i = \theta_i$ but $(r'_{i,0},r'_{i,1}) \neq (r_{i,0},r_{i,1})$.
Otherwise, Alice and Bob will use the set $[n] \setminus T \setminus U$ to generate their OT outputs.
The resulting protocol is sketched in Figure \ref{fig:3rOTover}. Crucially, we make use of a third round in order to allow Bob to transmit his choice of $U$ to Alice, so that they can both agree on the set $[n] \setminus T \setminus U$.

Again, we must argue that any Alice that passes Bob's check with noticeable probability loses information about one out of $r_{i,0}$ and $r_{i,1}$ for ``most'' $i \in [n] \setminus T \setminus U$.
Because she did not know the check subset $T$ or Bob's trap subset $U$ at the time of committing to her measurement outcomes, we can again conjecture that for ``most'' $i \in [n]\setminus T$, Alice also correctly committed to results of measuring her qubits in bases $(\theta_i, \theta_i)$. Moreover we can conjecture that the act of committing and passing Bob's check removed from Alice's view information about at least one out of $(r_{i,0}, r_{i,1})$ for most $i \in [n] \setminus T$. 
This requires carefully formulating and analyzing a quantum sampling strategy that is somewhat more involved than the one in \cref{subsec:overniot}. 
Furthermore, as in \cref{subsec:overniot}, we make the measurement check non-interactive by relying on the Fiat-Shamir transform.
A formal analysis of this protocol can be found in Section \ref{sec:bb84-3r-ot}.


\begin{figure}[ht!]
\begin{framed}
\begin{itemize}
    \item \textbf{Inputs:} 
    Bob has inputs $m_0, m_1$ each in $\{0,1\}^\lambda$, Alice has input $b \in \{0,1\}$. 
    \item \textbf{Bob's Message}: 
    \begin{enumerate}
    \item Sample a ``large enough'' subset $U \subset [n]$, and for every $i \in U$, sample $\vartheta_i \leftarrow \{+,\times\}$.
    \item For every $i \in [n]$, sample $(r_{i,0}, r_{i,1}) \leftarrow \{0,1\}$.
    \item For $i \in U$, set registers $(\cA_{i,0}, \cA_{i,1})$ to $(\ket{r_{i,0}}_{\vartheta_i}, \ket{r_{i,1}}_{\vartheta_i})$. 
    \item For $i \in [n] \setminus U$, set registers $(\cA_{i,0}, \cA_{i,1})$ to $(\ket{r_{i,0}}_{+}, \ket{r_{i,0}}_{\times})$. 
    \item Send $\{\cA_{i,0}, \cA_{i,1}\}_{i \in [n]}$ to Alice.
    \end{enumerate}
    \item \textbf{Alice's message}: 
    \begin{enumerate}
        \item Sample $\theta_1,\dots,\theta_n \gets \{+,\times\}^n$ and measure each $\cA_{i,0},\cA_{i,1}$ in basis $\theta_i$ to obtain $r'_{i,0},r'_{i,1}$.
        \item Compute commitments $\com_1,\dots,\com_n$ to $(\theta_1,r'_{1,0},r'_{1,1}),\dots,(\theta_n,r'_{n,0},r'_{n,1})$.
        \item Compute $T = \RO(\com_1,\dots,\com_n)$, where $T$ is parsed as a subset of $[n]$ of size $k$.
        \item Compute openings $\{u_i\}_{i \in T}$ for $\{\com_i\}_{i \in T}$.
        \item Let $\overline{T} = [n] \setminus T$, and for all $i \in \overline{T}$, set $d_i = b \oplus \theta_i$ (interpreting $+$ as 0 and $\times$ as 1).
        \item Send $\{\com_i\}_{i \in [n]}, T, \{r'_{i,0},r'_{i,1},\theta_i,u_i\}_{i \in T}, \{d_i\}_{i \in \overline{T}}$ to Bob.
    \end{enumerate}
    \item \textbf{Bob's Message}: 
    \begin{enumerate}
        \item Abort if $T \neq \RO(\com_1,\dots,\com_n)$ or if verifying any commitment in the set $T$ fails.
        \item If for any $i \in T \cap U$,
        $r_{i,0} \neq r_{i,0}'$ or $r_{i,1} \neq r_{i,1}'$, abort.
        \item Set $x_0 = m_0 \oplus \mathsf{Extract}(t_0 \coloneqq \{r_{i,d_i}\}_{i \in [n]\setminus T \setminus U})$ and $x_1 = m_1 \oplus \mathsf{Extract}(t_1 \coloneqq \{r_{i,d_i \oplus 1}\}_{i \in [n]\setminus T \setminus U})$.
        \item Send $(x_0, x_1, U)$ to Alice.
    \end{enumerate}
    \item \textbf{Alice's output}: $m_b \coloneqq x_b \oplus \mathsf{Extract}(t_b \coloneqq \{r'_{i,\theta_i}\}_{i \in \overline{T}})$.
\end{itemize}
\end{framed}
\caption{Three-message chosen-input OT without entanglement. $\mathsf{Extract}$ is an (unspecified) function used for randomness extraction. Since Bob is sending the final message, we may use a seeded function here.}
\label{fig:3rOTover}
\end{figure}



\subsection{The random basis framework}


Next, we shift our attention to analyzing the original template for commitment-based quantum OT, due to \cite{FOCS:CreKil88}, and studied in many followups including \cite{C:BBCS91,MayersSalvail94,STOC:Yao95,C:DFLSS09,C:BouFeh10,EC:Unruh10,GLSV,BCKM2021}. In this template, one party (say, Bob) prepares random BB84 states and sends them to Alice, who is then supposed to immediately measure each received state in a random basis. That is, each party samples their own uniformly random sequence of bases $\theta_A = \theta_{A,1},\dots,\theta_{A,n}, \theta_B = \theta_{B,1},\dots,\theta_{B,n}$ during the protocol, and thus we refer to this template as the ``random basis framework''. After this initial prepare-and-measure step, Alice then convinces Bob via a cut-and-choose measurement check that she indeed measured her states, thus simulating a type of erasure channel. The rest of the protocol can be viewed as a conversion from the resulting erasure channel to OT.

First, we observe that, given a non-interactive commitment for use in the measurement check, this protocol can naturally be written as a five-message OT between a receiver Alice and a sender Bob as follows.

\begin{enumerate}
    \item Bob samples and sends random BB84 states to Alice, where $\theta_B$ are the bases and $r_B$ are the bits encoded.
    \item Alice measures the received states in bases $\theta_A$, commits to $\theta_A$ and the measurement results, and sends the commitments to Bob.
    \item Bob samples a random subset $T$ of the commitments to ask Alice to open, and sends $T$ and $\theta_B$ to Alice.
    \item Alice computes openings to the commitments in $T$, and then encodes her choice bit $b$ as follows: set $S_b = \{i \in \overline{T}: \theta_{A,i} = \theta_{B,i}\}$ and set $S_{1-b} = \{i \in \overline{T}: \theta_{A,i} \neq \theta_{B,i}\}$. She sends her openings and $(S_0,S_1)$ to Bob.
    \item Bob checks that the commitment openings verify and that Alice was honestly measuring her qubits in $T$. If so, Bob encrypts $m_0$ using $\{r_{B,i}\}_{i \in S_0}$, encrypts $m_1$ using $\{r_{B,i}\}_{i \in S_1}$, and sends the two encryptions to Alice.
\end{enumerate}

Now, a natural question is whether we can reduce interaction in the QROM via a non-interactive measurment check, as accomplished above in the fixed basis framework. Unfortunately, the structure of the random basis framework appears to prevent this optimiziation. Indeed, Alice cannot encode her choice bit until {\em after} she receives $\theta_B$ from Bob, which he cannot send until after he receives Alice's commitments. 

However, while these reasons prevent us from obtaining a one or two message protocol as in the fixed basis framework, we do show a different optimiziation that allows us to obtain a four-message chosen-input OT and a three-message random-input OT utilizing this framework, which we discuss next. 

\paragraph{Reverse Crepeau-Kilian OT.} 
Suppose instead that \emph{Alice} sends random BB84 states $\{\ket{r_{A,i}}_{\theta_{A,i}}\}_{i \in [n]}$, after which Bob measures these states in random bases $\theta_B$ to obtain $\{r_{B,i}\}_{i \in [n]}$. Now, instead of waiting to obtain the ``correct'' bases $\theta_A$, Bob simply sends $\theta_B$ to Alice. When $\theta_{A,i}$ and $\theta_{B,i}$ match, $r_{A,i} = r_{B,i}$, and when $\theta_{A,i}$ and $\theta_{B,i}$ do not match, then $r_{A,i}$ and $r_{B,i}$ should be uncorrelated: again establishing an erasure channel on which Bob can send Alice messages. However, unlike CK88, the player that is performing measurements in random bases {\em need not wait to learn the right bases}, and instead simply announces his own bases to set up a reverse erasure channel.

However, this protocol leads to new avenues of attack for a malicious Alice. In particular, Alice may send halves of EPR pairs in the first round, and, given $\theta_B$, perform measurements to determine all the $r_{B,i}$ values. Such an attack can be prevented by means of a ``reverse'' measurement check: namely, Alice {\em commits to all $r_{A,i}$ and $\theta_{A,i}$ values} in the first message (she commits to the descriptions of her states), and, given a random check set $T$ chosen by Bob, reveals all committed values $\{r_{A,i}, \theta_{A,i}\}_{i \in T}$. Given Alice's openings, for every $i \in [T]$ such that $\theta_{A,i} = \theta_{B,i}$ Bob checks that $r_{A,i} = r_{B,i}$. The resulting four-round chosen-input OT protocol is summarized in Figure \ref{fig:3-round}. We also note that, using our seedless extaction techniques described above, this template can be used to obtain \emph{three-message} protocols for $\cF_{\SROT}$ and $\cF_{\RROT}$.

Finally, we note that it is unclear how to apply Fiat-Shamir to this reversed protocol in order to reduce interaction even further. Indeed, in this case it seems the Fiat-Shamir hash function would also have to take as input Alice's \emph{quantum states}, since otherwise she could determine these states after observing the result of the hash.  

\paragraph{The ideal commitment model.} We observe that the protocols that we obtain in the random basis framework (if we used seeded extraction or the XOR extractor) actually do not use the random oracle beyond its usage in building the commitment scheme. Thus, these protocols could be seen as being constructed in an ``ideal commitment model'', which is motivated by prior work \cite{C:DFLSS09,GLSV,BCKM2021} that established commitments as the only necessary cryptographic building block for quantum OT. It may be interesting to explore these protocols combined with other (say, plain model or CRS model) instantiations of the required commitments.

\begin{figure}[ht!]
      \begin{framed}

\begin{itemize}
\item {\bf Inputs:} Bob has inputs $m_0,m_1 \in \{0,1\}^\secp$, Alice has input $b \in \{0,1\}$.
\item {\bf Alice's first message:}
\begin{enumerate}
    \item For every $i \in [n]$, sample $r_{A,i} \gets \{0,1\}$ and $\theta_{A,i} \gets \{+,\times\}$, and prepare the state $\ket{r_{A,i}}_{\theta_{A,i}}$ on register $\cA_i$.
    \item Compute commitments $\com_1,\dots,\com_n$ to $(\theta_{A,i},r_{A,i}),\dots,(\theta_{A,n},r_{A,n})$.
    \item Send $\{\cA_{i}\}_{i \in [n]}$ and $\{\com_i\}_{i \in [n]}$ to Bob.
\end{enumerate}

\item {\bf Bob's first message:} 
\begin{enumerate}
    \item Sample $\theta_B = \theta_{B,1},\dots,\theta_{B,n} \gets \{+,\times\}^n$ and measure each $\cA_{i}$ in basis $\theta_{B,i}$ to obtain $r_{B,i}$.
    \item Sample a ``large enough'' subset $T \subset [n]$.
    \item Send $T$ and $\theta_B$ to Alice.
\end{enumerate}

\item {\bf Alice's second message:}
\begin{enumerate}
    \item Compute openings $\{u_i\}_{i \in T}$ for $\{\com_i\}_{i \in T}$.
    \item Set $S_b = \{i \in \overline{T} : \theta_{A,i} = \theta_{B,i}\}$ and $S_{1-b} = \{i \in \overline{T} : \theta_{A,i} \neq \theta_{B,i}\}$.
    \item Send $\{u_i\}_{i \in T},S_0,S_1$ to Bob.
\end{enumerate}

\item {\bf Bob's second message:}
\begin{enumerate}
    \item Check that the openings to the commitments in $T$ verify, and that for each $i \in T$ such that $\theta_{A,i} = \theta_{B,i}$, it holds that $r_{A,i} = r_{B,i}$.
    \item Set $x_0 = m_0 \oplus \mathsf{Extract}(\{r_{B,i}\}_{i \in S_0})$ and $x_1 = m_1 \oplus \mathsf{Extract}(\{r_{B,i}\}_{i \in S_1})$.
    \item Send $(x_0,x_1)$ to Alice.
\end{enumerate}

\item {\bf Alice's output:} $m_b \coloneqq x_b \oplus \mathsf{Extract}(\{r_{A,i}\}_{i \in S_b})$.

\end{itemize}
      \end{framed}
    \caption{Four-message chosen-input OT from commitments. $\mathsf{Extract}$ is an (unspecified) function used for randomness extraction. Since Bob is sending the final message, we may use a seeded function here.}
    \label{fig:3-round}
\end{figure}

\subsection{Extractable and Equivocal Commitments}
To achieve simulation-based security, our constructions rely on commitments that satisfy {\em extractability and equivocality}. 
We model these as classical non-interactive bit commitments that, informally, satisfy the following properties.
\begin{itemize}
\item Equivocality: This property ensures that the commitment scheme admits an efficient simulator, let's say $\cS_\Equ$, that can sample commitment strings that are indistinguishable from commitment strings generated honestly and later, during the opening phase, provide valid openings for either $0$ or $1$.
\item Extractability: This property ensures that the commitment scheme admits an efficient extractor, let's say $\cS_\Ext$, that, given access to the committer who outputs a commitment string, can output the committed bit.
\end{itemize}

The need for these two additional properties is not new to our work. Indeed, \cite{C:DFLSS09} showed that bit commitment schemes satisfying extraction and equivocation suffice to instantiate the original \cite{FOCS:CreKil88,C:BBCS91} QOT template. \cite{C:DFLSS09} called their commitments dual-mode commitments, and provided a construction based on the quantum hardness of the learning with errors (QLWE) assumption. In two recent works \cite{BCKM2021,GLSV}, constructions of such commitment schemes were achieved by relying on just post-quantum one-way functions (in addition to quantum communication). 

We show that the most common construction of random-oracle based commitments -- where a commitment to bit $b$ is $H(b||r)$ for uniform $r$ -- satisfies both extractability and equivocality in the QROM. Our proof of extractability applies the techniques of~\cite{C:Zhandry19,DFMS21} for on-the-fly simulation with extraction, and our proof of equivocality relies on a one-way-to-hiding lemma from~\cite{C:AmbHamUnr19}.



\subsection{Concrete parameters}\label{subsec:concrete-parameters}
Beyond proving the our protocols have negligible security error, we also compute both concrete bounds for the number of quantum resources required by our protocols (as a function of the security parameter), and derive exact security losses incurred by our protocols. This involves careful analyses of the cut-and-choose strategies underlying the measurement-check parts of our protocols. Such strategies were generically analyzed in~\cite{C:BouFeh10}, and we strengthen their classical analyses to obtain improved parameters for our quantum sampling games.

We summarize our parameters in Table \ref{tab:concretenums} below, where we discuss the number of EPR pairs/BB84 states required by each of our fixed-basis protocols in the first two columns, and in our optimization of random-basis protocols in the last two columns. We also compute concrete bounds that we obtain when relying on the XOR extractor (to obtain bit OT) versus when relying on the random oracle or seeded extractors (to obtain string OT). 
\begin{table}[H]
\centering
\resizebox{0.9\textwidth}{!}{%
\begin{tabular}{|c|cc|cc|}
\hline
  & \multicolumn{2}{c|}{Fixed Basis Framework} & \multicolumn{2}{c|}{Random Basis Framework} \\ \hline
 &
  \multicolumn{1}{c|}{\begin{tabular}[c]{@{}c@{}}1 round $\cF_{\SROT}$\\ (EPR pairs)\end{tabular}} &
  \begin{tabular}[c]{@{}c@{}}3 round $\cF_{\OT}$\\ (BB84 states)\end{tabular} &
  \multicolumn{1}{c|}{\begin{tabular}[c]{@{}c@{}}3 round $\cF_{\SROT}$\\ (BB84 states)\end{tabular}} &
  \begin{tabular}[c]{@{}c@{}}4 round $\cF_{\OT}$\\ (BB84 states)\end{tabular} \\ \hline
\begin{tabular}[c]{@{}c@{}}Bit OT \\ (XOR extractor)\end{tabular} & \multicolumn{1}{c|}{$300\secp$}     & $3200\secp$    & \multicolumn{1}{c|}{$1600\secp$}     & $1600\secp$   \\ \hline
\begin{tabular}[c]{@{}c@{}}String OT \\ (RO/seeded extractor)\end{tabular} &
  \multicolumn{1}{c|}{\begin{tabular}[c]{@{}c@{}}$6420\secp$\\ (RO)\end{tabular}} &
  \begin{tabular}[c]{@{}c@{}}$84\, 200\secp$\\ (seeded)\end{tabular} &
  \multicolumn{1}{c|}{\begin{tabular}[c]{@{}c@{}}$23\, 000\secp$\\ (RO)\end{tabular}} &
  \begin{tabular}[c]{@{}c@{}}$10\, 300\secp$\\ (seeded)\end{tabular} \\ \hline
\end{tabular}%
}
\caption{A summary of quantum resources required for our protocols. $\lambda$ denotes the security parameter. All of our protocols have security losses bounded by $\frac{O(q^{3/2}\lambda)}{2^\lambda}$, where $q$ is the number of queries made by the adversary to the random oracle. We refer the reader to the following sections for additional details and concrete bounds: (\cref{sec:epr-ot}, \ifsubmission \cref{sec:two-round} \else \cref{subsec:two-round-OT}\fi) for the fixed basis EPR pair protocols, (\cref{sec:bb84-3r-ot} and \cref{sec:fixed-basis-3r-xor}) for the fixed basis BB84 state protocols, and \cref{sec:ot-eecom} for the random basis protocols.}


\label{tab:concretenums}
\end{table}

\ifsubmission
\paragraph{Roadmap.}
We discuss some preliminaries and notational conventions in \cref{sec:prelims}. In the next section, we discuss our seedless extractors, with full proofs deferred to \cref{sec:adaptive-proof}. Next, we discuss definitions of extractable and equivocal commitments in \cref{sec:uccomm}. Then, we discuss our fixed-basis constructions in \cref{sec:epr-ot} and defer analyses to \cref{sec:non-interactive-security}, \cref{sec:two-round} and \cref{sec:bb84-3r-ot}.
Finally, in \cref{sec:ot-eecom}, we analyze a variant of the CK88 protocol, yielding chosen-input string OT in four and random-input OT in three rounds. A full version of this paper is also attached as supplementary material.
\else \fi





\ifsubmission\else\section{Preliminaries}
\label{sec:prelims}
We use $[n]$ to denote the set $\{1,2, \ldots n\}$ and $[a,b]$ (where $a<b$) to denote the set $\{a, a+1, \ldots b\}$. We use $\cH\cW(x)$ to denote the Hamming weight of a binary string $x \in \{0,1\}^*$, and $\omega(x)$ to denote its \emph{relative} Hamming weight $\cH\cW(x) / |x|$. For two strings $x,y \in \zo^*$, we use $\Delta(x,y)=\omega(x \oplus y)$ to denote the relative Hamming distance of $x,y$. For finite sets $X,Y$, let $F_{X \to Y}$ be the set of functions with domain $X$ and codomain $Y$. For a set $T \subseteq [n]$, $\{i\}_{i \in T}$ is used to represent a set indexed by $T$.
Let $h_b(x)$ denote the binary entropy function, $h_b(x) = -x\log_2(x)-(1-x)\log_2(1-x)$. We make use of the well-known fact that the number of strings of length $n$ with relative Hamming weight at most $\delta$ is $\leq 2^{h_b(\delta)n}$.

\subsection{Quantum preliminaries}

A register $\cX$ is a named Hilbert space $\bbC^{2^n}$. A pure quantum state on register $\cX$ is a unit vector $\ket{\psi}_\cX \in \bbC^{2^n}$, and we say that $\ket{\psi}_\cX$ consists of $n$ qubits. A mixed state on register $\cX$ is described by a density matrix $\rho_\cX \in \bbC^{2^n \times 2^n}$, which is a positive semi-definite Hermitian operator with trace 1. 

A quantum operation $F$ is a completely-positive trace-preserving (CPTP) map from a register $\cX$ to a register $\cY$, which in general may have different dimensions. That is, on input a density matrix $\rho_\cX$, the operation $F$ produces $F(\rho_\cX) = \tau_\cY$ a mixed state on register $\cY$. A \emph{unitary} $U: \cX \to \cX$ is a special case of a quantum operation that satisfies $U^\dagger U = U U^\dagger = \bbI_\cX$, where $\bbI_\cX$ is the identity matrix on register $\cX$. A \emph{projector} $\Pi$ is a Hermitian operator such that $\Pi^2 = \Pi$, and a \emph{projective measurement} is a collection of projectors $\{\Pi_i\}_i$ such that $\sum_i \Pi_i = \bbI$.

Let $\ptrace$ denote the trace operator. For registers $\cX,\cY$, the \emph{partial trace} $\ptrace_\cY$ is the unique operation from $\cX,\cY$ to $\cX$ such that for all $\rho_\cX,\tau_\cY$, $\ptrace_{\cY}(\rho,\tau) = \ptrace(\tau)\rho$. The \emph{trace distance} between states $\rho,\tau$, denoted $\TD(\rho,\tau)$ is defined as \[\TD(\rho,\tau) \coloneqq \frac{1}{2}\|\rho-\tau\|_1 \coloneqq \frac{1}{2}\ptrace\left(\sqrt{(\rho-\tau)^\dagger(\rho-\tau)}\right).\] We will often use the fact that the trace distance between two states $\rho$ and $\tau$ is an upper bound on the probability that any algorithm can distinguish $\rho$ and $\tau$.

\begin{lemma}[Gentle measurement \cite{DBLP:journals/tit/Winter99}]\label{lemma:gentle-measurement}
Let $\rho_\cX$ be a quantum state and let $(\Pi,\bbI-\Pi)$ be a projective measurement on $\cX$ such that $\Tr(\Pi\rho) \geq 1-\delta$. Let \[\rho' = \frac{\Pi\rho\Pi}{\Tr(\Pi\rho)}\] be the state after applying $(\Pi,\bbI-\Pi)$ to $\rho$ and post-selecting on obtaining the first outcome. Then, $\TD(\rho,\rho') \leq 2\sqrt{\delta}$.
\end{lemma}

Finally, we will make use of the convention that $+$ denotes the computational basis $\{\ket{0},\ket{1}\}$ and $\times$ denotes the Hadamard basis $\left\{\frac{\ket{0} + \ket{1}}{\sqrt{2}}, \frac{\ket{0} - \ket{1}}{\sqrt{2}}\right\}$. For a bit $r \in \{0,1\}$, we write $\ket{r}_+$ to denote $r$ encoded in the computational basis, and $\ket{r}_\times$ to denote $r$ encoded in the Hadamard basis.

\subsection{Quantum machines and protocols}
\label{subsec:quantum-protocols}

\fflater{\nishant{todo: add support for parallel queries}}

\paragraph{Quantum interactive machines.} A quantum interactive machine (QIM) is a family of machines $\{M_\secp\}_{\secp \in \bbN}$, where each $M_{\secp}$ consists of a sequence of quantum operations $M_{\secp,1},\dots,M_{\secp,\ell(\secp)}$, where $\ell(\secp)$ is the number of rounds in which $M_\secp$ operates. Usually, we drop the indexing by $\secp$ and refer to the machine $M = M_1,\dots,M_\ell$. Each machine $M_i$ may have a designated input and output register used to communicate with its environment. 


\paragraph{Quantum oracle machines.} Let $X,Y$ be finite sets and let $O : X \to Y$ be an arbitrary function. We say that $A^{O}$ is a $q$-query \emph{quantum oracle machine} (QOM) if it can be written as $A_{q+1}U_OA_{q}U_O\dots U_OA_2U_OA_1$, where $A_1,\dots,A_{q+1}$ are arbitrary quantum operations, and $U[O]$ is the unitary defined by \[U[O]: \ket{x}_\cX\ket{y}_\cY \to \ket{x}_\cX\ket{y \oplus O(x)}_\cY,\] operating on a designated oracle input register $\cX$ and oracle output register $\cY$. We say that $A$ is a \emph{quantum interactive oracle machine} (QIOM) if $A = A_1^O,\dots,A_\ell^O$ is such that each $A_i^O$ is a quantum oracle machine.

Sometimes, it will be convenient to consider \emph{controlled} queries to an oracle $O$, which would be implemented by a unitary \[U_c[O] : \ket{b}_\cB\ket{x}_\cX\ket{y}_\cY \to \ket{b}_\cB\ket{x}_\cX\ket{y \oplus b \cdot O(x)}_\cY.\] 

However, it is easy to see that such a controlled query can be implemented with two standard queries, by introducing an extra register $\cZ$, as follows:

\begin{align*}
    \ket{b}_\cB\ket{x}_{\cX}\ket{y}_\cY\ket{0}_\cZ 
&\xrightarrow{U[O_0]_{\cX,\cZ}} \ket{b}_\cB\ket{x}_{\cX}\ket{y}_\cY\ket{O(x)}_\cZ \to \ket{b}_\cB\ket{x}_{\cX}\ket{y \oplus b \cdot O(x)}_\cY\ket{O(x)}_\cZ\\  &\xrightarrow{U[O_0]_{\cX,\cZ}} \ket{b}_\cB\ket{x}_{\cX}\ket{y \oplus b \cdot O(x)}_\cY\ket{0}_\cZ.
\end{align*}

It will also be convenient to consider algorithms $A^{O_0,O_1}$ with access to \emph{multiple} oracles $O_0: X_0 \to Y_0, O_1: X_1 \to Y_1$, written as $A_{q+1}U_{O_1}A_{q}U_{O_0}\dots U_{O_1}A_2U_{O_0}A_1$. Defining $O(b,x) = O_b(x)$ for $(b,x) \in (0,X_0) \union (1,X_1)$, it is easy to see that any $A^{O_0,O_1}$ can be written as an oracle algorithm $B^O$. On the other hand, given a $q$-query oracle algorithm $A^O$ where $O: X \to Y$, and a partition of $X$ into $(0,X') \cup (1,X')$, we can write $A$ as a $4q$-query algorithm $B^{O_0,O_1}$, where $O_b : X' \to Y$ is such that $O_b(x') = O(b,x')$. This follows by answering each query to $O$ using one controlled query to $O_0$ and one controlled query to $O_1$. This can be extended to splitting up an oracle $O$ into $k$ oracles, with a multiplicative factor of $2k$ in the number of queries made by the adversary. Thus, throughout this work, we often consider machines that have access to multiple (potentially independently sampled) oracles, while noting that this model is equivalent to considering machines with access to a single oracle, up to a difference in the number of oracle queries. In particular, any adversarial algorithm that has superposition access to a single oracle $O$ with an input space that can be partitioned into $k$ parts may be written as an adversarial algorithm wth access to $k$ appropriately defined separate oracles $O_1,\dots,O_k$.




\paragraph{Functionalities and protocols in the quantum random oracle model.} \label{functionalities-and-protocols-in-qrom}
Let $\cF$ denote a \emph{functionality}, which is a classical interactive machine specifying the instructions to realize a cryptographic task. A two-party protocol\footnote{One can also consider multi-party protocols, but we restrict to the two-party setting in this work.} $\Pi$ for $\cF$ consists of two QIMs $(A,B)$.\footnote{Technically, $A$ and $B$ are infinite families of interactive machines, parameterized by the security parameter $\secp$.} A protocol $\Pi$ in the \emph{quantum random oracle model} (QROM) consists of two \emph{QIOMs} $(A^H,B^H)$ that have quantum oracle access to a uniformly random function $H$ sampled from $F_{X \to Y}$ for some finite sets $X$ and $Y$. 

An adversary intending to attack the protocol along with a distinguisher can be described by a family $\{\Adv_\secp,\sD_\secp,x_\secp\}_{\secp \in \bbN}$, where $\Adv_\secp$ is a QIOM that corrupts party $M \in \{A,B\}$, $\sD_\secp$ is a QOM, and $x_\secp$ is the input of the honest party $P \in \{A,B\}$. Define the one-bit random variable $\Pi[\Adv_\secp,\sD_\secp,x_\secp]$ as follows.


\begin{itemize}
    \item $H$ is sampled uniformly at random.
    \item $\Adv_\secp^H$ interacts with $P^H(x_\secp)$ during the execution of $\Pi$, and $\Adv_\secp$ outputs a quantum state $\rho$, while $P^H$ outputs a classical string $y$.
    \item $\sD_\secp^H(\rho, y)$ outputs a bit $b$.
\end{itemize}


An \emph{ideal-world} protocol $\widetilde{\Pi}_\cF$ for functionality $\cF$ consists of two ``dummy'' parties $\widetilde{A},\widetilde{B}$ that have access to an additional ``trusted'' party that implements $\cF$. That is, $\widetilde{A},\widetilde{B}$ each interact directly with $\cF$, which eventually returns outputs to $\widetilde{A},\widetilde{B}$. We consider the execution of ideal-world protocols in the presence of a simulator followed by a distinguisher, described by a family $\{\Sim_\secp,\sD_\secp,x_\secp\}_{\secp \in \bbN}$. Define the random variable $\widetilde{\Pi}_\cF[\Sim_\secp,\sD_\secp,x_\secp]$ over one bit output a follows. 


\begin{itemize}
    \item $\Sim_\secp$ interacts with $\widetilde{P}(x_\secp)$ during the execution of $\widetilde{\Pi}_\cF$, and $\Sim_\secp$ outputs a quantum state $\rho$, while $\widetilde{P}$ outputs a classical string $y$.
    \item $\sD_\secp^{\Sim_\secp}(\rho, y)$ outputs a bit $b$.
\end{itemize}

In the above, $\Sim_\secp$ may be \emph{stateful}, meaning that the part of $\Sim_\secp$ that interacts with $\widetilde{P}(x)$ may pass an arbitrary state to the part of $\Sim_\secp$ that answers $\sD_\secp's$ oracle queries.

Furthermore, we will only consider the notion of security with abort where every ideal functionality is slightly modified to (1) know the identities of corrupted parties and (2) be slightly reactive: after all parties have provided input, the functionality computes outputs and delivers the outputs to the corrupt parties only. Then the functionality awaits either a ``deliver'' or ``abort'' command from the corrupted parties. Upon receiving ``deliver'', the functionality delivers the outputs to all the honest parties. Upon receiving ``abort'', the functionality delivers an abort output ($\bot$) to all the honest parties.
\fflater{\nishant{removed the $(\cdot,\cdot)$ notation from below. Also, i think we can strengthen this definition by saying that there exists a simulator, polynomial s which works for all adv making q queries. Esp since now we don't have the query bound needed for simulation on adversaries? And making the defn similar to what is there for commitments currently}}

\begin{definition}[Securely Realizing Functionalities with Abort]\label{def:secure-realization}
A protocol $\Pi$ $\mu$-\emph{securely realizes} a functionality $\cF$ with abort in the quantum random oracle model if there exists a polynomial $s$ such that for any function $q$ and any $\{\Adv_\secp\}_{\secp \in \bbN}$, there exists a simulator $\{\Sim_\secp\}_{\secp \in \bbN}$ such that the run-time of $\Sim_\secp$ is at most the run-time of $\Adv_\secp$ plus $s(\secp,q(\secp))$, and for all $\{\sD_\secp,x_\secp\}_{\secp \in \bbN}$ with the property that the combined number of oracle queries made by $\Adv_\secp$ and $\sD_\secp$ is at most $q(\secp)$, it holds that \[\bigg|\Pr[\Pi[\Adv_\secp,\sD_\secp,x_\secp] = 1] - \Pr[\widetilde{\Pi}_\cF[\Sim_\secp,\sD_\secp,x_\secp] = 1]\bigg| = \mu(\secp,q(\secp)).\]
Furthermore, we say that a protocol $\Pi$ \emph{securely realizes} a functionality $\cF$ if it $\mu$-securely realizes $\cF$ where $\mu$ is such that for any $q(\secp) = \poly(\secp)$, $\mu(\secp,q(\secp)) = \negl(\secp)$.
\end{definition}


\subsection{Oblivious transfer functionalities} We will consider various oblivious transfer functionalities in this work. Some of these will be used as stepping stones towards other constructions. 

\begin{itemize}
    \item $\cF_{\OT[k]}$: the \emph{chosen-input string} OT functionality takes as input a bit $b$ from the receiver and two strings $m_0,m_1 \in \{0,1\}^k$ from the sender. It delivers $m_b$ to the receiver.
    \item $\cF_{\RROT[1]}$: the \emph{random-receiver-input bit} OT functionality takes as input $\top$ from the receiver and two bits $m_0,m_1 \in \{0,1\}$ from the sender. It samples $b \gets \{0,1\}$ and delivers $(b,m_b)$ to the receiver.
    \item $\cF_{\SROT[k]}$: the \emph{random-sender-input (string)} OT functionality takes as input $\top$ from the sender and $(b,m)$ from the receiver for $b \in \{0,1\}, m \in \{0,1\}^k$. It set $m_b = m$, samples $m_{1-b} \gets \{0,1\}^k$ and delivers  $(m_0,m_1)$ to the sender.
\end{itemize}

We will often refer to the following bit OT reversal theorem.

\begin{importedtheorem}[\cite{C:IKNP03}]\label{impthm:reversal}
Any protocol that securely realizes the functionality $\cF_{\SROT[1]}$ can be converted into a protocol that securely realizes the functionality $\cF_{\RROT[1]}$, without adding any messages.
\end{importedtheorem}

For concreteness, we specify how the OT reversal works. Suppose that Alice and Bob have access to an ideal OT functionality $\cF_{\SROT[1]}$ where Alice is the receiver and Bob is the sender. Their goal is to realize $\cF_{\RROT[1]}$ with roles reversed, i.e. with Alice as sender and Bob as receiver. This is achieved as follows.

\begin{itemize}
    \item Alice has input $m_0,m_1 \in \{0,1\}$, and samples  $c \gets \{0,1\}, r \gets \{0,1\}$.
    \item Alice and Bob run the protocol for $\cF_{\SROT[1]}$ where Alice inputs $(c,r)$ as receiver to $\cF_{\SROT[1]}$, and Alice sends \[\ell_0 \coloneqq m_0 \oplus r, \ell_1 \coloneqq m_1 \oplus r \oplus c\] along with her OT message to Bob.
    \item Bob obtains output $(r_0,r_1)$ from the protocol for $\cF_{\SROT[1]}$. Then, he sets \[b \coloneqq r_0 \oplus r_1, m_b \coloneqq \ell_b \oplus r_0,\] and outputs $(b,m_b)$.
\end{itemize}

\subsection{Quantum oracle results}

We state here some results on quantum oracle machine from prior literature, which we use in our proofs. 


\begin{importedtheorem}[One-way to hiding~\cite{C:AmbHamUnr19}]
\label{impthm:ow2h}
Let $X,Y$ be finite non-empty sets and let $(S,O_1,O_2,\ket{\psi})$ be sampled from an arbitrary distribution such that $S \subseteq X$, $O_1,O_2: X \rightarrow Y$ are such that $\forall x \not \in S, O_1(x) = O_2(x)$, and $\ket{\psi}$ is a quantum state on an arbitrary number of qubits. Let $A^{O}(\ket{\psi})$ be a quantum oracle algorithm that makes at most $q$ queries. Let $B^{O}(\ket{\psi})$ be an oracle algorithm that does the following: pick $i \leftarrow [q]$, run $A^{O}(\ket{\psi})$ until (just before) the $i^{th}$ query, measure the query input register in the computational basis, and output the measurement outcome $x$. Let
\begin{itemize}
    \item $P_{\text{left}} = \Pr[A^{O_1}(\ket{\psi})=1]$,
    \item $P_{\text{right}} = \Pr[A^{O_2}(\ket{\psi})=1]$,
    \item and $P_{\text{guess}} = \Pr[x \in S : x \leftarrow B^{O_1}(\ket{\psi})]$.
\end{itemize}
Then it holds that \[|P_{\text{left}} - P_{\text{right}}| \leq 2q \sqrt{P_{\text{guess}}}.\]
\end{importedtheorem}

The above theorem is actually a generalization of the theorem stated in \cite{C:AmbHamUnr19}, in which the input $\ket{\psi}$ is assumed to be a \emph{classical} bit string $z$. However, the proof given in \cite{C:AmbHamUnr19} readily extends to considering quantum input. The proof is split up into \cite[Lemma 8]{C:AmbHamUnr19} and \cite[Lemma 9]{C:AmbHamUnr19}. In Lemma 8, $(S,O_1,O_2,z)$ are fixed, and $z$ is used to define $A$'s initial state $\ket{\Psi_0}$. Here, we can just define $A$'s initial state as $\ket{\psi}$.
In Lemma 9, an expectation is taken over $(S,O_1,O_2,z)$, and the same expectation can be taken over $(S,O_1,O_2,\ket{\psi})$.

\begin{importedtheorem}[Measure-and-reprogram \cite{C:DFMS19,C:DonFehMaj20}]\footnote{This theorem was stated more generally in \cite{C:DFMS19,C:DonFehMaj20} to consider the drop in expectation for each specific $x^* \in X$. }\label{thm:measure-and-reprogram}
Let $X,Y$ be finite non-empty sets, and let $q \in \bbN$. Let $\Adv$ be a quantum oracle machine with initial state $\rho$ that makes at most $q$ queries to a uniformly random function $H: X \to Y$ and that outputs an $x \in X$ along with an arbitrary quantum state $\sigma$ on register $\cA$. There exists a quantum interactive machine $\Sim[\Adv]$ such that for any projection \[\Pi[y] \coloneqq \sum_x \ket{x}\bra{x} \otimes \Pi^{x,y}_\cA,\] where each $\Pi^{x,y}$ is an arbitrary projection on register $\cA$ that is parameterized by strings $x \in X$ and $y \in Y$, it holds that 
\begin{align*}&\E_H\left[\Tr\left(\Pi[H(x)] \left(\ket{x}\bra{x} \otimes \sigma\right)\right) : (x,\sigma) \gets \Adv^H(\rho)\right]\\ &\leq (2q+1)^2\E\left[\Tr\left(\Pi[y]\left(\ket{x}\bra{x} \otimes \sigma\right)\right) : \begin{array}{r} (x,\state) \gets \Sim[\Adv](\rho) \\ y \gets Y \\ \sigma \gets \Sim[\Adv](y,\state) \end{array}\right].\end{align*}

Moreover, $\Sim[\Adv]$ runs $\Adv$ except for the following differences: i) it introduces an intermediate measurement of one of the registers maintained by $\Adv$ to obtain $x$, and ii) it simulates responses to $\Adv$'s oracle queries to $H$.
\end{importedtheorem}


Finally, we will often make use of an ``on-the-fly'' method for simulating a quantum random oracle, due to~\cite{C:Zhandry19}. This method of simulation is \emph{efficient} and does not depend on an a priori upper bound on the number of queries $q$ to be made by the adversary. In fact, as shown by~\cite{DFMS21}, this simulation method may be augmented with an extraction interface that essentially allows to recovers a pre-image $x$ given an image $y = H(x)$.

Below, we define \emph{independent queries} to an interface to be two consecutive queries that can in principle be performed in either order. More formally, two consecutive queries are independent if they can applied to disjoint registers, meaning that one query may be applied to input and output registers $\cX$ and $\cY$, while the other may be applied to disjoint input and output registers $\cX'$ and $\cY'$.

Furthermore, we say that two quantum operations $E$ and $F$ $\alpha$-\emph{almost}-\emph{commute} if for any input state $\rho$, $\TD(E(F(\rho)),F(E(\rho)) \leq \alpha$.
We say that two quantum operations $E$ and $F$ \emph{commute} if for any input state $\rho$, $\TD(E(F(\rho)),F(E(\rho)) = 0$.

\begin{importedtheorem}[On-the-fly simulation with extraction~\cite{C:Zhandry19,DFMS21}]
\label{impthm:extROSim}
   Let $X$ be a finite non-empty set and $Y = \zo^n$. There exists a simulator $\simro$ that consists of an initialization step and an interface $\simro.\RO$ that maintains an internal state. $\simro.\RO$, given registers $\cX$ and $\cY$, applies a quantum operation to these registers and its internal state.\footnote{We can also consider applying $\simro.\RO$ to a classical input $x$ and producing classical output $y$, which corresponds to applying the quantum operation on $\ket{x}_\cX\ket{0}_\cY$ and measuring the register $\cY$ to produce the output.} The following properties hold for any oracle algorithm $A$.
   
   \begin{enumerate}
       \item \textbf{Indistinguishable simulation.}\[ \Pr_{H}[1 \leftarrow A^H] = \Pr[1 \leftarrow A^{\simro.\RO}]. \]
       \item \textbf{Efficiency.} Suppose that $A$ makes $q$ queries to $\simro.\RO$. Then the total runtime of $\simro$ is $O(q^2)$.
   \end{enumerate}
   
   There also exists an interface $\simro.\sE$ that upon input a classical value $y \in \zo^n$, outputs a classical value $\hat{x} \in X \cup \{\emptyset\}$. The following properties hold for any oracle algorithm $A$.

    
    \begin{enumerate}
        

        \item \textbf{Correctness of extraction.}\label{impthmstep:extcorrectness}
        Suppose that $A$ makes $q$ queries to $\simro.\RO$ and no queries to $\simro.\sE$, and outputs $\mathbf{x} \in X^\ell$ and $\mathbf{y} \in Y^\ell$. Then, 

        \[
            \Pr\left[ \exists \, i \, : \, (\mathbf{y}_i = \mathbf{\hat{y}}_i) \wedge (\mathbf{x}_i \neq \mathbf{\hat{x}}_i) \Bigg| \begin{array}{r} \mathbf{x},\mathbf{y} \leftarrow A^{\simro.\RO} \\ \mathbf{\hat{y}} \gets \simro.\RO(\mathbf{x}) \\ \mathbf{\hat{x}} \leftarrow \simro.\sE(\mathbf{y})  \end{array}\right] \leq \frac{296(q+\ell+1)^3+2}{2^n}.
        \]
        
         \item \textbf{Almost commutativity of $\simro.\RO$ and $\simro.\sE$.} \label{impthmstep:almostcommute} Any two independent queries to $\simro.\sE$ and $\simro.\RO$ $\frac{8\sqrt{2}}{2^{n/2}}$-almost-commute.
        
        
        \item \textbf{Efficiency.}\label{impthmstep:extROruntime}
        Suppose that $A$ makes $q_\RO$ queries to $\simro.\RO$ and $q_\sE$ queries to $\simro.\sE$. Then the total runtime of $\simro$ is $O(q_\RO q_\sE + q_\RO^2)$.

    \end{enumerate}
\end{importedtheorem}





\subsection{Quantum entropy and leftover hashing}
\label{subsec: quantum min entropy}





\paragraph{Quantum conditional min-entropy.} 

Let $\rho_{\cX\cY}$ denote a bipartite quantum state over registers $\cX,\cY$. Following~\cite{Renner08,KonRenSch09}, the conditional min-entropy of $\rho_{\cX\cY}$ given $\cY$ is then defined to be 
\[\mathbf{H}_\infty(\rho_{\cX\cY} \mid \cY) \coloneqq \sup_\tau \max \{h \in \mathbb{R} : 2^{-h} \cdot \bbI_\cX \otimes \tau_\cY - \rho_{\cX\cY} \geq 0\}.\]

In this work, we will exclusively consider the case where the $\rho_{\cX\cY}$ is a joint distribution of the form $(R,\tau)$ where $R$ is a classical random variable. In other words, $\rho_{\cX\cY}$  can be written as 
\[ \sum_x \Pr[X = x] \ket{x}\bra{x} \otimes \tau_x. \]

We refer to such $\rho_{\cX\cY}$ as a classical-quantum state. In this case, quantum conditional min-entropy exactly corresponds to the maximum probability of guessing $x$ given the state on register $\cY$.

\begin{importedtheorem}[\cite{KonRenSch09}]\label{impthm:conditional-min-entropy}
Let $\rho_{\cX,\cY}$ be a classical-quantum state, and let $p_{\mathsf{guess}}(\rho_{\cX,\cY} | \cY)$ be the maximum probability that any quantum operation can output the $x$ on register $\cX$, given the state on register $\cY$. Then \[p_{\mathsf{guess}}(\rho_{\cX,\cY} | \cY) = 2^{-\mathbf{H}_\infty(\rho_{\cX,\cY} | \cY)}.\]

\end{importedtheorem}

\paragraph{Leftover hash lemma with quantum side information.} We now state a generalization of the leftover hash lemma to the setting of quantum side information. 
\begin{importedtheorem}[\cite{TCC:RenKon05}]\label{impthm:privacy-amplification}
Let $\mathcal{H}$ be a family of universal hash functions from $X$ to $\{0,1\}^\ell$, i.e. for any $x \neq x'$, $\Pr_{h \leftarrow \mathcal{H}}[h(x) = h(x')] = 2^{-\ell}$. Let $\rho_{\cX\cY}$ be any classical-quantum state. Let $\cR$ be a register that holds $h \gets \cH$, let $\cK$ be a register that holds $h(x)$, where $x$ is on register $\cX$, and define $\rho_{\cX\cY\cR\cK}$ to be the entire system. Then, it holds that
\[\left\|\rho_{\cY\cR\cK} -  \rho_{\cY\cR} \otimes \frac{1}{2^\ell}\sum_{u}\ket{u}\bra{u}\right\|_1 \leq \frac{1}{2^{1+\frac{1}{2}(\mathbf{H}_\infty(\rho_{\cX\cY}|\cY)-\ell)}}.\]

\end{importedtheorem}

\paragraph{Small superposition of terms.} We will also make use of the following lemma from \cite{C:BouFeh10}.

\begin{importedtheorem}(\cite{C:BouFeh10})\label{impthm:small-superposition}
Let $\cX,\cY$ be registers of arbitrary size, and let $\{\ket{i}\}_{i \in I}$ and $\{\ket{w}\}_{w \in W}$ be orthonormal bases of $\cX$. Let $\ket{\psi}_{\cX\cY}$ and $\rho_{\cX\cY}$ be of the form \[\ket{\psi} = \sum_{i \in J}\alpha_i\ket{i}_\cX\ket{\psi_i}_\cY  \text{   and     } \rho = \sum_{i \in J}|\alpha_i|^2\ket{i}\bra{i}_\cX \otimes \ket{\psi_i}\bra{\psi_i}_\cY\] for some subset $J \subseteq I$. Furthermore, let $\widehat{\rho}_{\cX\cY}$ and $\widehat{\rho}^{\mathsf{mix}}_{\cX\cY}$ be the classical-quantum states obtained by measuring register $\cX$ of $\ket{\psi}$ and $\rho$, respectively, in basis $\{\ket{w}\}_{w \in W}$ to observe outcome $w$. Then, \[\mathbf{H}_\infty(\widehat{\rho}_{\cX,\cY} | \cY) \geq \mathbf{H}_\infty(\widehat{\rho}^{\mathsf{mix}}_{\cX,\cY} | \cY) - \log|J|.\]
\end{importedtheorem}

\subsection{Sampling in a quantum population}\label{subsec:quantum-sampling}

In this section, we describe a generic framework presented in \cite{C:BouFeh10} for analyzing cut-and-choose strategies applied to quantum states. 

\paragraph{Classical sampling stratiegies.} Let $A$ be a set, and let $\bq = (q_1,\dots,q_n) \in A^n$ be a string of length $n$. We consider the problem of estimating the relative Hamming weight of a substring $\omega(\bq_{\overline{t}})$ by only looking at the substring $\bq_t$ of $\bq$, for a subset $t \subset [n]$. We consider sampling strategies $\Psi = (P_T,P_S,f)$, where $P_T$ is an (independently sampled) distribution over subsets $t \subseteq [n]$, $P_S$ is a distribution over seeds $s \in S$, and $f: \{(t,\bv) : t \subset [n], \bv \in A^t\} \times S \to \bbR$ is a function that takes the subset $t$, the substring $\bv$, and a seed $s$, and outputs an estimate for the relative Hamming weight of the remaining string. For a fixed subset $t$, seed $s$, and a parameter $\delta$, define $B^\delta_{t,s}(\Psi) \subseteq A^n$ as \[B^{\delta}_{t,s} \coloneqq \{\bb \in A^n : |\omega(\bb_{\overline{t}}) - f(t,\bb_t,s)| < \delta\}.\] Then we define the \emph{classical error probability} of strategy $\Psi$ as follows.

\begin{definition}[Classical error probability]\label{def:classical-error-probability}
The classical error probability of a sampling strategy $\Psi = (P_T,P_S,f)$ is defined as the following value, paraterized by $0 < \delta < 1$: \[\epsilon^\delta_{\mathsf{classical}}(\Psi) \coloneqq \max_{\bq \in A^n}\Pr_{t \gets P_T,s \gets P_S}\left[\bq \notin B_{t,s}^\delta(\Psi)\right].\]
\end{definition}

\paragraph{Quantum sampling strategies.} Now, let $A = A_1,\dots,A_n$ be an $n$-partite quantum system on registers $\cA = \cA_1 \otimes \dots \otimes \cA_n$, where each system has dimension $d$. Let $\{\ket{a}\}_{a}$ be a fixed orthonormal basis for each $\cA_i$. $\cA$ may be entangled with another system $\cE$, and we write the purified state on $\cA$ and $\cE$ as $\ket{\psi}_{\cA\cE}$. We consider the problem of testing whether the state on $\cA$ is close to the all-zero reference state $\ket{0}_{\cA_1}\dots\ket{0}_{\cA_n}$. There is a natural way to apply any sampling strategy $\Psi = (P_T,P_S,f)$ to this setting: sample $t,s$ according to $P_T,P_S$, measure subsystems $\cA_i$ for $i \in [t]$ in basis $\{\ket{a}\}_a$ to observe $\bq_{t} \in A^{|t|}$, and compute an estimate $f(t,\bq_t,s)$. 

In order to analyze the effect of this strategy, we first consider the mixed state on registers $\cT$ (holding the subset $t$), $\cS$ (holding the seed $s$), and $\cA,\cE$ that results from sampling $t$ and $s$ according to $P_{TS} = P_TP_S$ 

\[\rho_{\cT\cS\cA\cE} = \sum_{t,s}P_{TS}(t,s)\ket{t,s}\bra{t,s} \otimes \ket{\psi}\bra{\psi}.\] Next, we compare this state to an \emph{ideal} state, parameterized by $0 < \delta < 1$, of the form 

\[\widetilde{\rho}_{\cT\cS\cA\cE} = \sum_{t,s} P_{TS}(t,s)\ket{t,s}\bra{t,s} \otimes \ket{\widetilde{\psi}^{ts}}\bra{\widetilde{\psi}^{ts}} \text{   with    } \ket{\psi^{ts}} \in \mathsf{span}\left(B_{t,s}^\delta\right) \otimes \cE,\] where \[\mathsf{span}\left(B_{t,s}^\delta\right) \coloneqq \mathsf{span}\left(\{\ket{\bb} : \bb \in B_{t,s}^\delta\}\right) = \mathsf{span}\left(\{\ket{\bb}: |\omega(\bb_{\overline{t}}) - f(t,\bb_t,s)| < \delta\}\right).\] That is, $\widetilde{\rho}_{\cT\cS\cA\cE}$ is a state such that it holds \emph{with certainty} that the state on registers $\cA_{\overline{t}}\cE$, after having measured $\cA_{t}$ and oberserving $\bq_t$, is in a superposition of states with relative Hamming weight $\delta$-close to $f(t,\bq_t,s)$. This leads us to the definition of the \emph{quantum error probability} of strategy $\Psi$.

\begin{definition}[Quantum error probability]\label{def:quantum-error-probability}
The quantum error probability of a sampling strategy $\Psi = (P_T,P_S,f)$ is defined as the following value, parameterized by $0 < \delta < 1$: \[\epsilon_{\mathsf{quantum}}^\delta(\Psi) \coloneqq \max_\cE\max_{\ket{\psi}_{\cA\cE}}\min_{\widetilde{\rho}_{\cT\cS\cA\cE}}\TD\left(\rho_{\cT\cS\cA\cE},\widetilde{\rho}_{\cT\cS\cA\cE}\right),\] where the first max is over all finite-dimensional registers $\cE$, the second max is over all state $\ket{\psi}_{\cA\cE}$ and the min is over all ideal state $\widetilde{\rho}_{\cT\cS\cA\cE}$ of the form described above.
\end{definition}

Finally, we relate the classical and quantum error probabilities.

\begin{importedtheorem}[\cite{C:BouFeh10}]\label{impthm:error-probability}
For any sampling strategy $\Psi$ and $\delta > 0$, \[\epsilon_{\mathsf{quantum}}^\delta(\Psi) \leq \sqrt{\epsilon_{\mathsf{classical}}^\delta(\Psi)}.\]
\end{importedtheorem}



\fi
\section{Seedless extraction from quantum sources}\label{sec:extraction}

In this section, we consider the problem of seedless randomness extraction from a quantum source of entropy. The source of entropy we are interested in comes from applying a Hadamard basis measurement to a state that is in a ``small'' superposition of computational basis vectors. More concretely, consider an arbitrarily entangled system on registers $\cA,\cX$, where $\cX$ is in a small superposition of computational basis vectors. Then, we want to specify an extractor $E$ such that, if $x$ is obtained by measuring register $\cX$ in the Hadamard basis, then $E(x)$ looks uniformly random, even given the ``side information'' on register $\cA$. Note that \emph{seeded} randomness extraction in this setting has been well-studied (e.g. \cite{TCC:RenKon05,C:DFLSS09,C:BouFeh10}).

Proofs of the following two theorems are given in \cref{sec:adaptive-proof}.


\subsection{The XOR extractor}

First, we observe that if $E$ just XORs all the bits of $x$ together, then the resulting bit $E(x)$ is \emph{perfectly} uniform, as long as the original state on $\cX$ is only supported on vectors with relative Hamming weight < 1/2.

\begin{theorem}\label{lemma:XOR-extractor}
Let $\cX$ be an $n$-qubit register, and consider any state $\ket{\gamma}_{\cA,\cX}$ that can be written as \[\ket{\gamma} = \sum_{u: \cH\cW(u) < n/2} \ket{\psi_u}_{\cA} \otimes \ket{u}_\cX.\] Let $\rho_{\cA,\cP}$ be the mixed state that results from measuring $\cX$ in the Hadamard basis to produce $x$, and writing $\bigoplus_{i \in [n]}x_i$ into the single qubit register $\cP$. Then it holds that \[\rho_{\cA,\cP} = \Tr_\cX(\ket{\gamma}\bra{\gamma}) \otimes \left(\frac{1}{2}\ket{0}\bra{0} + \frac{1}{2}\ket{1}\bra{1}\right).\] 
\end{theorem}

\subsection{The RO extractor}

Next, our goal is to extract multiple bits of randomness from $x$. To do this, we model $E$ as a \emph{random oracle}. We derive a bound on the advantage any adversary has in distinguishing $E(x)$ from a uniformly random string, based on the number of qubits $k$ in the register $\cX$, the number of vectors $C$ in the superposition on register $\cX$, and the number of queries $q$ made to the random oracle. In fact, to be as general as possible, we consider a random oracle with input length $n$, and allow $n-k$ of the bits of the input to the random oracle to be (adaptively) determined by the adversary, while the remaining $k$ bits are sampled by measuring a $k$-qubit register $\cX$.

\begin{theorem}\label{thm:ROM-extractor}
Let $H : \{0,1\}^n \to \{0,1\}^m$ be a uniformly random function, and let $q,C,k$ be integers. Consider a two-stage oracle algorithm $(A_1^H, A_2^H)$ that combined makes at most $q$ queries to $H$. Suppose that $A_1^H$ outputs classical strings $(T,\{x_i\}_{i \in T})$, and let $\ket{\gamma}_{\cA,\cX}$ be its left-over quantum state,\footnote{That is, consider sampling $H$, running a purified $A_1^H$, measuring at the end to obtain $(T,\{x_i\}_{i \in T})$, and then defining $\ket{\gamma}$ to be the left-over state on $\cA$'s remaining registers.} where $T \subset [n]$ is a set of size $n-k$, each $x_i \in \{0,1\}$, $\cA$ is a register of arbitary size, and $\cX$ is a register of $k$ qubits. Suppose further that with probability 1 over the sampling of $H$ and the execution of $A_1$, there exists a set $L \subset \{0,1\}^k$ of size at most $C$ such that $\ket{\gamma}$ may be written as follows: \[\ket{\gamma} = \sum_{u \in L}\ket{\psi_u}_\cA \otimes \ket{u}_\cX.\] Now consider the following two games.

\begin{itemize}
    \item $\mathsf{REAL}$: 
    \begin{itemize}
        \item $A_1^H$ outputs $T,\{x_i\}_{i \in T},\ket{\gamma}_{\cA,\cX}$.
        \item $\cX$ is measured in the Hadamard basis to produce a $k$-bit string which is parsed as $\{x_i\}_{i \in \overline{T}}$, and a left-over state $\ket{\gamma'}_{\cA}$ on register $\cA$. Define $x = (x_1,\dots,x_n)$. 
        \item $A_2^H$ is given $T, \{x_i\}_{i \in T},\ket{\gamma'}_{\cA}, H(x)$, and outputs a bit.
    \end{itemize}
    \item $\mathsf{IDEAL}$: 
    \begin{itemize}
        \item $A_1^H$ outputs $T,\{x_i\}_{i \in T},\ket{\gamma}_{\cA,\cX}$.
        \item $r \gets \{0,1\}^m$.
        \item $A_2^H$ is given $T, \{x_i\}_{i \in T},\Tr_\cX(\ket{\gamma}\bra{\gamma}),r$, and outputs a bit.
    \end{itemize}
\end{itemize}
Then, \[\left|\Pr[\mathsf{REAL} = 1] - \Pr[\mathsf{IDEAL} = 1]\right| \leq \frac{2\sqrt{q}C + 2q\sqrt{C}}{2^{k/2}} < \frac{4qC}{2^{k/2}}.\]
\end{theorem}

\section{Non-interactive extractable and equivocal commitments}
\label{sec:uccomm}


    


    

A non-interactive commitment scheme with partial opening in the quantum random oracle model consists of classical oracle algorithms $(\Com,\Open,\Rec)$ with the following syntax. 


\begin{itemize}
    \item $\Com^H(1^\secp,\{m_i\}_{i \in [n]})$: On input the security parameter $\secp$ and $n$ messages $\{m_i \in \{0,1\}^k\}_{i \in [n]}$, output $n$ commitments $\{\com_i\}_{i \in [n]}$ and a state $\state$.
    \item $\Open^H(\state,T)$: On input a state $\state$ and a set $T \subseteq [n]$, output messages $\{m_i\}_{i \in T}$ and openings $\{u_i\}_{i \in T}$.
    \item $\Rec^H(\{\com_i\}_{i \in [n]},T,\{m_i,u_i\}_{i \in T})$: on input $n$ commitments $\{\com_i\}_{i \in [n]}$, a set $T$, and a set of message opening pairs $\{m_i,u_i\}_{i \in T}$, output either $\{m_i\}_{i \in T}$ or $\bot$.
\end{itemize}

The commitment scheme is parameterized by $n = n(\secp)$ which is the number of messages to be committed in parallel, and $k = k(\secp)$ which is the number of bits per message.


\subsection{Definitions}\label{subsec:com-defs}

\begin{definition}[Correctness]\label{def:correctness}
A non-interactive commitment scheme with partial opening in the QROM is \emph{correct} if for any $\{m_i\}_{i \in [n]}$ and $T \subseteq [n]$, \[\Pr\left[\Rec^H(\left\{\com_i\right\}_{i \in [n]},T,\{m_i,u_i\}_{i \in T}) = \{m_i\}_{i \in T}: \begin{array}{r}(\state,\{\com_i\}_{i \in [n]}) \gets \Com^H(1^\secp,\{m_i\}_{i \in [n]}) \\ \{m_i,u_i\}_{i \in T} \gets \Open^H(\state,T)\end{array}\right] = 1.\]
\end{definition}

\begin{definition}[$\mu$-Hiding]\label{def:hiding}
A non-interactive commitment scheme with partial opening in the QROM is $\mu$-hiding if for any adversary $\Adv$ that makes at most $q(\secp)$ queries to the random oracle, and any two sets of messages $\{m_{i,0}\}_{i \in [n]}$ and $\{m_{i,1}\}_{i \in [n]}$, it holds that 
\begin{align*}\Bigg|&\Pr\left[\Adv^H(\{\com_i\}_{i \in [n]}) = 1 : \{\com_i\}_{i \in [n]} \gets \Com^H(1^\secp,\{m_{i,0}\}_{i \in [n]})\right] - \\ &\Pr\left[\Adv^H(\{\com_i\}_{i \in [n]}) = 1 : \{\com_i\}_{i \in [n]} \gets \Com^H(1^\secp,\{m_{i,1}\}_{i \in [n]})\right]\Bigg| = \mu(\secp,q(\secp)).\end{align*}
Furthermore, we say that a commitment is \emph{hiding} if it is $\mu$-hiding, where $\mu$ is such that for any $q(\secp) = \poly(\secp)$, $\mu(\secp,q(\secp)) = \negl(\secp)$.
\end{definition}

\begin{definition}[$\mu$-Extractability]
\label{def:fextractable}
A non-interactive commitment scheme with partial opening in the QROM is $\mu$-\emph{extractable} if there exists a polynomial $s$ and a simulator $\SimExt$ consisting of an interface $\SimExt.\RO$ and an algorithm $\SimExt.\Ext$ that may share a common state, such that for any family of quantum oracle algorithms $\{\Adv_\secp = (\Adv_{\Commit,\secp},\allowbreak\Adv_{\Open,\secp},\allowbreak\sD_{\secp})\}_{\secp \in \bbN}$ that makes at most $q(\secp)$ queries to the random oracle, it holds that 

\begin{align*}\Bigg|&\Pr_{H}\left[\sD_{\secp}^H(\brho_{2},\out) = 1: \begin{array}{r}(\brho_{1},\{\com_i\}_{i \in [n]}) \gets \Adv_{\Commit,\secp}^H\\ (\brho_{2},T,\{m_i,u_i\}_{i \in T}) \gets \Adv_{\Open,\secp}^H(\brho_{1}) \\ \out \gets \Rec^H(\{\com_i\}_{i \in [n]},T,\{m_i,u_i\}_{i \in T})\end{array}\right]\\
&-\Pr\left[\sD_{\secp}^{\SimExt.\RO}(\brho_{2},\out) = 1: \begin{array}{r}(\brho_{1},\{\com_i\}_{i \in [n]}) \gets \Adv_{\Commit,\secp}^{\SimExt.\RO}\\ \{m_i^*\}_{i \in [n]} \gets \SimExt.\Ext(\{\com_i\}_{i \in [n]}) \\ (\brho_{2},T,\{m_i,u_i\}_{i \in T}) \gets \Adv_{\Open,\secp}^{\SimExt.\RO}(\brho_{1}) \\ \out \gets \Rec^{\SimExt.\RO}(\{\com_i\}_{i \in [n]},T,\{m_i,u_i\}_{i \in T}) \\ \out \coloneqq \mathsf{FAIL} \text{ if } \out \notin \{\{m_i^*\}_{i \in T},\bot\}\end{array}\right]\Bigg| = \mu(\secp,q(\secp)),
\end{align*}
where the state of $\SimExt$ was kept implicit, and the total run-time of $\SimExt$ on security parameter $1^\secp$ is at most $s(\secp,q(\secp))$. The interface $\SimExt.\RO$ is invoked on each query to $H$ made by $\Adv$ and $\Rec$, while the algorithm $\SimExt.\Ext$ is invoked on the classical commitments output by $\Adv$.

Furthermore, we say that a commitment is \emph{extractable} if it is $\mu$-extractable, where $\mu$ is such that for any $q(\secp) = \poly(\secp)$, $\mu(\secp,q(\secp)) = \negl(\secp)$.

Finally, we say that the commitment scheme satisfies \emph{extraction with a $\nu$-commuting simulator} if a call to $\SimExt.\RO$ $\nu(\secp)$-almost-commutes with the operation $\SimExt.\Ext$ when the input and output registers of $\SimExt.\RO$ and $\SimExt.\Ext$ are disjoint.\footnote{Note that $\SimExt.\RO$ and $\SimExt.\Ext$ can share a common state, so do not necesarily commute even when their inputs and output registers are disjoint.}


\end{definition}

\begin{definition}[$\mu$-Equivocality]\label{def:fequivocal}
A non-interactive commitment scheme with partial opening in the QROM is $\mu$-\emph{equivocal} if there exists a polynomial $s$ and a simulator $\SimEqu$ that consists of an interface $\SimEqu.\RO$ and two algorithms $\SimEqu.\Com,\SimEqu.\Open$ that may all share a common state, such that for any family of quantum oracle algorithms $\{\Adv_{\secp} = (\Adv_{\RCommit,\secp},\allowbreak\Adv_{\ROpen,\secp},\allowbreak\sD_{\secp})\}_{\secp \in \bbN}$ that makes at most $q(\secp)$ queries to the random oracle, it holds that 
\begin{align*}
    \Bigg|&\Pr_H\left[\sD_{\secp}^H(\brho_{2},\{\com_i,m_i,u_i\}_{i \in [n]}) = 1 : \begin{array}{r}(\brho_{1},\{m_i\}_{i \in [n]}) \gets \Adv_{\RCommit,\secp}^H \\ (\state,\{\com_i\}_{i \in [n]}) \gets \Com^H(1^\secp,\{m_i\}_{i \in [n]}) \\ \brho_{2} \gets \Adv_{\ROpen,\secp}^H(\brho_{1},\{\com_i\}_{i \in [n]}) \\ \{m_i,u_i\}_{i \in [n]} \gets \Open^H(\state, [n]) \end{array}\right] \\
    &-\Pr\left[\sD_{\secp}^{\SimEqu.\RO}(\brho_{2},\{\com_i,m_i,u_i\}_{i \in [n]}) = 1: \begin{array}{r}(\brho_{1},\{m_i\}_{i \in [n]}) \gets \Adv^{\SimEqu.\RO}_{\RCommit,\secp} \\ \{\com_i\}_{i \in [n]} \gets \SimEqu.\Com \\ \brho_{2} \gets \Adv^{\SimEqu.\RO}_{\ROpen,\secp}(\brho_{1},\{\com_i\}_{i \in [n]}) \\ \{u_i\}_{i \in [n]} \gets \SimEqu.\Open(\{m_i\}_{i \in [n]})\end{array}\right] \Bigg| = \mu(\secp,q(\secp)),
\end{align*}
where the state of $\SimEqu$ was kept implicit, and the total run-time of $\SimEqu$ on security parameter $1^\secp$ is at most $s(\secp,q(\secp))$. The interface $\SimEqu.\RO$ is invoked on each query to $H$ made by $\Adv$, the algorithm $\SimEqu.\Com$ is invoked to produce commitments, and the algorithm $\SimEqu.\Open$ is invoked on a set of messages to produce openings.

Furthermore, we say that a commitment is \emph{equivocal} if it is $\mu$-equivocal, where $\mu$ is such that for any $q(\secp) = \poly(\secp)$, $\mu(\secp,q(\secp)) = \negl(\secp)$.
\end{definition}

It is easy to see that a $\mu$-equivocal commitment satisfies $2\mu$-hiding, since one can first move from committing to $\{m_{i,0}\}_{i \in [n]}$ to a hybrid where the equivocality simulator is run, and then move to committing to $\{m_{i,1}\}_{i \in [n]}$.

We also note that all our definitions consider {\em classical commitments}, where the commitment string itself is purely classical. Furthermore, we assume that any potentially quantum state sent by a malicious committer is immediately measured by an honest receiver to produce a classical string -- it is this classical string that serves as the commitment. This is similar to prior works that consider commitments in the QROM (eg.,~\cite{DFMS21}), and we refer the reader to~\cite{BitBra21} for additional discussions about enforcing classical (parts of) commitments via measurement.

\ifsubmission We provide a natural construction, and proofs of extractability and equivocality for this construction in \cref{sec:com-proofs}.\else\fi

\fflater{\nishant{if we get time, we can try to see if we can generalize this}}

\ifsubmission\else\ifsubmission\section{Commitments}\label{sec:com-proofs}\else \fi
\subsection{Construction} \label{sec:eecom-construction}

\protocol
{\proref{fig:uccomm-protocol}}
{Extractable and equivocal commitment scheme}
{fig:uccomm-protocol}
{
Parameters: security parameter $\secp$, number of commitments $n = n(\lambda)$\\
Random oracle: $H:\zo^{\secp+1} \rightarrow \zo^{\secp+1}$.

\begin{itemize}
    \item {$\Com^{H}(1^\secp,\{b_i\}_{i \in [n]})$}: For all $i \in [n]$, sample $r_i \leftarrow \zo^\secp$ and set $\com_i = H(b_i||r_i)$. Set $\state = \{b_i,r_i\}_{i \in [n]}$ and output $(\state, \{\com_i\}_{i \in [n]})$.
    
    \item $\Open^{H}(\state, T)$: Parse $\state$ as $\{b_i,r_i\}_{i \in [n]}$ and output $\{b_i,r_i\}_{i \in T}$. 
    
    \item $\Rec^{H}(\{\com_i\}_{i \in [n]},T,\{b_i,r_i\}_{i \in T})$: Output $\bot$ if there exists $i \in T$ s.t.\ $H(b_i||r_i) \neq \com_i$. Otherwise output $\{b_i\}_{i \in T}$.
\end{itemize}
}

We construct extractable and equivocal bit commitments in the QROM in \cref{fig:uccomm-protocol}. Without loss of generality, a committer can commit to strings of length $>1$ by committing to each bit in the string one by one, and sending all commitments in parallel.\fi
\ifsubmission\else

\subsection{Extractability}

In this section, we prove the following theorem by relying on
\cref{impthm:extROSim}. 
We remark that our proof of extraction uses ideas already present in \cite{DFMS21} to establish that our construction satisfies \cref{def:fextractable}. 

\begin{theorem}\label{thm:extractable}
    \proref{fig:uccomm-protocol} is a $\mu$-extractable non-interactive commitment scheme with partial opening in the QROM, with message length $1$ (i.e.\ $k=1$), satisfying \cref{def:fextractable}, where $\mu(\secp,q, n) =\frac{8qn}{2^{\secp/2}} + \frac{148(q+n+1)^3+1}{2^{\secp}}$\footnote{When $n = c \lambda$ for some arbitrary fixed constant $c$, then we can define $\mu_c(\secp, q) = \frac{2q(c \secp)^{1/2}}{2^{\secp/2}}$. In all our OT protocols, we will set $n$ in this manner and will assume that $\mu$ is a function of $\lambda, q$.}, and where the runtime of the simulator is bounded by $s(\secp,q) = O(q^2 + q\cdot n(\secp))$. In addition, the protocol satisfies extraction with $\nu$-commuting simulator, where $\nu(\secp) = \frac{8}{2^{\secp/2}}$.
\end{theorem}
\begin{proof}

Let $(\simro.\RO,\simro.\sE)$ be the on-the-fly random oracle simulator with extraction from \cref{impthm:extROSim}. The extractable commitment simulator $\SimExt = (\SimExt.\RO,\SimExt.\Ext)$ is defined as follows.

\begin{itemize}
    \item $\SimExt.\RO = \simro.\RO$
    \item $\SimExt.\Ext$ runs $\simro.\sE$ to obtain either a $\secp+1$ bit string $x^*$, or $\emptyset$. In the case of $x^*$, output the first bit of $x^*$. In the case of $\emptyset$, output 0.
\end{itemize}

    
    We now prove that for any family of quantum oracle algorithms $\{\Adv_\secp = (\Adv_{\Commit,\secp},\Adv_{\Open,\secp},\sD_{\secp})\}_{\secp \in \bbN}$, the two experiments in \cref{def:fextractable} are $\mu(\secp,q)$ close, where $\mu(\secp,q) = \frac{148(q+n+1)^3+1}{2^{\secp}} + \frac{8qn}{2^{\secp/2}}$. We consider the following sequence of hybrids (where parts in \nblue{blue} indicate difference from the previous hybrid):
    
    \begin{itemize}
        \item $\Hyb_0$: This corresponds to the ``real'' experiment in \cref{def:fextractable}. 
            \begin{enumerate}
                \item Sample oracle $H \leftarrow F_{\zo^{\secp+1} \rightarrow \zo^{\secp+1}}$.
                \item $(\brho_{1},\{\com_i\}_{i \in [n]}) \gets \Adv_{\Commit,\secp}^H$
                \item $(\brho_{2},\{b_i,r_i\}_{i \in T}) \gets \Adv_{\Open,\secp}^H(\brho_{1})$
                \item $\out \gets \Rec^H(\{\com_i\}_{i \in [n]},\{b_i,r_i\}_{i \in T})$
                \item Output $b \gets \sD_{\secp}^H(\brho_{2},\out)$
            \end{enumerate}
        
        \item $\Hyb_1$: This is the same as previous hybrid, except that all oracle calls to $H$ are answered by $\simro.\RO$. 
            \begin{enumerate}    
                \item 
                \nblue{Initialize the extractable random oracle simulator, $\simro$.}
                \item $(\brho_{1},\{\com_i\}_{i \in [n]}) \gets \Adv_{\Commit,\secp}^{\nblue{\simro.\RO}}$
                \item $(\brho_{2},\{b_i,r_i\}_{i \in T}) \gets \Adv_{\Open,\secp}^{\nblue{\simro.\RO}}(\brho_{1})$
                \item $\out \gets \Rec^{\nblue{\simro.\RO}}(\{\com_i\}_{i \in [n]},\{b_i,r_i\}_{i \in T})$
                \item Output $b \leftarrow \sD_{\secp}^{\nblue{\simro.\RO}}(\brho_{2},\out)$
            \end{enumerate}
            
        \item $\Hyb_2$: This is the same as the previous hybrid except for an additional query to $\simro.\RO$ that is performed at the end of the experiment, along with an event $\mathsf{BAD}$ that we define. Notice also that we have opened up the description of the algorithm $\Rec$ below. The hybrid outputs the following distribution.
            \begin{enumerate}
                \item  Initialize the extractable random oracle simulator, $\simro$.
                \item \label{step:extComSimHyb2S1} $(\brho_{1},\{\com_i\}_{i \in [n]}) \gets \Adv_{\Commit,\secp}^{\simro.\RO}$
                \item \label{step:extComSimHyb2S2} $(\brho_{2},\{b_i,r_i\}_{i \in T}) \gets \Adv_{\Open,\secp}^{\simro.\RO}(\brho_{1})$
                \item $\out \gets \Rec^{\simro.\RO}(\{\com_i\}_{i \in [n]},\{b_i,r_i\}_{i \in T})$
                \begin{itemize}
                \item If there exists $i \in T$ s.t.\ $\simro.\RO(b_i||r_i) \neq \com_i$, set $\out \coloneqq \bot$, otherwise set $\out \coloneqq \{b_i\}_{i \in T}$. \end{itemize}
                \item \label{step:extComSimHyb2S4}
                Obtain bit $b \leftarrow \sD_{\secp}^{\simro.\RO}(\brho_{2},\out)$.
                \item \nblue{For all $i \in [n]$, set $y_i \leftarrow \simro.\RO(b_i||r_i)$.}
                \item \nblue{If $\out \neq \bot$ and there exists $i \in T$ such that $y_i \neq \com_i$, output $\mathsf{BAD}$}, otherwise output $b$.
            \end{enumerate}
            
            
        \item $\Hyb_3:$ This is the same as the previous hybrid except that there is a query to $\simro.\sE$, and an extra condition in the $\mathsf{BAD}$ event. The hybrid outputs the following distribution:
            \begin{enumerate}
                \item  Initialize the extractable random oracle simulator, $\simro$.
                \item  $(\brho_{1},\{\com_i\}_{i \in [n]}) \gets \Adv_{\Commit,\secp}^{\simro.\RO}$
                \item $(\brho_{2},\{b_i,r_i\}_{i \in T}) \gets \Adv_{\Open,\secp}^{\simro.\RO}(\brho_{1})$
                \item $\out \gets \Rec^{\simro.\RO}(\{\com_i\}_{i \in [n]},\{b_i,r_i\}_{i \in T})$
                \begin{itemize}
                \item If there exists $i \in T$ s.t.\ $\simro.\RO(b_i||r_i) \neq \com_i$, set $\out \coloneqq \bot$, otherwise set $\out \coloneqq \{b_i\}_{i \in T}$. \end{itemize}
                \item 
                Obtain bit $b \leftarrow \sD_{\secp}^{\simro.\RO}(\brho_{2},\out)$
                
                \item For all $i \in [n]$, set $y_i \leftarrow \simro.\RO(b_i||r_i)$.
                \item \nblue{For all $i \in [n]$, set $x^*_i \leftarrow \simro.\sE(\com_i)$.}
                \item \nblue{If there exists $i \in T$ such that 
                $(x_i^* \neq (b_i||r_i)) \wedge (y_i=\com_i)$, output $\mathsf{BAD}$}, or if $\out \neq \bot$ and there exists $i \in T$ such that $y_i \neq \com_i$, output $\mathsf{BAD}$, otherwise output
                $b$.
            \end{enumerate}
            
        
        \item $\Hyb_4$: This hybrid is identical to the previous one except that $\simro.\sE$ is called earlier on in the hybrid. The hybrid outputs the following distribution:
            \begin{enumerate}
                \item  Initialize the extractable random oracle simulator, $\simro$.
                \item  $(\brho_{1},\{\com_i\}_{i \in [n]}) \gets \Adv_{\Commit,\secp}^{\simro.\RO}$
                \item \nblue{For all $i \in [n]$, set $x^*_i \leftarrow \simro.\sE(\com_i)$. }
                \item  $(\brho_{2},\{b_i,r_i\}_{i \in T}) \gets \Adv_{\Open,\secp}^{\simro.\RO}(\brho_{1})$
                \item $\out \gets \Rec^{\simro.\RO}(\{\com_i\}_{i \in [n]},\{b_i,r_i\}_{i \in T})$
                \begin{itemize}
                \item If there exists $i \in T$ s.t.\ $\simro.\RO(b_i||r_i) \neq \com_i$, set $\out \coloneqq \bot$, otherwise set $\out \coloneqq \{b_i\}_{i \in T}$. \end{itemize}
                \item 
                $b \leftarrow \sD_{\secp}^{\simro.\RO}(\brho_{2},\out)$
                \item For all $i \in [n]$, set $y_i \leftarrow \simro.\RO(b_i||r_i)$.
                \item If there exists $i \in T$ such that 
                $(x_i^* \neq (b_i||r_i)) \wedge (y_i=\com_i)$, output $\mathsf{BAD}$, or if $\out \neq \bot$ and there exists $i \in T$ such that $y_i \neq \com_i$, output $\mathsf{BAD}$, otherwise output
                $b$.
            \end{enumerate}
            
        
        \item $\Hyb_5$: This hybrid is identical to the previous hybrid except for altering the variable $\mathsf{out}$ to sometimes take the value $\mathsf{FAIL}$. The hybrid outputs the following distribution:
            \begin{enumerate}
                \item  Initialize the extractable random oracle simulator, $\simro$.
                \item $(\brho_{1},\{\com_i\}_{i \in [n]}) \gets \Adv_{\Commit,\secp}^{\simro.\RO}$
                \item For all $i \in [n]$, set $x^*_i \leftarrow \simro.\sE(\com_i)$.
                \nblue{For all $i \in [n]$, if $x_i^* = \emptyset$, set $b^*_i \coloneqq 0$, and otherwise set $b^*_i$ equal to the first bit of $x_i^*$.}
                \item $(\brho_{2},\{b_i,r_i\}_{i \in T}) \gets \Adv_{\Open,\secp}^{\simro.\RO}(\brho_{1})$
                \item \label{step:ExtROHyb3RecStep} Obtain $\out \leftarrow \Rec^{\simro.\RO}(\{\com_i\}_{i \in [n]},\{b_i,r_i\}_{i \in T})$ as follows: 
                \begin{itemize}
                    \item If there exists $i \in T$ s.t.\ $\simro.\RO(b_i||r_i) \neq \com_i$, set $\out \coloneqq \bot$, otherwise set $\out \coloneqq \{b_i\}_{i \in T}$. 
                \end{itemize}
                \item \label{step:extComSimHyb4S6} \nblue{If $\out \notin \{\{b_i^*\}_{i \in T},\bot\}$, set $\out = \mathsf{FAIL}$}.
                \item   $b \leftarrow \sD_{\secp}^{\simro.\RO}(\brho_{2},\out)$
                \item For all $i \in [n]$, set $y_i \leftarrow \simro.\RO(b_i||r_i)$.
                \item If there exists $i \in T$ such that 
                $(x_i^* \neq (b_i||r_i)) \wedge (y_i=\com_i)$, output $\mathsf{BAD}$, or if $\out \neq \bot$ and there exists $i \in T$ such that $y_i \neq \com_i$, output $\mathsf{BAD}$, otherwise output $b$.
            \end{enumerate}
            
            \item $\Hyb_6$: This hybrid is identical to the previous hybrid except for removing the final query to $\simro.\RO$ and the event $\mathsf{BAD}$. 
            \begin{enumerate}
                \item  Initialize the extractable random oracle simulator, $\simro$.
                \item $(\brho_{1},\{\com_i\}_{i \in [n]}) \gets \Adv_{\Commit,\secp}^{\simro.\RO}$
                \item For all $i \in [n]$, set $x^*_i \leftarrow \simro.\sE(\com_i)$.
                For all $i \in [n]$, if $x_i^* = \emptyset$, set $b^*_i \coloneqq 0$, and otherwise set $b^*_i$ equal to the first bit of $x_i^*$.
                \item $(\brho_{2},\{b_i,r_i\}_{i \in T}) \gets \Adv_{\Open,\secp}^{\simro.\RO}(\brho_{1})$
                \item \label{step:ExtROHyb3RecStep} Obtain $\out \leftarrow \Rec^{\simro.\RO}(\{\com_i\}_{i \in [n]},\{b_i,r_i\}_{i \in T})$ as follows: 
                \begin{itemize}
                    \item If there exists $i \in T$ s.t.\ $\simro.\RO(b_i||r_i) \neq \com_i$, set $\out \coloneqq \bot$, otherwise set $\out \coloneqq \{b_i\}_{i \in T}$. 
                \end{itemize}
                \item If $\out \notin \{\{b_i^*\}_{i \in T},\bot\}$, set $\out = \mathsf{FAIL}$.
                \item $b \leftarrow \sD_{\secp}^{\simro.\RO}(\brho_{2},\out)$
                \item  \nblue{\sout{For all $i \in [n]$, set $y_i \leftarrow \simro.\RO(b_i||r_i)$.}}
                \item \nblue{\sout{If there exists $i \in T$ such that 
                $(x_i^* \neq (b_i||r_i)) \wedge (y_i=\com_i)$, output $\mathsf{BAD}$, or if $\out \neq \bot$ and there exists $i \in T$ such that $y_i \neq \com_i$, output $\mathsf{BAD}$, otherwise output $b$.}}
            \end{enumerate}
    \end{itemize}
    We note that $\Hyb_6$ is the simulated distribution. We prove indistinguishability between the hybrids below.
    \begin{claim}
        $\Pr[\Hyb_0=1] = \Pr[\Hyb_1=1]$
    \end{claim}
    \begin{proof}
        This follows from the indistinguishable simulation property of Imported Theorem~\ref{impthm:extROSim}.
    \end{proof}
    
    \begin{claim}
        $\Pr[\Hyb_1=1] = \Pr[\Hyb_2=1]$
    \end{claim}
    \begin{proof}
    First, adding the extra query to $\simro.\RO$ does not affect the output of the experiment since it is performed after $b$ is computed. Next, the event $\mathsf{BAD}$ only occurs if some classical query $(b_i||r_i)$ to $\simro.\RO$ returns different classical values at different points in the experiment. However, this can never occur due to the indistinguishable simulation property of Imported Theorem~\ref{impthm:extROSim}, and because two classical queries to an oracle $H$ always return the same value.
    
    \end{proof}
    
    \begin{claim}
            $|\Pr[\Hyb_2=1] - \Pr[\Hyb_3=1]| \leq  \frac{148(q+n+1)^3+1}{2^{\secp}}$
    \end{claim}
    \begin{proof}
    First, adding the query to $\simro.\sE$ does not affect the output of the experiment since it is performed after the information needed to determine the output is already computed.
    
    Thus, to prove this claim, it suffices to show that
    \[
            \Pr_{\substack{\Hyb_2}}\left[\exists \, i \in T \, : \, (x_i^* \neq (b_i||r_i)) \wedge (y_i=\com_i)\right] \leq \frac{296(q+n+1)^3+2}{2^{\secp+1}}.
    \]    
        Consider adversary $\sB$ that runs steps~\ref{step:extComSimHyb2S1} through~\ref{step:extComSimHyb2S4} in $\Hyb_2$, and outputs $\{\com_i\}_{i \in T}, \allowbreak\{b_i||r_i\}_{i \in T}$. Note that $\sB$ does not make any queries to $\simro.\sE$. Now consider the experiment where $\sB$ is run as above, followed by running $y_i \gets \simro.\RO(b_i||r_i)$ for all $i \in T$ and then  $x_i^* \gets \simro.\sE(\com_i)$ for all $i \in T$, and outputting $1$ if $\exists \, i \in T \, : \, (x_i^* \neq (b_i||r_i)) \wedge (\simro.\RO(b_i||r_i)=\com_i)$. Applying the correctness of extraction property of \cref{impthm:extROSim}, and bounding $|T|$ by $n$, we get the required claim.
    \end{proof}
    
    \begin{claim}
        $|\Pr[\Hyb_3=1] - \Pr[\Hyb_4=1]| \leq  \frac{8qn}{2^{\secp/2}}$
    \end{claim}
    \begin{proof}
        This follows from the almost commutativity of $\simro.\sE$ and $\simro.\RO$ property of \cref{impthm:extROSim}. Indeed, since $\{\com_i\}_{i \in [n]}$ are classical strings output by the experiment after step 2, all subsequent queries to $\simro.\RO$ are independent of $\simro.\sE(\com_i)$ for any $i$, in the sense that they may operate on disjoint input and output registers. Thus, the statistical distance between the two experiments is at most $\frac{8qn}{2^{\secp/2}}$, since there are most $q$ queries to $\simro.\RO$, and $n$ queries to $\simro.\sE$.

    \end{proof}
    
    \begin{claim}
        \label{claim:extROProofHyb34}
        $\Pr[\Hyb_4=1] = \Pr[\Hyb_5=1]$
    \end{claim}
    \begin{proof}
    The only change in $\Hyb_5$ is that the variable $\mathsf{out}$ is modified and set to $\mathsf{FAIL}$ when $\mathsf{out} \not\in \{\{b^*_i\}_{i \in T}, \bot\}$. We show that whenever $\out$ is set of $\mathsf{FAIL}$, the event $\mathsf{BAD}$ occurs, which means that the output of the experiment is anyway $\mathsf{BAD}$. 
    
    Indeed, in the case of $\mathsf{FAIL}$, we know that $\out$ is not equal to $\bot$ or $\{b_i^*\}_{i \in T}$. Since $\out \neq \bot$, this means that either $\mathsf{BAD}$ occurs, or $y_i = \com_i$ for all $i \in T$. Since $\out \neq \{b_i^*\}_{i \in T}$, there must there exist $i \in T$ such that $x_i^* \neq (b_i||r_i)$. But then if $y_i = \com_i$ for all $i \in T$, the event $\mathsf{BAD}$ also occurs.

\end{proof}
\begin{claim}
$\Pr[\Hyb_6 = 1] \geq \Pr[\Hyb_5 = 1]$

\end{claim}
\begin{proof}
    This follows by observing that the distribution $\Hyb_6$ is identical to $\Hyb_5$ except that it never outputs $\mathsf{BAD}$, and therefore the probability that it outputs $1$ cannot possibly reduce.
\end{proof}
Combining all claims, we have that 
\begin{equation}
\label{eq:newa}
\Pr[\Hyb_0 = 1] \leq \Pr[\Hyb_6 = 1] + \Bigg(\frac{148(q+n+1)^3+1}{2^{\secp}} + \frac{8qn}{2^{\secp/2}}\Bigg),
\end{equation}
and by a similar argument
\begin{equation*}
\Pr[\Hyb_0 = 0] \leq \Pr[\Hyb_6 = 0] + \Bigg(\frac{148(q+n+1)^3+1}{2^{\secp}} + \frac{8qn}{2^{\secp/2}}\Bigg).
\end{equation*}

Because the output of $\Hyb_0$ and $\Hyb_6$ is a single bit, the equation above implies that
\begin{equation}
\label{eq:newb}
\Pr[\Hyb_6 = 1] \leq \Pr[\Hyb_0 = 1] + \Bigg(\frac{148(q+n+1)^3+1}{2^{\secp}} + \frac{8qn}{2^{\secp/2}}\Bigg)
\end{equation}

Combining equations~(\ref{eq:newa}) and~(\ref{eq:newb}), we have
\begin{equation}
\Big| \Pr[\Hyb_6 = 1] - \Pr[\Hyb_0 = 1] \Big| \leq \Bigg(\frac{148(q+n+1)^3+1}{2^{\secp}} + \frac{8qn}{2^{\secp/2}}\Bigg),
\end{equation}
    
    In addition, by Property~\ref{impthmstep:extROruntime} in \cref{impthm:extROSim}, the runtime of $\Sim$ is bounded by a polynomial $s(\secp,q) = O(q^2 + q\cdot n)$. Finally, by the almost commutativity of $\simro.\RO$ and $\simro.\sE$ property of \cref{impthm:extROSim}, it follows that the simulator $\Sim$ is $\nu$-commuting, with $\nu(\secp) = \frac{8}{2^{
    \secp/2}}$.
\end{proof}

\subsection{Equivocality}
\begin{theorem} \label{thm:equcom}
    \proref{fig:uccomm-protocol} is a $\mu$-equivocal bit commitment scheme with partial opening in the QROM 
    satisfying Definition~\ref{def:fequivocal}, where $\mu(\secp,q, n) = \frac{2qn^{1/2}}{2^{\secp/2}}$ \footnote{When $n = c \lambda$ for some arbitrary fixed constant $c$, then we can define $\mu_c(\secp, q) = \frac{2q(c \secp)^{1/2}}{2^{\secp/2}}$. In all our OT protocols, we will set $n$ in this manner and will assume that $\mu$ is a function of $\lambda, q$.} and where the runtime of the simulator is bounded by $s(\secp,q) = O(q^2 + \poly(\secp))$.
\end{theorem}


\begin{proof}
We construct a simulator $\SimEqu = (\SimEqu.\RO,\SimEqu.\Com,\SimEqu.\Open)$ as follows: 

\begin{enumerate}
    \item Initialize the efficient on-the-fly random oracle simulator, $\simro.\RO$, from \cref{impthm:extROSim}. For all $i \in [n]$, sample $r_i \leftarrow \zo^\secp, R^{i}_0, R^{i}_1 \gets \zo^{\secp+1}$. 
    
    \item Let $\SimEqu.\RO$ answer oracle queries of $\Adv_{\RCommit,\secp}$ and $\Adv_{\ROpen,\secp}$ using the oracle $\puncrom{H}$ which is defined as follows:
            \[   
                \puncrom{H}(x) = 
                 \begin{cases}
                  R^{i}_0 & \, \text{if }x = 0||r_i \text{ for some }i \in [n]\\
                  R^{i}_1 & \, \text{if }x = 1||r_i \text{ for some }i \in [n]\\
                  \simro.\RO(x) & \, \text{otherwise} \\ 
                 \end{cases}
            \]
            
            In an abuse of notation, we have defined $H^\bot$ using the quantum operation $\simro.\RO$. $H^\bot$ will actually be implemented by issuing a \emph{controlled} query to $\simro.\RO$ (see discussion on controlled queries in \cref{subsec:quantum-protocols}), controlled on the $x$ in input register $\cX$ not being in the set $\{b||r_i\}_{b \in \{0,1\},i \in [n]}$, and then, for each $i \in [n]$ and $b \in \{0,1\}$, implementing a controlled query to a unitary that maps $\ket{x,y} \to \ket{x,y \oplus R_{b}^i}$, controlled on the input $\cX$ register being $(b||r_i)$.
    
    \item \label{itm:romcom-comm-sec-sampleci}
    $\SimEqu.\Com$: To output commitments, for all $i \in [n]$, sample $c_i \leftarrow \zo^{\secp+1}$, set $\com_i = c_i$ and output $\{\com_i\}_{i \in [n]}$. 
    
    \item $\SimEqu.\Open$: When given input $\{b_i\}_{i \in [n]}$, output $\{b_i,r_i\}_{i \in [n]}$. 
    
    \item Let $\SimEqu.\RO$ answer oracle queries of $\sD_\secp$ using the oracle $\puncrom{H}_R$ which is defined as follows:
    
    \[   
                \puncrom{H}_R(x) = 
                 \begin{cases}
                  c_i & \, \text{if }x = b_i||r_i \text{ for some }i \in [n]\\

                  \simro.\RO(x) & \, \text{otherwise} \\ 
                 \end{cases}
            \]
            
            Note that $\puncrom{H}_R$ can be implemented in a similar way as described above.
    
\end{enumerate}

Consider then the following sequence of hybrids to prove that \proref{fig:uccomm-protocol} is an equivocal bit commitment scheme in QROM satisfying Definition~\ref{def:fequivocal} (parts in \nblue{blue} are different from previous hybrid):
\begin{itemize}
    \item $\Hyb_0$: This hybrid outputs the following distribution, which matches the real output distribution in Definition~\ref{def:fequivocal}.
        \begin{itemize}
            \item Sample oracle $H \leftarrow F_{\zo^{\secp+1} \rightarrow \zo^{\secp+1}}$.
            \item $(\brho_{1}, \{b_i\}_{i \in [n]}) \gets \Adv_{\RCommit,\secp}^{ \  H}$
            \item $(\state, \{\com_i\}_{i \in [n]}) \gets \Com^{H}(1^\secp, \{b_i\}_{i \in [n]})$
            \item $\brho_{2} \gets \Adv_{\ROpen,\secp}^{ \  H}(\brho_{1}, \{\com_i\}_{i \in [n]})$
            \item $\{b_i,r_i\}_{i \in [n]} \gets \Open^{H}(\state,[n])$
            \item Output $\sD_{\secp}^{  H}(\brho_{2},\{\com_i\}_{i \in [n]}, \{b_i,r_i\}_{i \in [n]})$.
        \end{itemize}

    \item $\Hyb_1$: This is the same as the previous hybrid except that the randomness used in $\Com$ is sampled at the beginning of the experiments and is used to define a different oracle $\puncrom{H}$. $\puncrom{H}$ is then used to answer queries of $\Adv_{\RCommit,\secp},\allowbreak\Adv_{\ROpen,\secp}$. Concretely, this hybrid outputs the following distribution.
        \begin{itemize}
            \item Sample oracle $H \leftarrow F_{\zo^{\secp+1} \rightarrow \zo^{\secp+1}}$.
            \item \nblue{For all $i \in [n]$, sample $r_i \leftarrow \zo^\secp, R^{i}_0,R^{i}_1 \gets \zo^{\secp+1}$ and define oracle $\puncrom{H}$ as:
            \[   
                \puncrom{H}(x) = 
                 \begin{cases}
                  R^{i}_0 & \, \text{if }x = 0||r_i \text{ for some }i \in [n]\\
                  R^{i}_1 & \, \text{if }x = 1||r_i \text{ for some }i \in [n]\\
                  H(x) & \, \text{otherwise} 
                 \end{cases}
            \]}
            \item $(\brho_{1},\{b_i\}_{i \in [n]}) \gets \Adv_{\RCommit,\secp}^{ \  \nblue{\puncrom{H}}}$
            \item $(\state, \{\com_i\}_{i \in [n]}) \gets \Com^{H}(1^\secp, \{b_i\}_{i \in [n]})$
            \item $\brho_{2} \gets \Adv_{\ROpen,\secp}^{ \  \nblue{\puncrom{H}}}(\brho_{1}, \{\com_i\}_{i \in [n]})$
            \item Output $\sD_{\secp}^{  H}(\brho_{2},\{\com_i\}_{i \in [n]}, \{b_i,r_i\}_{i \in [n]})$.
        \end{itemize}
    
            
    
    \item $\Hyb_2$: This is the same as previous hybrid, except that the commitments are sampled as fresh uniformly random string, and another oracle $\puncrom{H}_R$ is defined that is used to answer oracle queries of $\sD_{\secp}$. Concretely, this hybrid outputs the following distribution.
        \begin{itemize}
            \item Sample oracle $H \leftarrow F_{\zo^{\secp+1} \rightarrow \zo^{\secp+1}}$.
            \item For all $i \in [n]$, sample $r_i \leftarrow \zo^\secp, R^{i}_0,R^{i}_1 \gets \zo^{\secp+1}$ and define oracle $\puncrom{H}$ as:
            \[   
                \puncrom{H}(x) = 
                 \begin{cases}
                  R^{i}_0 & \, \text{if }x = 0||r_i \text{ for some }i \in [n]\\
                  R^{i}_1 & \, \text{if }x = 1||r_i \text{ for some }i \in [n]\\
                  H(x) & \, \text{otherwise} 
                 \end{cases}
            \]
            \item $(\brho_{1},\{b_i\}_{i \in [n]}) \gets \Adv_{\RCommit,\secp}^{ {\puncrom{H}}}$
            \item \nblue{For all $i \in [n]$, sample $c_i \gets \zo^{\secp+1}$ and set $\com_i = c_i$.}
            \item $\brho_{2} \gets \Adv_{\ROpen,\secp}^{ {\puncrom{H}}}(\brho_{1}, \{\com_i\}_{i \in [n]})$
            \item \nblue{Define oracle $\puncrom{H}_R$ as follows: \[   
                \puncrom{H}_R(x) = 
                 \begin{cases}
                  c_i & \, \text{if }x = b_i||r_i \text{ for some }i \in [n]\\
                  H(x) & \, \text{otherwise} \\ 
                 \end{cases}
            \]}
            \item Output $\sD_{\secp}^{  \nblue{\puncrom{H}_R}}(\brho_{2},\{\com_i\}_{i \in [n]}, \{b_i,r_i\}_{i \in [n]})$.
        \end{itemize}
    \item $\Hyb_3$: 
    This is the same as previous hybrid, except that the oracle $H$ is replaced by the efficient on-the-fly simulator $\simro.\RO$.
    This hybrid distribution is also the simulated output distribution in Definition~\ref{def:fequivocal}.
    Concretely, this hybrid outputs the following distribution.
    \begin{itemize}
            \item \nblue{Initialize on-the-fly simulator $\simro.\RO$.}
            \item For all $i \in [n]$, sample $r_i \leftarrow \zo^\secp, R^{i}_0,R^{i}_1 \gets \zo^{\secp+1}$ and define oracle $\puncrom{H}$ as:
            \[   
                \puncrom{H}(x) = 
                 \begin{cases}
                  R^{i}_0 & \, \text{if }x = 0||r_i \text{ for some }i \in [n]\\
                  R^{i}_1 & \, \text{if }x = 1||r_i \text{ for some }i \in [n]\\
                  \nblue{\simro.\RO(x)} & \, \text{otherwise} 
                 \end{cases}
            \]
            \item $(\brho_{1},\{b_i\}_{i \in [n]}) \gets \Adv_{\RCommit,\secp}^{{\puncrom{H}}}$
            \item For all $i \in [n]$, sample $c_i \gets \zo^{\secp+1}$, and set $\com_i = c_i$.
            \item $\brho_{2} \gets \Adv_{\ROpen,\secp}^{ {\puncrom{H}}}(\brho_{1}, \{\com_i\}_{i \in [n]})$
            \item Define oracle $\puncrom{H}_R$ as follows: \[   
                \puncrom{H}_R(x) = 
                 \begin{cases}
                  c_i & \, \text{if }x = b_i||r_i \text{ for some }i \in [n]\\
                  \nblue{\simro.\RO}(x) & \, \text{otherwise} \\ 
                 \end{cases}
            \]
            \item Output $\sD_{\secp}^{\puncrom{H}_R}(\brho_{2},\{\com_i\}_{i \in [n]}, \{b_i,r_i\}_{i \in [n]})$.
        \end{itemize}
    
\end{itemize}
Consider the following indistinguishability claims between the hybrids:
\begin{claim}
\label{claim:equivhyb01}
    $|\Pr[\Hyb_0=1]-\Pr[\Hyb_1=1]|\leq \frac{2q n^{1/2}}{2^{\secp/2}}$
\end{claim}
\begin{proof}
The two hybrids differ in the way oracle queries of $\Adv_\secp$ are answered. In $\Hyb_1$, queries of $\Adv_{\RCommit,\secp}$ and $\Adv_{\ROpen,\secp}$ are answered using oracle $\puncrom{H}$ instead of $H$ as in $\Hyb_0$.
Assume then for sake of contradiction that there exists some $\Adv = \{\brho_\secp,\Adv_\secp\}_{\secp \in \bbN}$ for which $|\Pr[\Hyb_1=1]-\Pr[\Hyb_0=1]| > \frac{2q^{3/2} n^{1/2}}{2^{\secp/2}}$. 
Fix such $\Adv$.

We derive a contradiction by relying on the One-Way to Hiding lemma (Imported Theorem~\ref{impthm:ow2h}). 
We first define oracle algorithms $\sA,\sB,\sC$.
Our goal after defining these algorithms will be to show $\sC$ succeeds in a particular event with more probability than is allowed by the statement of the lemma, which gives us a contradiction. 

\paragraph{\underline{$\sA^{O}(H,\{r_i\}_{i \in [n]})$}} 
\begin{itemize}
    \item $(\brho_{1},\{b_i\}_{i \in [n]}) \gets \Adv_{\RCommit,\secp}^{O}$
    \item $(\state, \{\com_i\}_{i \in [n]}) \gets \Com^{H}(1^\secp, \{b_i\}_{i \in [n]} ; \{r_i\}_{i \in [n]})$
    \item $\brho_{2} \gets \Adv_{\ROpen,\secp}^{O}(\brho_{1}, \{\com_i\}_{i \in [n]})$
    \item $\{b_i,r_i\}_{i \in [n]} \gets \Open^{H}(\state,[n])$
    \item Output $\sD_{\secp}^{  H}(\brho_{2},\{\com_i\}_{i \in [n]}, \{b_i,r_i\}_{i \in [n]})$.
\end{itemize}

\paragraph{\underline{$\sB^{O}(H,\{r_i\}_{i \in [n]})$}} Fix $q \coloneqq q(\lambda)$ non-uniformly as (an upper bound on) the number of oracle queries of $\Adv_\secp$, and thus also $\sA^O$.
Pick $i \leftarrow [q]$, run $\sA^{O}$ until just before the $i^{th}$ query, measure the query register and output the measurement outcome $x$. 
\fflater{\nishant{read more on expected time simulation; what if we had expected q here?}}

\paragraph{\underline{$\sC^{O}(H,\{r_i\}_{i \in [n]})$}} Run $x \gets \sB^{O}(H,\{r_i\}_{i \in [n]})$, parse $x$ as $b||r$, where $|b|=1,|r|=\secp$, and output $r$. \\

We begin by proving the following claim about $\sB$.
\begin{subclaim}
\label{subclaim:ow2haboutB}
    Given oracle $H$ and $(\br, \bR_0, \bR_1) = \{r_i,R^i_0,R^i_1\}_{i \in [n]}$, define oracle $\puncrom{H}_{\br, \bR_0, \bR_1}$ as
    \[   
                \puncrom{H}_{\br, \bR_0, \bR_1}(x) = 
                 \begin{cases}
                  R^{i}_0 & \, \text{if }x = 0||r_i \text{ for some }i \in [n]\\
                  R^{i}_1 & \, \text{if }x = 1||r_i \text{ for some }i \in [n]\\
                  H(x) & \, \text{otherwise} 
                 \end{cases}
            \]
Then,
     \[
            \Pr[x \in S_{\br}  \ \Bigg| \ \begin{array}{r}
                 H \leftarrow F_{\zo^{\secp+1} \rightarrow \zo^{\secp+1}}\\
         \forall i \in [n], r_i \leftarrow \zo^\secp, R^i_0,R^i_1 \leftarrow \zo^{\secp+1}\\
                 x \leftarrow \sB^{\puncrom{H}_{\br, \bR_0, \bR_1}}(H,\{r_i\}_{i \in [n]})\\
                 S_{\br} = \{(b||r_i)\}_{b \in \zo,i \in [n]}
            \end{array}] > \frac{n}{2^\secp}
        \]
    \end{subclaim}

    \begin{proof}
    Note that over the randomness of sampling $H, \br, \bR_0, \bR_1$, $\sA^{\puncrom{H}_{\br, \bR_0, \bR_1}}(H,\{r_i\}_{i \in [n]})$ is the experiment in $\Hyb_1$, while $\sA^{H}(H,\{r_i\}_{i \in [n]})$ is the experiment in $\Hyb_0$.
    
        For any oracle $H$, and any $\br \coloneqq \{r_i\}_{i\in [n]}, \bR_0 \coloneqq \{R^i_0\}_{i\in [n]},\bR_1 \coloneqq \{R^i_1\}_{i\in [n]}$,
        \begin{align*}
            &P_{\text{left}}^{H,\br,\bR_0,\bR_1} \coloneqq \Pr[\sA^{\puncrom{H}_{\br, \bR_0, \bR_1}}(H,\{r_i\}_{i \in [n]})=1], \quad P_{\text{right}}^{H,\br,\bR_0,\bR_1} \coloneqq \Pr[\sA^{H}(H,\{r_i\}_{i \in [n]})=1]
        \end{align*}
        This implies that     
        \begin{align*}
             &\E_{H,\br,\bR_0,\bR_1} \left[P_{\text{left}}^{H,\br,\bR_0,\bR_1}\right] = \Pr[\Hyb_1=1], \quad \E_{H,\br,\bR_0,\bR_1} \left[P_{\text{right}}^{H,\br,\bR_0,\bR_1}\right] = \Pr[\Hyb_0=1]
        \end{align*}
        Therefore,
        \begin{align}
            \E_{H,\br,\bR_0,\bR_1} \Big| P_{\text{left}}^{H,\br,\bR_0,\bR_1} - P_{\text{right}}^{H,\br,\bR_0,\bR_1} \Big| &\geq \Big| \E_{H,\br,\bR_0,\bR_1}  \left[P_{\text{left}}^{H,\br,\bR_0,\bR_1}\right] - \E_{H,\br,\bR_0,\bR_1} \left[P_{\text{right}}^{H,\br,\bR_0,\bR_1}\right] \Big|\nonumber\\
            &= \Big|\Pr[\Hyb_1=1] - \Pr[\Hyb_0=1]\Big| > \frac{2q n^{1/2}}{2^{\secp/2}}\label{eq:comow2heq1}
        \end{align}
        where the first inequality follows by Jensen's inequality and linearity of expectation.
        Also, letting $S_{\br} = \{(b||r_i)\}_{b \in \{0,1\}, i \in [n]}$, define 
        \[
            P_{\text{guess}}^{H,\br,\bR_0,\bR_1} \coloneqq \Pr[x \in S_{\br} \ | \ x \leftarrow \sB^{\puncrom{H}_{\br, \bR_0, \bR_1}}(H,\{r_i\}_{i \in [n]})].
        \]
        Invoking the one-way to hiding lemma (Imported Theorem~\ref{impthm:ow2h}), with $O_1,O_2$ set as $\puncrom{H}_{\br,\bR_0, \bR_1},H$, and noting the oracle algorithm $B$ in the lemma is exactly the same as $\sB$ in our claim, and that set $S_{\br}$ is the set of points such that $\forall x \notin S_{\br}, H(x) = \puncrom{H}_{\br,\bR_0,\bR_1}(x)$, we get
        \begin{align*}
            \forall {H,\br,\bR_0,\bR_1}, \,  P_{\text{guess}}^{H,\br,\bR_0,\bR_1}  
            &\geq \frac{\Big|P_{\text{left}}^{H,\br,\bR_0,\bR_1}-P_{\text{right}}^{H,\br,\bR_0,\bR_1}\Big|^2}{4q^2}\\
            \implies \E_{H,\br,\bR_0,\bR_1}[P_{\text{guess}}^{H,\br,\bR_0,\bR_1}]
            &\geq \E_{H,\br,\bR_0,\bR_1} \Bigg[\frac{\Big|P_{\text{left}}^{H,\br,\bR_0,\bR_1}-P_{\text{right}}^{H,\br,\bR_0,\bR_1}\Big|^2}{4q^2}\Bigg] > \frac{n}{2^{\secp}} \quad \text{(using \cref{eq:comow2heq1})}
        \end{align*}
        Therefore,
        \[
             \Pr_{H,\br,\bR_0,\bR_1,\sB}[x \in S_{\br} \ | \ x \leftarrow \sB^{\puncrom{H}_{\br, \bR_0, \bR_1}}(H,\{r_i\}_{i \in [n]})] > \frac{n}{2^{\secp}}
        \]
        as desired.
    \end{proof}
    
    \begin{subclaim}
    \label{subclaim:ow2haboutC}
    Given oracle $H$ and $(\br, \bR_0, \bR_1) = \{r_i,R^i_0,R^i_1\}_{i \in [n]}$, define oracle $\puncrom{H}_{\br, \bR_0, \bR_1}$ as
    \[   
                \puncrom{H}_{\br, \bR_0, \bR_1}(x) = 
                 \begin{cases}
                  R^{i}_0 & \, \text{if }x = 0||r_i \text{ for some }i \in [n]\\
                  R^{i}_1 & \, \text{if }x = 1||r_i \text{ for some }i \in [n]\\
                  H(x) & \, \text{otherwise} 
                 \end{cases}
            \]
    Then,
    \[
    \Pr\left[y \in \{r_i\}_{i \in [n]} \ \Bigg| \begin{array}{r}
         H \leftarrow F_{\zo^{\secp+1} \rightarrow \zo^{\secp+1}}\\
         \forall i \in [n], r_i \leftarrow \zo^\secp, R^i_0,R^i_1 \leftarrow \zo^{\secp+1}\\
         y \leftarrow \sC^{\puncrom{H}_{r, R_0, R_1}}(H,\{r_i\}_{i \in [n]})
    \end{array}\right] > {\frac{n}{2^\secp}}
    \]
\end{subclaim}
\begin{proof}

This follows from Subclaim \ref{subclaim:ow2haboutB}, and noting that for any $x \in S_{\br}$, $x$ is of the form $b||r$, where $|b|=1,|r|=\secp$ and $r \in \{r_i\}_{i \in [n]}$. 

    
\end{proof}
\noindent To complete the proof of Claim~\ref{claim:equivhyb01}, we note that by SubClaim \ref{subclaim:ow2haboutC}, it holds that
$$\Pr [y \in \{r_i\}_{i \in [n]}] > \frac{n}{2^\secp}$$
where $y$ and $\{r_i\}_{i \in [n]}$ are sampled according to the process below:
\begin{itemize}
    \item Sample oracle $H \leftarrow F_{\zo^{\secp+1} \rightarrow \zo^{\secp+1}}$.
    \item For all $i \in [n]$, sample $r_i \leftarrow \zo^\secp, R^{i}_0,R^{i}_1 \gets \zo^{\secp+1}$
    \item Sample $\iota \leftarrow [q]$ and execute the steps below until the adversary makes the $\iota^{th}$ query.
    \begin{itemize}
        \item $(\brho_{1},\{b_i\}_{i \in [n]}) \gets \Adv_{\RCommit,\secp}^{ \  \puncrom{H}_{\br,\bR_0,\bR_1}}$
        \item $\brho_{2} \gets \Adv_{\ROpen,\secp}^{ \  \puncrom{H}_{\br,\bR_0,\bR_1}}(\brho_{1}, \{H(b_i||r_i)\}_{i \in [n]})$
    \end{itemize}
    \item Measure the adversary's query register to obtain $x$, parse $x$ as $b||y$ where $|b| = 1, |y| = \secp$.
\end{itemize}
Note that the view of $\Adv_{\RCommit,\secp}$ and $\Adv_{\ROpen,\secp}$ consists of $(\puncrom{H}_{\br,\bR_0,\bR_1},\{H(b_i||r_i)\}_{i \in [n]}) \equiv (O,\{c_i\}_{i \in [n]})$ for a oracle $O$ and strings $\{c_i\}_{i \in [n]}$ that are sampled uniformly and \emph{independently} of each other and independently of $\{r_i\}_{i \in [n]}$. This means that the adversary is required to guess one out of $n$ uniform $\lambda$-bit strings $\{r_i\}_{i \in [n]}$ given uniform and independent auxiliary information.
Since this is impossible except with probability at most $\frac{n}{2^\lambda}$, we obtain a contradiction, proving our claim.

\end{proof}
    
\begin{claim} 
$\Pr[\Hyb_1=1]=\Pr[\Hyb_2=1]$
\end{claim} 
\begin{proof}
    Note that the distribution of $\{\com_i\}_{i \in [n]}$ in either hybrid is a set of uniformly independently sampled random strings, sampled independently of the oracle that is accessed by $\Adv_{\ROpen,\secp}$. Therefore, the distribution of the output of $\Adv_{\ROpen,\secp}$ in either hybrid is identical. Conditioned on this, note that the following two distributions representing the inputs/oracle $\sD_{\secp}$ has access to, are identical:
    \begin{itemize}
        \item In $\Hyb_1$, $(H,\{\com_i,b_i,r_i\}_{i \in [n]}) = (H,\{H(b_i||r_i),b_i,r_i\}_{i \in [n]})$. 
        \item In $\Hyb_2$, $(\puncrom{H}_R,\{\com_i,b_i,r_i\}_{i \in [n]}) = (\puncrom{H}_R,\{c_i,b_i,r_i\}_{i \in [n]})$, where for all $i \in [n], \puncrom{H}_R(b_i||r_i)=c_i$.
    \end{itemize}
    Since the distributions are identical, the claim then follows.
\end{proof}

\begin{claim} 
$\Pr[\Hyb_2=1]=\Pr[\Hyb_3=1]$ 
\end{claim}
\begin{proof}
Indistinguishability follows immediately from the indistinguishable simulation property of \cref{impthm:extROSim}.
\end{proof}
Combining all claims, we get $|\Pr[\Hyb_0=1]-\Pr[\Hyb_3=1]| \leq \frac{2q n^{1/2}}{2^{\secp/2}}$. In addition, note that the runtime of $\simuequ$ is bounded by $s(\secp,q) = O(q^2 + \poly(\secp))$, where the $O(q^2)$ terms comes from using $\simro.\RO$ (\cref{impthm:extROSim}).
\end{proof}
\fi

\section{The fixed basis framework: OT from entanglement}

We first obtain non-interactive OT in the shared EPR model, and then show that the protocol remains secure even when one player does the EPR pair setup.
\ifsubmission
\else
\subsection{Non-interactive OT in the shared EPR pair model}
\fi
\label{sec:epr-ot}
\begin{figure}[hbt!]
\begin{framed}
\begin{center}
\textbf{Protocol~\ref{fig:2rotepr}}
\end{center}
\vspace{-2mm}
{\bf Ingredients and parameters.}
\vspace{-2mm}
\begin{itemize}
\itemsep 0.5mm
\item Security parameter $\secp$, and constants $A,B$. Let $n = (A+B)\secp$ and $k = A\secp$. 
\item A non-interactive extractable commitment scheme $(\Com,\allowbreak\Open,\allowbreak\Rec)$, where commitments to 3 bits have size $\ell \coloneqq \ell(\secp)$.
\item A random oracle $\fsro: \{0,1\}^{n\ell} \to \{0,1\}^{\lceil\log\binom{n}{k}\rceil}$.
\item An extractor $E$ with domain $\{0,1\}^{n-k}$ which is either 
\begin{itemize}
    \item The XOR function, so $E(r_1,\dots,r_{n-k}) = \bigoplus_{i \in [n-k]} r_i$.
    \item A random oracle $\extro : \{0,1\}^{n-k} \to \{0,1\}^\secp$.
\end{itemize}
\end{itemize}
\vspace{-2mm}
{\bf Setup.}
\begin{itemize}
    \item $2n$ EPR pairs on registers $\{\cR_{i,b},\cS_{i,b}\}_{i \in [n], b \in \{0,1\}}$, where the receiver has register $\cR \coloneqq \{\cR_{i,b}\}_{i \in [n], b \in \{0,1\}}$ and the sender has register $\cS \coloneqq \{\cS_{i,b}\}_{i \in [n], b \in \{0,1\}}$.
\end{itemize}
\vspace{-2mm}
{\bf Protocol.} 
\vspace{-2mm}

\begin{itemize}

\item {\bf Receiver message.} $\mathsf{R}$, on input $b \in \{0,1\},m \in \{0,1\}^\secp$, does the following.
\begin{itemize}
\item {\bf Measurement.} Sample $\theta_1 \theta_2 \ldots \theta_n \leftarrow \{+,\times \}^n$ and for $i \in [n]$, measure registers $\cR_{i,0,},\cR_{i,1}$ in basis $\theta_i$ to obtain $r_{i,0},r_{i,1}$.

\item {\bf Measurement check.} 
    \begin{itemize}
        \item Compute $\left(\state,\{c_i\}_{i \in [n]}\right) \gets \Com\left(\{(r_{i,0}, r_{i,1}, \theta_i)\}_{i \in [n]}\right)$.
        \item Compute $T = \fsro(c_1\|\dots\|c_n)$, parse $T$ as a subset of $[n]$ of size $k$.
        \item Compute $\{(r_{i,0},r_{i,1},\theta_i),u_i\}_{i \in [T]} \gets \Open(\state,T)$.
    \end{itemize}

\item
{\bf Reorientation.} Let $\overline{T} = [n] \setminus T$, and for all $i \in \overline{T}$, set $d_i = b \oplus \theta_i$ (interpreting $+$ as 0, $\times$ as 1).

\item {\bf Sampling.} Set $x_b = E\left(\{r_{i,\theta_i}\}_{i \in \overline{T}}\right) \oplus m$, and sample $x_{1-b} \gets \{0,1\}^{\secp}$.

\item {\bf Message.} Send to $\mathsf{S}$ \[(x_0,x_1),\{c_i\}_{i \in [n]},T, \{r_{i,0},r_{i,1},\theta_i,u_i\}_{i \in [T]}, \{d_i\}_{i \in \overline{T}}.\]

\end{itemize}

\item {\bf Sender computation.} 
$\mathsf{S}$ does the following.
\begin{itemize}
\item {\bf Check Receiver Message.} Abort if any of the following fails.
\begin{itemize}
\item Check that $T = \fsro(c_1\|\dots\|c_n)$.
\item Check that $\Rec(\{c_i\}_{i \in T},\{(r_{i,0}, r_{i,1}, \theta_i),u_i\}_{i \in T}) \neq \bot$.
\item For every $i \in T$, measure the registers $\cS_{i,0},\cS_{i,1}$ in basis $\theta_i$ to obtain $r_{i,0}',r_{i,1}'$, and check that $r_{i,0} = r_{i,0}'$ and $r_{i,1} = r_{i,1}'$.

\end{itemize}

\item {\bf Output.}
For all $i \in \overline{T}$, measure the register $\cS_{i,0}$ in basis $+$ and the register $\cS_{i,1}$ in basis $\times$ to obtain $r_{i,0}',r_{i,1}'$. Output 
\[m_0 \coloneqq x_0 \oplus E\left(\{r_{i,d_i}'\}_{i \in \overline{T}}\right), m_1 \coloneqq x_1 \oplus E\left(\{r_{i,d_i \oplus 1}'\}_{i \in \overline{T}}\right).\]
\end{itemize}
\vspace{-4mm}
\end{itemize}
     \end{framed}
    \caption{Non-interactive random-sender-input OT in the shared EPR pair model.}
    \label{fig:2rotepr}
\end{figure}


\begin{theorem}\label{thm:non-interactive-security}
Instantiate \proref{fig:2rotepr} with any non-interactive commitment scheme that is \emph{correct} (\cref{def:correctness}), \emph{hiding} (\cref{def:hiding}), and \emph{extractable} (\cref{def:fextractable}). Then the following hold.
\begin{itemize}
    \item When instantiated with the XOR extractor, there exist constants $A,B$ such that \proref{fig:2rotepr} securely realizes (\cref{def:secure-realization}) $\cF_{\SROT[1]}$.
    \item When instantiated with the ROM extractor, there exist constants $A,B$ such that \proref{fig:2rotepr} securely realizes (\cref{def:secure-realization}) $\cF_{\SROT[\secp]}$.
\end{itemize}

Furthermore, letting $\secp$ be the security parameter, $q$ be an upper bound on the total number of random oracle queries made by the adversary, and using the commitment scheme from \cref{sec:eecom-construction} with security parameter $\secp_\com = 4\secp$, the following hold. 
\begin{itemize}
    \item When instantiatied with the XOR extractor and constants $A = 50$, $B= 100$, \proref{fig:2rotepr} securely realizes $\cF_{\SROT[1]}$ with $\mu_{\sR^*}$-security against a malicious receiver and $\mu_{\sS^*}$-security against a malicious sender, where \[\mu_{\sR^*} = \left(\frac{8q^{3/2}}{2^\secp}+\frac{3600\secp q}{2^{2\secp}} + \frac{148(450\secp + q + 1)^3+1}{2^{4\secp}}\right), \mu_{\sS^*} = \left(\frac{85\secp^{1/2}q}{2^{2\secp}}\right). \]This requires a total of $2(A+B)\secp = 300\secp$ EPR pairs.
    \item When instantiated with the ROM extractor and constants $A = 1050$, $B= 2160$, \proref{fig:2rotepr} securely realizes $\cF_{\SROT[\secp]}$ with $\mu_{\sR^*}$-security against a malicious receiver and $\mu_{\sS^*}$-security against a malicious sender, where \[\mu_{\sR^*} = \left(\frac{8q^{3/2} + 4\secp}{2^\secp}+\frac{77040\secp q}{2^{2\secp}} + \frac{148(9630\secp + q + 1)^3+1}{2^{4\secp}}\right), \mu_{\sS^*} = \left(\frac{197\secp^{1/2}q}{2^{2\secp}}\right). \] This requires a total of $2(A+B)\secp = 6420\secp$ EPR pairs.
\end{itemize}
\end{theorem}

Then, applying non-interactive bit OT reversal (\cref{impthm:reversal}) to the protocol that realizes $\cF_{\SROT[1]}$ immediately gives the following corollary.

\begin{corollary}
Given a setup of $300\secp$ shared EPR pairs, there exists a one-message protocol in the QROM that $O\left(\frac{q^{3/2}}{2^\secp}\right)$-securely realizes $\cF_{\RROT[1]}$.
\end{corollary}

\ifsubmission We provide the proof of \cref{thm:non-interactive-security} in \cref{sec:non-interactive-security}.\else\ifsubmission\section{Proof of non-interactive OT in shared EPR model}\label{sec:non-interactive-security}\else\fi

In this section, we provide the proof of \cref{thm:non-interactive-security}.

\begin{proof}
We will prove the part of the theorem statement that considers instantiating \proref{fig:2rotepr} with the specific commitment from \cref{sec:eecom-construction}, and note that the more general part of the theorem statement follows along the same arguments.

Let $\comro$ be the random oracle used by the commitment scheme. We treat $\comro$ and $\fsro$ (and $\extro$ in the case of the ROM extractor) as separate oracles that the honest parties and adversaries query, which is without loss of generality (see \cref{subsec:quantum-protocols}). 

\paragraph{Sender security.} First, we show security against a malicious receiver $\rcv^*$. Let $(\SimExt.\RO,\allowbreak\SimExt.\Ext)$ be the simulator for the commitment scheme (\cref{def:fextractable}) against a malicious committer. We describe a simulator for our OT protocol against a malicious receiver below. \\

\noindent $\Sim[\rcv^*]$:
\begin{itemize}
    \item Prepare $2n$ EPR pairs on registers $\cR$ and $\cS$.
    \item Initialize $\rcv^*$ with the state on register $\cR$. Answer $\fsro$ (and $\extro$) queries using the efficient on-the-fly random oracle simulator (\cref{impthm:extROSim}), and answer $\comro$ queries using $\SimExt.\RO$.
    \item Obtain $(x_0,x_1),\{c_i\}_{i \in [n]},T,\{(r_{i,0},r_{i,1},\theta_i),u_i\}_{i \in T},\{d_i\}_{i \in \overline{T}}$ from $\rcv^*$ and run \[\{(r_{i,0}^*,r_{i,1}^*,\theta_i^*)\}_{i \in [n]} \allowbreak \gets \allowbreak\SimExt.\Ext(\{c_i\}_{i \in [n]}).\] 
    \item Run the ``check receiver message'' part of the honest sender strategy, except that $\{r^*_{i,0},r^*_{i,1}\}_{i \in T}$ are used in place of $\{r_{i,0},r_{i,1}\}_{i \in T}$ for the third check. If any check fails, send $\abort$ to the ideal functionality, output $\rcv^*$'s state, and continue to answering the distinguisher's queries.
    \item Let $b \coloneqq \maj \{\theta_i^* \oplus d_i\}_{i \in \overline{T}}$. For all $i \in \overline{T}$, measure the register $\cS_{i,b \oplus d_i}$ in basis $+$ if $b \oplus d_i =0$ or basis $\times$ if $b \oplus d_i =1$ to obtain $r_i'$. Let $m_b \coloneqq x_b \oplus E(\{r_i'\}_{i \in \overline{T}})$. 
    \item Send $(b,m_b)$ to the ideal functionality, output $\rcv^*$'s state, and continue to answering the distinguisher's queries.
    \item Answer the distinguisher's $\fsro$ (and $\extro$) queries with the efficient on-the-fly random oracle simulator and $\comro$ queries with $\SimExt.\RO$.
\end{itemize}

Now, given a distinguisher $\sD$ such that $\rcv^*$ and $\sD$ make a total of at most $q$ queries combined to $\fsro$ and $\comro$ (and $\extro$), consider the following sequence of hybrids.

\begin{itemize}
    \item $\Hyb_0$: The result of the real interaction between $\rcv^*$ and $\sendr$. Using the notation of \cref{def:secure-realization}, this is a distribution over $\{0,1\}$ described by $\Pi[\rcv^*,\sD,\top]$.
    \item $\Hyb_1$: This is identical to $\Hyb_0$, except that all $\comro$ queries of $\rcv^*$ and $\sD$ are answered via the $\Sim.\RO$ interface, and  $\{(r_{i,0}^*,r_{i,1}^*,\theta_i^*)\} \gets \Sim.\Ext(\{c_i\}_{i \in [n]})$ is run after $\rcv^*$ outputs its message. The values $\{r^*_{i,0},r^*_{i,1}\}_{i \in T}$ are used in place of $\{r_{i,0},r_{i,1}\}_{i \in T}$ for the third sender check.
    \item $\Hyb_2$: The result of $\Sim[\rcv^*]$ interacting in $\widetilde{\Pi}_{\cF_{\SROT[1]}}$ (or $\widetilde{\Pi}_{\cF_{\SROT[\secp]}}$). Using the notation of \cref{def:secure-realization}, this is a distribution over $\{0,1\}$ described by $\widetilde{\Pi}_{\cF_{\SROT[1]}}[\Sim[\rcv^*],\sD,\top]$ (or $\widetilde{\Pi}_{\cF_{\SROT[\secp]}}[\Sim[\rcv^*],\sD,\top]$).
\end{itemize}
The proof of security against a malicious $\rcv^*$ follows by combining the two claims below, \cref{claim:rcv-1} and \cref{claim:rcv-2}.

\begin{claim}\label{claim:rcv-1}
\[\left|\Pr[\Hyb_0 = 1] - \Pr[\Hyb_1 = 1]\right| \leq \frac{24(A+B)\secp q}{2^{2\secp}} + \frac{148(q+3(A+B)\secp+1)^3 + 1}{2^{4\secp}}.\] 
\end{claim}

\begin{proof}
This follows by a direct reduction to extractability of the commitment scheme (\cref{def:fextractable}). Indeed, let $\Adv_\mathsf{Commit}$ be the machine that runs $\Hyb_0$ until $\rcv^*$ outputs its message, which includes $\{c_i\}_{i \in [n]}$. Let $\Adv_\mathsf{Open}$ be the machine that takes as input the rest of the state of $\Hyb_0$, which includes $T$ and the openings $\{(r_{i,0},r_{i,1},\theta_i),u_i\}_{i \in [T]}$, and outputs $T$ and these openings. Let $\sD$ be the machine that runs the rest of $\Hyb_0$ and outputs a bit.

Then, the bound follows from plugging in $\secp_\com = 4\secp$ and $n = 3(A+B)\secp$ (the number of bits committed) to the bound from \cref{thm:extractable}.

\end{proof}

\begin{claim}\label{claim:rcv-2}
For any $q \geq 4$, when $E$ is the XOR extractor and $A = 50$, $B = 100$, or when $E$ is the ROM extractor and $A = 1050, B = 2160$, \[\left|\Pr[\Hyb_1 = 1] - \Pr[\Hyb_2 = 1]\right| \leq \frac{8q^{3/2}}{2^\secp}.\]
\end{claim}

\begin{proof}
First, note that the only difference between these hybrids is that in $\Hyb_2$, the $m_{1-b}$ received by $\sD$ as part of the sender's output is sampled uniformly at random (by the ideal functionality), where $b$ is defined as $\maj \{\theta_i^* \oplus d_i\}_{i \in \overline{T}}$. Now, we introduce some notation. 

\begin{itemize}
    \item Let $\bc \coloneqq (c_1,\dots,c_n)$ be the classical commitments.
    \item Write the classical extracted values $\{(r_{i,0}^*,r_{i,1}^*,\theta^*_i)\}_{i \in [n]}$ as
    \[\bR^* \coloneqq \begin{bmatrix}r^*_{1,0} \ \dots \ r^*_{n,0} \\ r^*_{1,1} \ \dots \ r^*_{n,1}\end{bmatrix}, \btheta^* \coloneqq \begin{bmatrix}\theta^*_{1} \ \dots \ \theta^*_{n}\end{bmatrix}.\]
    \item Given any $\bR,\btheta \in \{0,1\}^{2 \times n}$, define $\ket{\bR_{\btheta}}$ as a state on $n$ two-qubit registers, 
    where register $i$ contains the vector $\ket{\bR_{i,0},\bR_{i,1}}$ prepared in the $(\btheta_i,\btheta_i)$-basis.
    \item Given $\bR,\bR^* \in \{0,1\}^{2 \times n}$ and a subset $T \subset [n]$, define $\bR_T$ to be the columns of $\bR$ indexed by $T$, and define $\Delta\left(\bR_T,\bR^*_{T}\right)$ as the fraction of columns $i \in T$ such that $(\bR_{i,0},\bR_{i,1}) \neq (\bR^*_{i,0},\bR^*_{i,1})$.
    \item For $T \subset [n]$, let $\overline{T} \coloneqq [n] \setminus T$.
    \item Given $\bR^*,\btheta^* \in \{0,1\}^{2 \times n}$, $T \subseteq [n]$, and $\delta \in (0,1)$, define \[\Pi^{\bR^*,\btheta^*,T,\delta} \coloneqq \sum_{\bR : \bR_T = \bR^*_T, \Delta\left(\bR_{\overline{T}}, \bR^*_{\overline{T}}\right) \geq \delta}\ket{\bR_{\btheta^*}}\bra{\bR_{\btheta^*}}.\]
    Intuitively, this is a projection onto ``bad'' states as defined by $\bR^*,\btheta^*,T,\delta$, i.e., states that agree with $\bR^*$ on all registers $T$ but are at least $\delta$-``far'' from $\bR^*$ on registers $\overline{T}$.

\end{itemize}

Now, consider the following projection, which has hard-coded the description of $\fsro$:

\[\Pi_{\mathsf{bad}}^\delta \coloneqq \sum_{\bc,\bR^*,\btheta^*}\ket{\bc}\bra{\bc}_{\cC} \otimes \ket{\bR^*,\btheta^*}\bra{\bR^*,\btheta^*}_{\cZ} \otimes \Pi^{\bR^*,\btheta^*,\fsro(\bc),\delta}_{\cS},\]

where $\cC$ is the register holding the classical commitments, $\cZ$ is the register holding the output of $\SimExt.\Ext$, and $\cS$ is the register holding the sender's halves of EPR pairs.

\begin{subclaim}\label{subclaim:tau}
Let \[\tau \coloneqq \sum_{\bc,\bR^*,\btheta^*}p^{(\bc,\bR^*,\btheta^*)}~\tau^{(\bc,\bR^*,\btheta^*)},\] where  \[\tau^{(\bc,\bR^*,\btheta^*)} = \ket{\bc}\bra{\bc}_{\cC} \otimes \ket{\bR^*,\btheta^*}\bra{\bR^*,\btheta^*}_\cZ \otimes \rho^{(\bc,\bR^*,\btheta^*)}_{\cS,\cX}\] is the entire state of the system, including the sender's halves of EPR pairs and the receiver's entire state in $\Hyb_1$ (equivalently also $\Hyb_2$) at the point in the experiment that is right after $\rcv^*$ outputs its message and $\SimExt.\Ext$ is run. Here, each $p^{(\bc,\bR^*,\btheta^*)}$ is the probability that the registers $\cC,\cZ$ holds the classical string $\bc,\bR^*,\btheta^*$, $\cS$ is the register holding the sender's halves of EPR pairs, and $\cX$ is a register holding the remaining state of the system, which includes the rest of the receiver's classical message and its private state. Then, 
\begin{itemize}
    \item If $A = 50, B = 100$, then $\Tr(\Pi_\mathsf{bad}^{0.25}\tau) \leq \frac{64q^3}{2^{2\secp}}.$
    \item If $A = 1050, B = 2160$, then $\Tr(\Pi_\mathsf{bad}^{0.054}\tau) \leq \frac{64q^3}{2^{2\secp}}.$
\end{itemize}
\end{subclaim}

\begin{proof}
Define $\Adv_{\rcv^*}^{\fsro}$ to be the oracle machine that runs $\Hyb_1$ until $\rcv^*$ outputs $\bc$ (and the rest of its message), then runs $\SimExt.\Ext$ to obtain $\ket{\bR^*,\btheta^*}\bra{\bR^*,\btheta^*}$, and then outputs the remaining state $\rho_{\cS,\cX}$. Consider running the measure-and-reprogram simulator $\Sim[\Adv_{\rcv^*}]$ from \cref{thm:measure-and-reprogram}, which simulates $\fsro$ queries, measures and outputs $\bc$, then receives a uniformly random subset $T \subset [n]$ of size $k$, and then continues to run $\Adv_{\rcv^*}$ until it outputs $\ket{\bR^*,\btheta^*}\bra{\bR^*,\btheta^*} \otimes \rho_{\cS,\cX}$. Letting \[\Pi_{\mathsf{bad}}^{\delta}[T] \coloneqq \sum_{\bc,\bR^*,\btheta^*}\ket{\bc}\bra{\bc}_{\cC} \otimes \ket{\bR^*,\btheta^*}\bra{\bR^*,\btheta^*}_{\cZ} \otimes \Pi^{\bR^*,\btheta^*,T,\delta}_{\cS},\] for $T \subset [n]$, \cref{thm:measure-and-reprogram} implies that \begin{align*}\Tr&\left(\Pi_\mathsf{bad}^{\delta}\tau\right) \\ &\leq (2q+1)^2 \E\left[\Tr\left(\Pi_{\mathsf{bad}}^{\delta}[T]\sigma\right): \begin{array}{r}(\bc,\state) \gets \Sim[\Adv_{\rcv^*}] \\ T \gets S_{n,k} \\ (\bR^*,\btheta^*,\rho_{\cS,\cX}) \gets \Sim[\Adv_{\rcv^*}](T,\state)\end{array}\right],\end{align*} where

\[\sigma = \ket{\bc}\bra{\bc}_\cC  \otimes \ket{\bR^*,\btheta^*}\bra{\bR^*,\btheta^*}_\cZ \otimes \rho_{\cS,\cX},\]
and $S_{n,k}$ is the set of all subsets of $[n]$ of size $k$.

Now, recall that the last thing that $\Adv_{\rcv^*}$ does in $\Hyb_1$ is run $\SimExt.\Ext$ on $\bc$ to obtain $(\bR^*,\btheta^*)$. Consider instead running $\SimExt.\Ext$ on $\bc$ immediately after $\Sim[\Adv_{\rcv^*}]$ outputs $\bc$. Note that $\SimExt.\Ext$ only operates on the register holding $\bc$ and its own private state used for simulating $\comro$, so since $\com$ has a $\frac{8}{2^{{\secp_\com}/2}}$-commuting simulator (\cref{def:fextractable}), we have that, 

\begin{align*}
\Tr&\left(\Pi_\mathsf{bad}^{\delta}\tau\right)\\ &\leq (2q+1)^2\left(\E\left[\Tr\left(\Pi_{\mathsf{bad}}^{\delta}[T]\sigma\right): \begin{array}{r}(\bc,\state) \gets \Sim[\Adv_{\rcv^*}] \\ (\bR^*,\btheta^*) \gets \SimExt.\Ext(\bc) \\ T \gets S_{n,k} \\ \rho_{\cS,\cX} \gets \Sim[\Adv_{\rcv^*}](T,\state)\end{array}\right] + \frac{8q}{2^{2\secp}}\right) \\ &\coloneqq (2q+1)^2\epsilon + \frac{8q(2q+1)^2}{2^{2\secp}},
\end{align*}
where 
\[\sigma = \ket{\bc}\bra{\bc}_\cB  \otimes \ket{\bR^*,\btheta^*}\bra{\bR^*,\btheta^*}_\cZ \otimes \rho_{\cS,\cX},\]
and where we denote the expectation inside the parantheses by $\epsilon$, and we plugged in $\secp_\com = 4\secp$. 

Towards bounding $\epsilon$, we now consider the following quantum sampling game.
\begin{itemize}
    \item Fix a state on register $\cS$ (and potentially other registers of arbitrary size), where $\cS$ is split into $n$ registers $\cS_1,\dots,\cS_n$ of dimension 4, and fix $\bR^*,\btheta^* \in \{0,1\}^{2 \times n}$.
    \item Sample $T \subset [n]$ as a uniformly random subset of size $k$.
    \item For each $i \in T$, measure registers $\cS_i$ in the $(\btheta^*_{i},\btheta^*_{i})$-basis to obtain a matrix $\bR_{T} \in \{0,1\}^{2 \times |T|}$, and output $\Delta\left(\bR_{T},\bR_{T}^*\right)$. 

\end{itemize}

Next, we argue that $\epsilon$ is bounded by the quantum error probability $\epsilon^{\delta}_{\mathsf{quantum}}$ (\cref{def:quantum-error-probability}) of the above game.
This corresponds to the trace distance between the initial state on register $\cS$ and an ``ideal'' state (as defined in \cref{def:quantum-error-probability}). This ideal state is supported on vectors $\ket{\bR_{\btheta^*}}$ such that $\Delta(\bR_{\overline{T}},\bR^*_{\overline{T}}) < \Delta(\bR_T,\bR^*_T) + \delta$. In particular, for any $\ket{\bR_{\btheta^*}}$ with $\Delta(\bR_T,\bR^*_T) = 0$ in the support of the ideal state, it holds that $\Delta(\bR_{\overline{T}},\bR^*_{\overline{T}}) < \delta$. Thus, this ideal state is orthogonal to the subspace $\Pi^{\bR^*,\btheta^*,T,\delta}_\cS$, and so it follows that $\epsilon$ is bounded by $\epsilon^{\delta}_{\mathsf{quantum}}$.


Thus, by \cref{impthm:error-probability}, $\epsilon$ is then bounded by $\sqrt{\epsilon^\delta_{\mathsf{classical}}}$, where $\epsilon^\delta_{\mathsf{classical}}$ is the \emph{classical} error probability (\cref{def:classical-error-probability}) of the following sampling game.

\begin{itemize}
    \item Let $\bR \in \{0,1\}^{2 \times n}$ be an arbitrary matrix.
    \item Sample a uniformly random subset $T \subset [n]$ of size $k$.
    \item Let $\delta^*$ be the fraction of columns $(\bR_{i,0},\bR_{i,1})$ for $i \in T$ that are non-zero, and output $\delta^*$.
\end{itemize}

The classical error of the above game is the probability that $\geq \delta^*+\delta$ of the columns $(\bR_{i,0},\bR_{i,1})$ for $i \in \overline{T}$ are non-zero. Using the analysis in \cref{appsubsec:CK88sampling},
we can bound this probability by $2\exp(-2(1-k/n)^2\delta^2k)$.

\begin{itemize}
    \item For $\delta = 0.25$, this probability is bounded by
    
    \[2\exp(-2(0.25)^2(1-A/(A+B))^2A) < 2^{-4\secp-1},\] for $A = 50, B= 100$.
    Thus, we can bound $\epsilon^\delta_{\mathsf{classical}}$ by $2/2^{4\secp}$ and thus $\epsilon$ by $\sqrt{2}/2^{2\secp}$.
    \item For $\delta = 0.054$, this probability is bounded by
    
    \[2\exp(-2(0.054)^2(1-A/(A+B))^2A)  < 2^{-4\secp-1},\] for $A = 1050, B = 2160$.
    Thus, we can bound $\epsilon_{\mathsf{classical}}$ by $2/2^{4\secp}$ and thus $\epsilon$ by $\sqrt{2}/2^{2\secp}$.
\end{itemize}

Summarizing, we have that in either case, 
\[\Tr\left(\Pi_\mathsf{bad}^{\delta}\tau\right) \leq \frac{\sqrt{2}(2q+1)^2 + 8q(2q+1)^2}{2^{2\secp}} \leq \frac{64q^3}{2^{2\secp}},\] for $q \geq 4$.

\end{proof}

Thus, by gentle measurement (\cref{lemma:gentle-measurement}), the $\tau$ defined in \cref{subclaim:tau} is within $\frac{8q^{3/2}}{2^\secp}$ trace distance of a state $\tau_{\mathsf{good}}$ in the image of $\bbI - \Pi_\mathsf{bad}^{0.25}$ if $A = 50,B = 100$ and in the image of  $\bbI - \Pi_\mathsf{bad}^{0.054}$ if $A = 1050,B = 2160$. 

For readability, we note that 

\[\bbI-\Pi_{\mathsf{bad}}^\delta = \sum_{\bc,\bR^*,\btheta^*}\ket{\bc}\bra{\bc}_\cC \otimes \ket{\bR^*,\btheta^*}\bra{\bR^*,\btheta^*}_\cZ \otimes \left(\bbI - \Pi^{\bR^*,\btheta^*,\fsro(\bc),\delta}\right)_\cS,\]

where for any $T$,

\[\bbI - \Pi^{\bR^*,\btheta^*,T,\delta} = \sum_{\bR:(\bR_T \neq \bR^*_T) \vee (\Delta(\bR_{\overline{T}},\bR^*_{\overline{T}}) < \delta) }\ket{\bR_{\btheta^*}}\bra{\bR_{\btheta^*}}.\]

We require the following two sub-claims to complete the proof of \cref{claim:rcv-2}.

\begin{subclaim}
If $E$ is the XOR extractor, then conditioned on $\tau$ being in the image of $\bbI - \Pi_\mathsf{bad}^{0.25}$, it holds that \[\Pr[\Hyb_1 = 1] = \Pr[\Hyb_2 = 1].\]
\end{subclaim}

\begin{proof}
First note that if the $T$ sent by $\rcv^*$ to the sender is not equal to $\fsro(\bc)$, then the sender will abort, and the hybrids are perfectly indistinguishable. So it suffices to analyze the state $\tau$ conditioned on the register that contains $T$ being equal to $\fsro(\bc)$. 

Now, if $\tau$ is in $\bbI - \Pi_\mathsf{bad}^{0.25}$, it must be the case that the register $\cS$ is in the image of $\bbI - \Pi^{\bR^*,\btheta^*,T,0.25}$, where $\bR^*,\btheta^*$ were output by $\SimExt.\Ext$. Recall that the sender aborts if the positions measured in $T$ are not equal to $\bR^*_T$, and in this case the hybrids would be perfectly indistinguishable. Thus, we can condition on the sender not aborting, which, by the definition of $\bbI - \Pi^{\bR^*,\btheta^*,T,0.25}$ implies that register $\cS_{\overline{T}}$ is supported on vectors $\ket{\left(\bR_{\overline{T}}\right)_{\btheta^*}}$ such that $\Delta(\bR_{\overline{T}},\bR^*_{\overline{T}}) < 0.25$.

Now, to obtain $m_{1-b}$, the sender measures register $\cS_{i,d_i \oplus b \oplus 1}$ in basis $d_i \oplus b \oplus 1$ for each $i \in \overline{T}$ to obtain a string $r' \in \{0,1\}^{n-k}$. Then, $m_{1-b}$ is set to $E(r')$. Since $b$ is defined as $\maj \{\btheta_i^* \oplus d_i\}_{i \in \overline{T}}$ in $\Hyb_2$, at least $(n-k)/2$ of the bits $r'_i$ are obtained by measuring in $1 \oplus \btheta^*_i$. Let $M \subset \overline{T}$ be this set of size at least $(n-k)/2$, and define $\br^* \in \{0,1\}^n$ such that $\br^*_i = \bR^*_{i,d_i \oplus b \oplus 1}$ . We know from above that the register $\cS_M$ is supported on vectors $\ket{\left(\br_M\right)_{\btheta^*}}$ for $\br_M$ such that $\Delta(\br_M,\br^*_M) < 0.5$. 
Thus, recalling that each of these states is measured in the basis $1 \oplus \btheta^*_i$, we can appeal to \cref{lemma:XOR-extractor} (with an appropriate change of basis) to show that $m_{1-b}$ is perfectly uniformly random from $\rcv^*$'s perspective, completing the proof.

\end{proof}

\begin{subclaim}
If $E$ is the ROM extractor and $B \geq 326, q \geq 4$, then conditioned on $\tau$ being in the image of $\bbI - \Pi_\mathsf{bad}^{0.054}$, it holds that \[|\Pr[\Hyb_1 = 1] = \Pr[\Hyb_2 = 1]| \leq \frac{4q}{2^\secp}.\]
\end{subclaim}

\begin{proof}
This follows the same argument as the above sub-claim, until we see that there are $(n-k)/2$ qubits of $\cS$ that are measured in basis $1\oplus\btheta^*_M$, and that the state on these qubits is supported on vectors $\ket{\left(\br_M\right)_{\btheta^*}}$ for $\br_M$ such that $\Delta(\br_M,\br^*_M) < 0.108$. We can then apply \cref{thm:ROM-extractor} with random oracle input size $n-k$, register $\cX$ size $(n-k)/2$, and $|L| \leq 2^{h_b(0.108)(n-k)/2}$. Note that, when applying this theorem, we are fixing any outcome of the $(n-k)/2$ bits of the random oracle input that are measured in $\btheta^*$, and setting register $\cX$ to contain the $(n-k)/2$ registers that are measured in basis $1 \oplus \btheta^*$. This gives a bound of \[\frac{4q2^{h_b(0.108)(n-k)/2}}{2^{(n-k)/4}} = \frac{4q}{2^{(n-k)(\frac{1}{4} - \frac{1}{2}h_b(0.108))}} = \frac{4q}{2^{B\secp(\frac{1}{4} - \frac{1}{2}h_b(0.108))}} \leq \frac{4q}{2^\secp},\] for $B \geq 326$.

\end{proof}
This completes the proof of \cref{claim:rcv-2}.
\end{proof}


\paragraph{Receiver security.} Next, we show security against a malicious sender $\sendr^*$. During the proof, we will use an efficient quantum random oracle ``wrapper'' algorithm $W[(x,z)]$ that provides an interface between any quantum random oracle simulator, such as the on-the-fly simulator (\cref{impthm:extROSim}), and the machine querying the random oracle. The wrapper will implement a controlled query to the actual random oracle simulator, controlled on the input $\cX$ register not being equal to $x$. Then, it will implement a controlled query to a unitary that maps $\ket{x,y} \to \ket{x,y \oplus z}$, controlled on the input $\cX$ register being equal to $x$. The effect of this wrapper is that the oracle presented to the machine is the oracle $H$ simulated by the simulator, but with $H(x)$ reprogrammed to $z$. \\

\noindent $\Sim[\sendr^*]:$
\begin{itemize}
    \item Query the ideal functionality with $\bot$ and obtain $m_0,m_1$.
    \item Sample $T$ as a uniformly random subset of $[n]$ of size $k$, sample $d_i \gets \{0,1\}$ for each $i \in \overline{T}$, and sample $\theta_i \gets \{+,\times\}$ for each $i \in T$.
    \item For each $i \in [n]$, sample $r_{i,0},r_{i,1} \gets \{0,1\}$ and prepare BB84 states $\ket{\psi_{i,0}},\ket{\psi_{i,1}}$ as follows.
    \begin{itemize}
        \item If $i \in T$, set $\ket{\psi_{i,0}} = \ket{r_{i,0}}_{\theta_i}, \ket{\psi_{i,1}} = \ket{r_{i,1}}_{\theta_i}$.
        \item If $i \in \overline{T}$, set $\ket{\psi_{i,0}} = \ket{r_{i,0}}_{+}, \ket{\psi_{i,1}} = \ket{r_{i,1}}_{\times}$.
    \end{itemize}
    \item For each $i \in T$, let $e_i \coloneqq (r_{i,0},r_{i,1},\theta_i)$ and for each $i \in \overline{T}$, let $e_i \coloneqq (0,0,0)$. Compute $(\state,\{c_i\}_{i \in [n]}) \gets \Com(\{e_i\}_{i \in [n]})$ and $\{u_i\}_{i \in T} \gets \Open(\state,T)$.
    \item Set $x_0 \coloneqq E(\{r_{i,d_i}\}_{i \in \overline{T}}) \oplus m_0$ and $x_1 \coloneqq E(\{r_{i,d_i \oplus 1}\}_{i \in \overline{T}}) \oplus m_1$ (where if $E$ is the ROM extractor, this is accomplished via classical queries to an on-the-fly random oracle simulator for $\extro$).
    \item Run $\sendr^*$ on input $(x_0,x_1),\{c_i\}_{i \in [n]},T,\{r_{i,0},r_{i,1},\theta_i,u_i\}_{i \in T},\{d_i\}_{i \in \overline{T}}$, $\{\ket{\psi_{i,b}}\}_{i \in [n], b \in \{0,1\}}$. Answer $\comro$ queries using the on-the-fly random oracle simulator, answer $\fsro$ queries using the on-the-fly random oracle simulator wrapped with $W[\{c_i\}_{i \in [n]},T]$, and if $E$ is the ROM extractor, answer $\extro$ queries using the on-the-fly random oracle simulator. Output $\sendr^*$'s final state and continue to answering the distinguisher's random oracle queries.
\end{itemize}

Now, given a receiver input $b \in \{0,1\}$, and distinguisher $\sD$ such that $\sendr^*$ and $\sD$ make a total of at most $q$ queries combined to $\fsro$ and $\comro$ (and $\extro$), consider the following sequence of hybrids.

\begin{itemize}
    \item $\Hyb_0$: The result of the real interaction between $\rcv(b)$ and $\sendr^*$. Using the notation of \cref{def:secure-realization}, this is a distribution over $\{0,1\}$ desrcibed by $\Pi[\sendr^*,\sD,b]$.
    \item $\Hyb_1$: This is the same as the previous hybrid except that $T$ is sampled uniformly at random as in the simulator, and $\fsro$ queries are answered with the wrapper $W[(\{c_i\}_{i \in [n]},T)]$.
    \item $\Hyb_2$: This is the same as the previous hybrid except that the messages $\{(r_{i,0},r_{i,1},\theta_i)\}_{i \in \overline{T}}$ are replaced with $(0,0,0)$ inside the commitent.
    \item $\Hyb_3$: The result of $\Sim[\sendr^*]$ interacting in $\widetilde{\Pi}_{\cF_{\SROT[1]}}$ (or $\widetilde{\Pi}_{\cF_{\SROT[\secp]]}}$). Using the notation of \cref{def:secure-realization}, this is a distribution over $\{0,1\}$ described by $\widetilde{\Pi}_{\cF_{\SROT[1]}}[\Sim[\sendr^*],\sD,b]$ (or $\widetilde{\Pi}_{\cF_{\SROT[\secp]}}[\Sim[\sendr^*],\sD,b]$).
\end{itemize}

The proof of security against a malicious $\sendr^*$ follows by combining the following three claims.

\begin{claim}
\[\Pr[\Hyb_0 = 1] = \Pr[\Hyb_1 = 1].\]
\end{claim}

\begin{proof}
These hybrids are identically distributed, since $\fsro$ is a random oracle and $T$ is uniformly random in $\Hyb_1$.
\end{proof}

\begin{claim}
\[|\Pr[\Hyb_1 = 1] - \Pr[\Hyb_2 = 1]| \leq \frac{4q\sqrt{3(A+B)\secp}}{2^{2\secp}}.\]
\end{claim}

\begin{proof}
This follows directly from the hiding of the commitment scheme (\cref{def:hiding}), which is implied by its equivocality (see \cref{subsec:com-defs}). To derive the bound, we plug in $\secp_\com = 4\secp$ and $n = 3(A+B)\secp$ to the bound from \cref{thm:equcom}.
\end{proof}

\begin{claim}
\[\Pr[\Hyb_2 = 1] = \Pr[\Hyb_3 = 1].\]
\end{claim}

\begin{proof}

First, note that one difference in how the hybrids are specified is that in $\Hyb_2$, the receiver samples $x_{1-b}$ uniformly at random, while in $\Hyb_3$, $x_{1-b}$ is set to $E(\{r_{i,d_i \oplus b \oplus 1}\}_{i \in \overline{T}}) \oplus m_{1-b}$. However, since $m_{1-b}$ is sampled uniformly at random by the functionality, this is an equivalent distribution.

Thus, the only difference between these  these hybrids is the basis in which the states on registers $\{\cS_{i,d_i \oplus b \oplus 1}\}_{i \in \overline{T}}$ are prepared (which are the registers $\{\cS_{i,\theta_i \oplus 1}\}_{i \in \overline{T}}$ in $\Hyb_2$). Indeed, note that in $\Hyb_2$, the state on register $\cS_{i,d_i \oplus b_i \oplus 1}$ is prepared by having the receiver measure their corresponding half of an EPR pair (register $\cR_{i,d_i \oplus b_i \oplus 1}$) in basis $\theta_i = d_i \oplus b$, while in $\Hyb_3$, this state is prepared by sampling a uniformly random bit and encoding it in the basis $d_i \oplus b_i \oplus 1$. However, these sampling procedures both produce a maximally mixed state on register $\cS_{i,d_i \oplus b \oplus 1}$, and thus these hybrids are equivalent.

\end{proof}

This completes the proof of the theorem.
\end{proof}\fi

\ifsubmission In \cref{sec:two-round}, we analyze a variant of this protocol where we allow the receiver to set up the entanglement, yielding a two-round protocol without setup. \else \ifsubmission\section{Two-round OT without setup}\label{sec:two-round}\else\subsection{Two-round OT without setup}
\label{subsec:two-round-OT}
\fi

In this section, we analyze a variant of the EPR-based protocol (\cref{fig:2rotepr}) where we allow the sender to generate the EPR setup. That is, an honest sender will prepare $2n$ EPR pairs between registers $\cR$ and $\cS$, and send $\cR$ to the reciever, while a malicious sender may prepare and send an arbitary state.

Thus, the resulting protocol is a two-round protocol without setup. We show that it securely realizes the $\cF_{\SROT[\secp]}$ OT ideal functionality, where the receiver can send chosen inputs $(b, m)$ to the functionality and the functionality outputs to the sender random $(m_0,m_1)$ such that $m_b = m$.

\begin{theorem}\label{thm:two-round-security}

Consider instantiating the two-round variant of \proref{fig:2rotepr} with any non-interactive commitment scheme that is \emph{correct} (\cref{def:correctness}), \emph{equivocal} (\cref{def:fequivocal}), and \emph{extractable} (\cref{def:fextractable}). Then the following hold.
\begin{itemize}
    \item When instantiated with the XOR extractor, there exist constants $A,B$ such that the two-round variant of \proref{fig:2rotepr} securely realizes (\cref{def:secure-realization}) $\cF_{\SROT[1]}$.
    \item When instantiated with the ROM extractor, there exist constants $A,B$ such that the two-round variant of \proref{fig:2rotepr} securely realizes (\cref{def:secure-realization}) $\cF_{\SROT[\secp]}$.
\end{itemize}

Letting $\secp$ be the security parameter, $q$ be an upper bound on the total number of random oracle queries made by the adversary, and using the commitment scheme from \cref{sec:eecom-construction} with security parameter $\secp_\com = 4\secp$, the following hold. 
\begin{itemize}
    \item When instantiatied with the XOR extractor and constants $A = 50$, $B= 100$, the two-round variant of \proref{fig:2rotepr} securely realizes $\cF_{\SROT[1]}$ with $\mu_{\cR^*}$-security against a malicious receiver and $\mu_{\cS^*}$-security against a malicious sender, where \[\mu_{\cR^*} = \left(\frac{8q^{3/2}}{2^\secp}+\frac{3600\secp q}{2^{2\secp}} + \frac{148(450\secp + q + 1)^3+1}{2^{4\secp}}\right), \mu_{\cS^*} = \left(\frac{85\secp^{1/2}q}{2^{2\secp}}\right). \]This requires a total of $2(A+B)\secp = 300\secp$ EPR pairs.
    \item When instantiated with the ROM extractor and constants $A = 1050$, $B= 2160$, the two-round variant of \proref{fig:2rotepr} securely realizes $\cF_{\SROT[\secp]}$ with $\mu_{\cR^*}$-security against a malicious receiver and $\mu_{\cS^*}$-security against a malicious sender, where \[\mu_{\cR^*} = \left(\frac{8q^{3/2} + 4\secp}{2^\secp}+\frac{77040\secp q}{2^{2\secp}} + \frac{148(9630\secp + q + 1)^3+1}{2^{4\secp}}\right), \mu_{\cS^*} = \left(\frac{197\secp^{1/2}q}{2^{2\secp}}\right). \] This requires a total of $2(A+B)\secp = 6420\secp$ EPR pairs.
\end{itemize}
\end{theorem}

Then, applying non-interactive bit OT reversal (\cref{impthm:reversal}) to the protocol that realizes $\cF_{\SROT[1]}$ immediately gives the following corollary.

\begin{corollary}
Given a setup of $300\secp$ shared EPR pairs, there exists a one-message protocol in the QROM that $O\left(\frac{q^{3/2}}{2^\secp}\right)$-securely realizes $\cF_{\RROT[1]}$.
\end{corollary}

\begin{proof}

Security against a malicious receiver remains the same as \cref{thm:non-interactive-security}, so we only show security against a malicious sender. Let $\sendr^*$ be a malicious sender. Let $(\SimEqu.\RO,\SimEqu.\Com,\SimEqu.\Open)$ be the equivocal simulator for the commitment scheme (\cref{def:fequivocal}).\\

\noindent $\Sim[\sendr^*]:$
\begin{itemize}
    \item Run $\sendr^*$. Answer $\fsro$ (and $\extro$) queries using the efficient on-the-fly random oracle simulator, and answer $\comro$ queries using $\SimEqu.\RO$. Eventually, $\sendr^*$ outputs a state on register $\cR = (\cR_{1,0},\cR_{1,1},\dots,\cR_{n,0},\cR_{n,1})$.
    \item Query the ideal functionality with $\bot$ and obtain $m_0,m_1$.
    \item Run the following strategy on behalf of the receiver.
    \begin{itemize}
        \item Compute $\{c_i\}_{i \in [n]} \gets \SimEqu.\Com$.
        \item Compute $T = \fsro(c_1\|\dots\|c_n)$ and parse $T$ as a subset of $[n]$ of size $k$.
        \item For each $i \in T$, sample $\theta \gets \{+,\times\}$ and measure registers $\cR_{i,0}$ and $\cR_{i,1}$ in basis $\theta_i$ to obtain $r_{i,0},r_{i,1}$.
        \item Compute $\{u_i\}_{i \in T} \gets \SimEqu.\Open(\{r_{i,0},r_{i,1},\theta_i\}_{i \in T})$.
        \item For each $i \in \overline{T}$, measure register $\cR_{i,0}$ in basis $+$ and register $\cR_{i,1}$ in basis $\times$ to obtain $r_{i,0},r_{i,1}$. 
        \item For each $i \in \overline{T}$, sample $d_i \gets \{0,1\}$. Compute $x_0 \coloneqq E(\{r_{i,d_i}\}_{i \in \overline{T}}) \oplus m_0, x_1 \coloneqq E(\{r_{i,d_i \oplus 1}\}_{i \in \overline{T}}) \oplus m_1$.
    \end{itemize}
    \item Send $(x_0,x_1),\{c_i\}_{i \in [n]},T,\{(r_{i,0},r_{i,1},\theta_i),u_i\}_{i \in T}, \{d_i\}_{i \in \overline{T}}$ to $\sendr^*$, and run $\sendr^*$ until it outputs a final state, answering $\fsro$ (and $\extro$) queries using the efficient on-the-fly random oracle simulator and $\comro$ queries using $\SimEqu.\RO$. Output $\sendr^*$'s final state.
    \item Answer the distinguisher's $\fsro$ (and $\extro$) queries using the efficient on-the-fly random oracle simulator and $\comro$ queries using $\SimEqu.\RO$.
\end{itemize}

Now, given a distinguisher $\sD$ such that $\sendr^*$ and $\sD$ make a total of at most $q$ queries combined to $\fsro$ and $\comro$, and a receiver input $(b,m_b)$, consider the following sequence of hybrids.

\begin{itemize}
    \item $\Hyb_0$: The result of the real interaction between $\sendr^*$ and $\rcv$. Using the notation of \cref{def:secure-realization}, this is a distribution over bits described by $\Pi[\sendr^*,\sD,(b,m_b)]$.
    \item $\Hyb_1$: Answer all $\comro$ queries of $\sendr^*$ and $\sD$ with $\SimEqu.\RO$. Run the honest receiver strategy, except $\{c_i\}_{i \in [n]} \gets \SimEqu.\Com$, and $\{u_i\} \gets \SimEqu.\Open(\{(r_{i,0},r_{i,1},\theta_i)\}_{i \in T})$.
    \item $\Hyb_2$: The result of $\Sim[\sendr^*]$ interacting in $\widetilde{\Pi}_{\cF_{\SROT[1]}}$ (or $\widetilde{\Pi}_{\cF_{\SROT[\secp]]}}$). Using the notation of \cref{def:secure-realization}, this is a distribution over bits described by $\widetilde{\Pi}_{\cF_{\SROT[1]}}[\Sim[\sendr^*],\sD,(b,m_b)]$ (or $\widetilde{\Pi}_{\cF_{\SROT[\secp]}}[\Sim[\sendr^*],\sD,(b,m_b)]$).
\end{itemize}

\begin{claim}
\[|\Pr[\Hyb_0 = 1] - \Pr[\Hyb_1 = 1]| \leq \frac{2q\sqrt{3(A+B)\secp}}{2^{2\secp}}.\] 
\end{claim}
\begin{proof}
This follows by a direct reduction to equivocality of the commitment scheme (\cref{def:fequivocal}). Indeed, let $\Adv_{\mathsf{RCommit}}$ be the machine that runs $\Hyb_0$ until $\sendr^*$ outputs its message on register $\cR$ and $\sR$ runs the $\mathbf{Measurement}$ portion of its honest strategy to produce $\{r_{i,0},r_{i,1},\theta_i\}_{i \in [n]}$. Let $\Adv_{\mathsf{ROpen}}$ be the machine computes $T = \fsro(c_1\|\dots\|c_n)$. Let $\sD$ be the machine that runs the rest of $\Hyb_0$, from the $\mathbf{Reorientation}$ portion of its honest receiver's strategy to the final bit output by the distinguisher.

Then, plugging in $\secp_\com = 4\secp$ and $n = 3(A+B)\secp$ to \cref{thm:equcom} gives the bound in the claim.
\end{proof}

\begin{claim}
\[\Pr[\Hyb_1 = 1] = \Pr[\Hyb_2 = 1].\]
\end{claim}

\begin{proof}

First, note that one difference in how the hybrids are specified is that in $\Hyb_1$, the receiver samples $x_{1-b}$ uniformly at random, while in $\Hyb_2$, $x_{1-b}$ is set to $E(\{r_{i,d_i \oplus b \oplus 1}\}_{i \in \overline{T}}) \oplus m_{1-b}$. However, since $m_{1-b}$ is sampled uniformly at random by the functionality, this is an equivalent distribution.

Then, the only difference between these these hybrids is the basis in which the states on registers $\{\cR_{i,d_i \oplus b \oplus 1}\}_{i \in \overline{T}}$ are measured (which are the registers $\{\cR_{i,\theta_i \oplus 1}\}_{i \in \overline{T}}$ in $\Hyb_1$). Indeed, since the resulting bits $r_{i,d_i \oplus b \oplus 1}$ are unused by the receiver in $\Hyb_1$, and masked by $m_{1-b}$ in $\Hyb_2$, they are independent of the sender's view. Thus, measuring them in different bases has no effect on the sender's view, and so the hybrids are identical.
%
%
%
\end{proof}
This completes the proof of the claim, as desired.
\end{proof}

\fi

\ifsubmission
\section{The fixed basis framework: OT without entanglement}
\label{sec:bb84-3r-ot}

\cref{fig:qot-3r-bb84} formalizes our 3 round chosen-input OT without entanglement or setup.

\begin{figure}[H]
\begin{framed}
\begin{center}
\textbf{Protocol~\ref{fig:qot-3r-bb84}}
\end{center}
{\bf Ingredients / parameters / notation.}
\begin{itemize}
    \itemsep 0.5mm
    \item Security parameter $\secp$ and constants $A,B$. Let $k = A\secp, n = (A+B)\secp$. 
    \item For bits $(x,\theta)$, let $\ket{x}_\theta$ denote $\ket{x}$ if $\theta = 0$, and $(\ket{0} + (-1)^x\ket{1})/\sqrt{2}$ if $\theta = 1$.
    \item A non-interactive extractable and equivocal commitment $(\Com,\allowbreak\Open,\allowbreak\Rec)$, where commitments to 3 bits have size $\ell \coloneqq \ell(\secp)$.
    \item A random oracle $\fsro: \zo^{n\ell} \to  \zo^{\lceil\log\binom{n}{k}\rceil}$, and
    a universal hash function family $h : \zo^{p(\secp)} \times \zo^{\leq B\secp} \rightarrow \zo^\secp$.
\end{itemize}

{\bf Sender Input:} Messages $m_0,m_1 \in \zo^\secp$.     {\bf Receiver Input:} Choice bit b.

\begin{enumerate}
\item {\bf Sender Message.}
    $\sendr$ samples strings $r^0 \leftarrow \zo^n,r^1 \leftarrow \zo^n$, a random subset $U \subset [n]$ of size $k$, and for $i \in U$, it samples $b_i \leftarrow \zo$ uniformly at random.
    It computes state $\ket{\psi} = \ket{\psi}_1 \ldots \ket{\psi}_n$ as follows, and sends it to $\recv$: for $i \in U$, $\ket{\psi}_i = (\ket{r^0_i}_{b_i}, \ket{r^1_i}_{b_i})$ and for $i \in [n] \setminus U$, $\ket{\psi}_i = (\ket{r^0_i}_{0}, \ket{r^1_i}_{1})$.

\item {\bf Receiver Message.} $\recv$ does the following.
\begin{itemize}
    \item Choose $\widehat{\theta} \leftarrow \zo^{n}$ and measure the $i^{th}$ pair of qubits in basis $\widehat{\theta}_i$ to obtain $\widehat{r_i}^0, \widehat{r_i}^1$.
    \item {\bf Measurement Check Message.} 
        \begin{itemize}
            \item Compute $\left(\state,\{c_i\}_{i \in [n]}\right) \gets \Com\left(\{(\widehat{r_i}^0, \widehat{r_i}^1, \widehat{\theta}_i)\}_{i \in [n]}\right)$.
            \item Compute $T = \fsro(c_1||\dots||c_n)$, parse $T$ as a subset of $[n]$ of size $k$.
            \item Compute $\{(\widehat{r_i}^0, \widehat{r_i}^1, \widehat{\theta}_i),u_i\}_{i \in T} \gets \Open(\state,T)$.
    \end{itemize}

    \item {\bf Reorientation.} Let $\overline{T} \coloneqq [n] \setminus T$, and for all $i \in \overline{T}$, set $d_i = b \oplus \widehat{\theta}_i$.

    \item {\bf Message.} Send to $\sendr$ the values $\{c_i\}_{i \in [n]},T, \{\widehat{r_i}^0, \widehat{r_i}^1, \widehat{\theta}_i,u_i\}_{i \in T}, \{d_i\}_{i \in \overline{T}}$.

\end{itemize}

\item {\bf Sender Message.} $\sendr$ does the following:
    \begin{itemize}
        \item {\bf Check Receiver Message.} $\sendr$ aborts if any of these checks fail:
        \begin{itemize}
            \item Check that $T = \fsro(c_1||\ldots || c_n)$.
            \item Check that $\Rec(\{c_i\}_{i \in T},T,\{(\widehat{r_i}^0, \widehat{r_i}^1, \widehat{\theta}_i),u_i\}_{i \in T}) \neq \bot$.
            \item For every $i \in T \cap U$ such that $\widehat{\theta}_i=b_i$, check that $\widehat{r_i}^0 = r_i^0$ and $\widehat{r_i}^1 = r_i^1$.
    \end{itemize}

    \item {\bf Message.} Sample $s \leftarrow \zo^{p(\lambda)}$ and let $R_0,R_1$ be the concatenation of the bits $\{r^{d_i}_i\}_{i \in \overline{T} \setminus U}, \{r^{d_i \oplus 1}_{i}\}_{i \in \overline{T} \setminus U}$ respectively and send to $\recv$ the following:
    \[
        s, U, ct_0 = m_0 \oplus h(s, R_0),ct_1 = m_1 \oplus h(s, R_1)
    \]
    \end{itemize}

\item {\bf Receiver Output.} Set $R$ as the concatenation of $\{\widehat{r_i}^{\widehat{\theta}_i}\}_{i \in \overline{T} \setminus U}$ and output
$
    m_b = ct_b \oplus h(s,R).
$
\end{enumerate}
     \end{framed}
    \caption{Three-round chosen-input OT without entanglement}
    \label{fig:qot-3r-bb84}
\end{figure}


\begin{theorem}
    Instantiate \proref{fig:qot-3r-bb84} with any non-interactive commitment scheme that is \emph{extractable} (\cref{def:fextractable}) and \emph{equivocal} (\cref{def:fequivocal}). Then there exist constants $A,B$ such that \proref{fig:qot-3r-bb84} securely realizes (\cref{def:secure-realization}) $\cF_{\OT[\secp]}$. 
    
    Furthermore, letting $\secp$ be the security parameter, $q$ be an upper bound on the total number of random oracle queries made by the adversary, and using the commitment scheme from \cref{sec:eecom-construction} with security parameter $\secp_\com = 4\secp$, for constants $A = 11\ 700, B=30\ 400$, \proref{fig:qot-3r-bb84} securely realizes $\cF_{\OT[\secp]}$ with $\mu_{\sR^*}$-security against a malicious receiver and $\mu_{\sS^*}$-security against a malicious sender, where 
    \[
        \mu_{\sR^*} = \frac{3\sqrt{10}q^{3/2}}{2^{\secp}}+\frac{1}{2^{5\secp}}+\frac{148(q + 126300\secp + 1)^3 + 1}{2^{4\secp}} + \frac{1010400q\secp}{2^{2\secp}}, \quad \mu_{\sS^*} = \left(\frac{712q\secp^{1/2}}{2^{2 \secp}}\right).
    \]
    This requires a total of $2(A+B)\secp = 84\ 200\secp$ BB84 states.
\end{theorem}
\else\fi

\ifsubmission \noindent We defer a full proof of this theorem to \cref{sec:bb84-3r-ot}. 
In \cref{sec:ot-eecom}, we analyze a variant of the CK88 protocol yielding chosen input string OT in $4$ and random input bit OT in $3$ rounds. We defer these to the appendix due to lack of space.
\else 
\fi

\ifsubmission\else\section{The fixed basis framework: OT without entanglement or setup}
\label{sec:bb84-3r-ot}

\ifsubmission \else
In \cref{fig:qot-3r-bb84}, we formalize our 3 round chosen-input OT protocol that does not rely on entanglement or setup. 

\begin{figure}[hbt!]
\begin{framed}
\begin{center}
\textbf{Protocol~\ref{fig:qot-3r-bb84}}
\end{center}
{\bf Ingredients / parameters / notation.}
\vspace{-2mm}
\begin{itemize}
    \itemsep 0.5mm
    \item Security parameter $\secp$ and constants $A,B$. Let $k = A\secp, n = (A+B)\secp$. 
    \item For classical bits $(x,\theta)$, let $\ket{x}_\theta$ denote $\ket{x}$ if $\theta = 0$, and $(\ket{0} + (-1)^x\ket{1})/\sqrt{2}$ if $\theta = 1$.
    \item A non-interactive extractable and equivocal commitment $(\Com,\allowbreak\Open,\allowbreak\Rec)$, where commitments to 3 bits have size $\ell \coloneqq \ell(\secp)$.
    \item A random oracle $\fsro: \zo^{n\ell} \to  \zo^{\lceil\log\binom{n}{k}\rceil}$, and
    a universal hash function family $h : \zo^{p(\secp)} \times \zo^{\leq B\secp} \rightarrow \zo^\secp$.
\end{itemize}
\vspace{-2mm}
{\bf Sender Input:} Messages $m_0,m_1 \in \zo^\secp$.     {\bf Receiver Input:} Choice bit b.

\begin{enumerate}
\item {\bf Sender Message.}
    $\sendr$ samples strings $r^0 \leftarrow \zo^n,r^1 \leftarrow \zo^n$, a random subset $U \subset [n]$ of size $k$, and for $i \in U$, it samples $b_i \leftarrow \zo$ uniformly at random.
    It computes state $\ket{\psi} = \ket{\psi}_1 \ldots \ket{\psi}_n$ as follows, and sends it to $\recv$: for $i \in U$, $\ket{\psi}_i = (\ket{r^0_i}_{b_i}, \ket{r^1_i}_{b_i})$ and for $i \in [n] \setminus U$, $\ket{\psi}_i = (\ket{r^0_i}_{0}, \ket{r^1_i}_{1})$.

\item {\bf Receiver Message.} $\recv$ does the following.
\vspace{-2mm}
\begin{itemize}
    \item Choose $\widehat{\theta} \leftarrow \zo^{n}$ and measure the $i^{th}$ pair of qubits in basis $\widehat{\theta}_i$ to obtain $\widehat{r_i}^0, \widehat{r_i}^1$.
    \item {\bf Measurement Check Message.} 
        \begin{itemize}
            \item Compute $\left(\state,\{c_i\}_{i \in [n]}\right) \gets \Com\left(\{(\widehat{r_i}^0, \widehat{r_i}^1, \widehat{\theta}_i)\}_{i \in [n]}\right)$.
            \item Compute $T = \fsro(c_1||\dots||c_n)$ and parse $T$ as a subset of $[n]$ of size $k$.
            \item Compute $\{(\widehat{r_i}^0, \widehat{r_i}^1, \widehat{\theta}_i),u_i\}_{i \in T} \gets \Open(\state,T)$.
    \end{itemize}

    \item {\bf Reorientation.} Let $\overline{T} \coloneqq [n] \setminus T$, and for all $i \in \overline{T}$, set $d_i = b \oplus \widehat{\theta}_i$.

    \item {\bf Message.} Send to $\sendr$ the values $\{c_i\}_{i \in [n]},T, \{\widehat{r_i}^0, \widehat{r_i}^1, \widehat{\theta}_i,u_i\}_{i \in T}, \{d_i\}_{i \in \overline{T}}$.

\end{itemize}

\item {\bf Sender Message.} $\sendr$ does the following.
\vspace{-2mm}
    \begin{itemize}
        \item {\bf Check Receiver Message.} $\sendr$ aborts if any of these checks fail:
        \begin{itemize}
            \item Check that $T = \fsro(c_1||\ldots || c_n)$.
            \item Check that $\Rec(\{c_i\}_{i \in T},T,\{(\widehat{r_i}^0, \widehat{r_i}^1, \widehat{\theta}_i),u_i\}_{i \in T}) \neq \bot$.
            \item For every $i \in T \cap U$ such that $\widehat{\theta}_i=b_i$, check that $\widehat{r_i}^0 = r_i^0$ and $\widehat{r_i}^1 = r_i^1$.
    \end{itemize}

    \item {\bf Message.} Sample $s \leftarrow \zo^{p(\lambda)}$, let $R_\beta$ denote the concatenation of 
    $\{r^{d_i \oplus \beta}_{i}\}_{i \in \overline{T} \setminus U}$ and send to $\recv$ the values
        $\left( s, U, ct_0 = m_0 \oplus h(s, R_0),ct_1 = m_1 \oplus h(s, R_1) \right)
    $.
    \end{itemize}

\item {\bf Receiver Output.} 
Output
$
    m_b = ct_b \oplus h(s,R)
$
where $R$ is the concatenation $\{\widehat{r_i}^{\widehat{\theta}_i}\}_{i \in \overline{T} \setminus U}$.
\end{enumerate}
     \end{framed}
    \caption{Three-round chosen-input OT without entanglement}
    \label{fig:qot-3r-bb84}
\end{figure}
\fi

%
%
%
%

\begin{theorem}
    Instantiate \proref{fig:qot-3r-bb84} with any non-interactive commitment scheme that is \emph{extractable} (\cref{def:fextractable}) and \emph{equivocal} (\cref{def:fequivocal}). Then there exist constants $A,B$ such that \proref{fig:qot-3r-bb84} securely realizes (\cref{def:secure-realization}) $\cF_{\OT[\secp]}$. 
    
    Furthermore, letting $\secp$ be the security parameter, $q$ be an upper bound on the total number of random oracle queries made by the adversary, and using the commitment scheme from \cref{sec:eecom-construction} with security parameter $\secp_\com = 4\secp$, for constants $A = 11\ 700, B=30\ 400$, \proref{fig:qot-3r-bb84} securely realizes $\cF_{\OT[\secp]}$ with $\mu_{\sR^*}$-security against a malicious receiver and $\mu_{\sS^*}$-security against a malicious sender, where 
    \[
        \mu_{\sR^*} = \frac{3\sqrt{10}q^{3/2}}{2^{\secp}}+\frac{1}{2^{5\secp}}+\frac{148(q + 126300\secp + 1)^3 + 1}{2^{4\secp}} + \frac{1010400q\secp}{2^{2\secp}}, \quad \mu_{\sS^*} = \left(\frac{712q\secp^{1/2}}{2^{2 \secp}}\right).
    \]
    This requires a total of $2(A+B)\secp = 84\ 200\secp$ BB84 states.
\end{theorem}
\begin{proof}
    We begin by proving security against malicious senders below.
    
    %
    %
    %
    %

    \noindent\underline{\textbf{Receiver security}}
    %
    %
    %
    %
    %
    
    We now describe a simulator $\simu$ that simulates the view of an arbitrary malicious sender $\sS^*$. $\simu$ will answer random oracle queries to $H$ using $\SimEqu.\RO$, the random oracle simulator for the commitment scheme $(\Com^\comro,\allowbreak\Open^\comro,\allowbreak\Rec^\comro)$. Additionally, the queries to $\fsro$ will be simulated using an efficient on-the-fly random oracle simulator $\simro.\RO$ as mentioned in Imported Theorem \ref{impthm:extROSim}.

    \paragraph{The Simulator.} $\simu[\sS^*]$ does the following.
    \begin{enumerate}
        \item Receive $\{ \ket{\psi} \}_{i \in [n]}$ from $\sS^*$.
        
        \item Perform the following steps.
        
        \begin{itemize}
        \item {\bf Measurement Check Message.} 
            \begin{itemize}
                \item Compute $\left(\{c_i\}_{i \in [n]}\right) \gets \SimEqu.\Com$ .
                
                \item Compute $T = \fsro(c_1||\dots||c_n)$ and parse $T$ as a subset of $[n]$ of size $k$.
                
                \item Perform (delayed) measurements on $\{ \ket{\psi} \}_{i \in [n]}$ as follows:
                
                \begin{itemize}
                    \item Sample $\widehat{\theta} \leftarrow \zo^{n}$.
                    
                    \item For all $i \in T$, measure the $i^{th}$ pair of qubits in basis $\widehat{\theta}_i$ to obtain $\widehat{r_i}^0, \widehat{r_i}^1$.
                    
                    \item For all $i \in \overline{T}$, measure the first qubit of $\ket{\psi}_i$ in the computational basis and the second qubit in the Hadamard basis to obtain $\widehat{r_i}^0, \widehat{r_i}^1$ respectively. 
                \end{itemize}

                \item Compute $\{u_i\}_{i \in [n]} \gets \SimEqu.\Open(\{(\widehat{r_i}^0, \widehat{r_i}^1, \widehat{\theta}_i)\}_{i \in [n]})$.
        \end{itemize}
    
        \item {\bf Reorientation.} Let $\overline{T} = [n] \setminus T$, and for all $i \in \overline{T}$, set $d_i = \widehat{\theta}_i$.
    
        \item {\bf Message.} Send to $\sendr$ \[\{c_i\}_{i \in [n]},T, \{\widehat{r_i}^0, \widehat{r_i}^1, \widehat{\theta}_i,u_i\}_{i \in T}, \{d_i\}_{i \in \overline{T}}.\]
        \end{itemize}
        
        \item Upon receiving $(s, U, ct_0, ct_1)$ from $\sS^*$,
        
        \begin{itemize}
                \item
                Set $R_0$ to be the concatenation $\{\widehat{r_i}^{\widehat{\theta}_i}\}_{i \in \overline{T} \setminus U}$ and $R_1$ to be the concatenation of $\{\widehat{r_i}^{\widehat{\theta}_i \oplus 1}\}_{i \in \overline{T} \setminus U}$.
                
                \item Compute $\widehat{m_0} := ct_0 \oplus h(s,R_0)$, $\widehat{m_1} := ct_1 \oplus h(s,R_1)$,
                and send $\widehat{m_0}, \widehat{m_1}$ to the ideal functionality.
        \end{itemize}

    \end{enumerate}
    
    \paragraph{Analysis.}
    Fix any adversary 
    $\{\sS^*_\secp,\sD_\secp,b_\secp\}_{\secp \in \bbN}$, where 
    $\sS^*_\secp$ is a QIOM that corrupts the sender, $\sD_\secp$ is a QOM, and $b_\secp$ is the input of the honest receiver. 
    For any receiver input $b_\secp \in \{0,1\}$ consider the random variables $\Pi[\sS^*_\secp,\sD_\secp,b_\secp]$ and 
    $\widetilde{\Pi}_{\cF_{\OT[\secp]}}[\Sim_\secp,\sD_\secp,b_\secp]$
    according to Definition \ref{functionalities-and-protocols-in-qrom} for the protocol in Figure \ref{fig:qot-3r-bb84}.
    Let $q(\cdot)$ denote an upper bound on the combined number of queries of $\sS^*_\secp$ and $\sD_\secp$.
    We will show that :
    \[\bigg|\Pr[\Pi[\sS^*_\secp,\sD_\secp,b_\secp] = 1] - \Pr[\widetilde{\Pi}_{\cF_{\OT[\secp]}}[\Sim_\secp,\sD_\secp,b_\secp] = 1]\bigg| = \mu(\secp,q(\secp)).\]
    
    This is done via a sequence of hybrids, as follows:

    \begin{itemize}
        \item $\Hyb_0$ : The output of this hybrid is the {\em real} distribution $\Pi[\sS^*_\secp,\sD_\secp,b_\secp]$.
        
        \item $\Hyb_1$: The output of this hybrid is the same as the previous hybrid except that the challenger uses switches $\fsro$ with an efficient on-the-fly random oracle simulator $\simro.\RO$ as mentioned in Imported Theorem \ref{impthm:extROSim}.
        
        \item $\Hyb_2$ : The output of this hybrid is the same as the previous hybrid except that instead of running $(\Com,\allowbreak\Open)$, the challenger uses $(\SimEqu.\RO,\allowbreak \SimEqu.\Com,\allowbreak\SimEqu.\Open)$ to prepare their commitments. It answers any random oracle queries to $\comro$ by calling $\Sim\Equ.\RO$ instead. 
            

        \item $\Hyb_3$ : The output of this hybrid is the same as the previous hybrid except that the measurement of $\{ \ket{\psi}_i\}_{i \in [n]}$ on behalf of $\sR$ is done after computing set $T$ and before invoking $\SimEqu.\Open$ on the measured values.
        
        \item $\Hyb_4$ : The output of this hybrid is the same as the previous hybrid except the following modification on behalf of $\sR$, for all $i \in \overline{T}$:
        
        \begin{itemize}
            \item Sample $\widehat{\theta_i} \gets \zo$
            
            \item Measure the first qubit of $\ket{\psi}_i$ in the computational basis and the second qubit in the Hadamard basis. Let the outcomes be $\widehat{r_i}^0, \widehat{r_i}^1$ respectively. 
        \end{itemize}

        \item $\Hyb_5$ :  The output of this hybrid is the same as the previous hybrid except the following modification.

        \begin{itemize}
            \item For $i \in \overline{T}$, set reorientation bit $d_i := \widehat{\theta_i}$.
            
            \item After receiving the last sender message.
            \begin{itemize}
            \item
                Set $R_0$ to be the concatenation $\{\widehat{r_i}^{\widehat{\theta}_i}\}_{i \in \overline{T} \setminus U}$ and $R_1$ to be the concatenation of $\{\widehat{r_i}^{\widehat{\theta}_i \oplus 1}\}_{i \in \overline{T} \setminus U}$.
                
                \item Compute $\widehat{m_0} := ct_0 \oplus h(s,R_0)$, $\widehat{m_1} := ct_1 \oplus h(s,R_1)$,
                and send $\widehat{m_0}, \widehat{m_1}$ to the ideal functionality.
            \end{itemize}
        \end{itemize}
        The output of this last hybrid is identical to the {\em ideal} distribution $\widetilde{\Pi}_{\cF_{\OT[\secp]}}[\Sim_\secp,\sD_\secp,b_\secp]$. 
    \end{itemize}

    We show that $|\Pr[\Hyb_5=1]-\Pr[\Hyb_0=1]| \leq \mu(\secp,q(\secp))$, where $(\Com^\comro,\allowbreak\Open^\comro,\allowbreak\Rec^\comro)$ is a $\mu(\secp,q(\secp))$-equivocal bit commitment scheme, where $\mu(\secp,q, n_\com) = \frac{2qn_{\com}^{1/2}}{2^{\secp_\com/2}}$ for the specific commitment scheme that we construct in Section \ref{sec:eecom-construction}, where $n_{\com}$ is the number of bit commitments and $\secp_\com$ is the security parameter for the commitment scheme. Later, we will set $n_\com = c_1 \secp$ and $\secp_\com = c_2 \secp$ for some fixed constants $c_1, c_2$. Thus $\mu$ will indeed be a function of $\secp$ and $q$. 
    We now procced with the proof by arguing indistinguishability of each pair of consecutive hybrids in the sequence above.
    
    \begin{claim}
    $\Pr[\Hyb_0=1] = \Pr[\Hyb_1=1]$. 
    \end{claim}
    
    \begin{proof}
       This follows from the indistinguishable simulation property of $\simro.\RO$ as mention in the Imported Theorem \ref{impthm:extROSim}. 
    \end{proof}

    \begin{claim}
    $|\Pr[\Hyb_1=1]-\Pr[\Hyb_2=1]|\leq \mu(\lambda, q(\lambda))$. 
    \end{claim}
    
    \begin{proof}
    Suppose there exists an adversary $\Adv_\secp$ corrupting $\sS$, a distinguisher $\sD_\secp$, 
    and a bit $b$ such that,
    
    \[
     \bigg|\Pr[\Hyb_1=1]-\Pr[\Hyb_2=1]\bigg| > \mu(\lambda, q(\lambda))
    \]
    
    We will build a reduction adversary $\{\Adv^*_{\secp} = (\Adv_{\RCommit,\secp},\allowbreak\Adv_{\ROpen,\secp},\allowbreak\sD^*_{\secp})\}_{\secp \in \bbN}$ that makes at most $q(\secp)$ queries to the random oracle, and contradicts the $\mu$-equivocality of the commitment $(\Com^\comro,\allowbreak\Open^\comro,\allowbreak\Rec^\comro)$ as defined in Definition \ref{def:fequivocal}. In the following reduction, all random oracle queries to $\comro$ will be answered by the equivocal commitment challenger whereas calls to $\fsro$ will be simulated by $\Adv^*_{\secp}$ by internally running $\simro.\RO$. \\
    
    $\Adv_{\RCommit,\secp}$:
    \begin{itemize}
        \item Initalize the OT protocol with between honest receiver $\sR$ and $\Adv$ corrupting $\sS$. 
        
        \item Output intermediate state $\rho^*_{\secp, 1}$ representing the joint state of $\sS$ and $\sR$ along with the measurement information $\{(\widehat{r_i}^0, \widehat{r_i}^1, \widehat{\theta}_i)\}_{i \in [n]}$ computed by $\sR$.
        
    \end{itemize}
    
    The measurement information $\{(\widehat{r_i}^0, \widehat{r_i}^1, \widehat{\theta}_i)\}_{i \in [n]}$  is sent as messages to the reduction challenger which then returns a set of commitments $\{ \com_i \}_{i \in [n]}$.\\
    
    $\Adv_{\ROpen,\secp}(\rho^*_{\secp, 1}, \{ \com_i \}_{i \in [n]})$: Use $\rho^*_{\secp, 1}$ to initialize the joint state of $\sS$ and $\sR$, and output the new joint state $\rho^*_{\secp, 2}$ after $\sR$ has computed $T$.\\
    
    The challenger returns $\{u_i\}_{i \in [n]}$ which is then fed to the following distinguisher (along with the information $\{ \com_i, (\widehat{r_i}^0, \widehat{r_i}^1, \widehat{\theta}_i) \}_{i \in [n]}$ from the aforementioned execution).\\
    
    $\sD^*_\secp(\rho^*_{\secp, 2}, \{ \com_i, (\widehat{r_i}^0, \widehat{r_i}^1, \widehat{\theta}_i), u_i \}_{i \in [n]}):$
    
    \begin{itemize}
        \item Use $\rho^*_{\secp, 2}$ to initialize the joint state of $\sS$ and $\sR$. Run it until completion using $\{ (\widehat{r_i}^0, \widehat{r_i}^1, \widehat{\theta}_i), u_i \}_{i \in T}$ as openings of $\sR$ in the measurement check proof
        
        \item Let $\tau_\secp^*$ be the final state of $\Adv$ and $y^*$ be the output of $\sR$. Run $\sD_\secp(\tau^*_\secp, y^*)$ and output the bit $b$ returned by the distinguisher.
    \end{itemize}
    
    By construction, when the challenger executes $(\Com^\comro,\allowbreak\Open^\comro)$, the reduction will generate a distribution identical to $\Hyb_1$. Similarly, when the challenger executes $(\SimEqu.\Com, \SimEqu.\Open)$, the reduction will generate a distribution identical to $\Hyb_2$. Therefore, the reduction $\{\Adv^*_{\secp} = (\Adv_{\RCommit,\secp},\allowbreak\Adv_{\ROpen,\secp},\allowbreak\sD^*_{\secp})\}_{\secp \in \bbN}$ directly contradicts the $\mu$-equivocality of the underlying commitment scheme $(\Com^\comro,\allowbreak\Open^\comro,\allowbreak\Rec^\comro)$ as in Definition \ref{def:fequivocal}.
    \end{proof}

    \begin{claim}
    $\Pr[\Hyb_2=1] = \Pr[\Hyb_3=1]$ 
    \end{claim}
    
    \begin{proof}
    The only difference in $\Hyb_3$ from $\Hyb_2$ is that we commute the measurement of $\{ \ket{\psi} \}_{i \in [n]}$ past the invocation of $\SimEqu.\Com$ and the computation of $T$. Since these two operators are applied to disjoint subsystems, this can be done without affecting the hybrid distribution.
    \end{proof}
    
    \begin{claim}
    $\Pr[\Hyb_3=1] = \Pr[\Hyb_4=1]$ 
    \end{claim}
    
    \begin{proof}
    The only difference in $\Hyb_4$ from $\Hyb_3$ is the following. For all $i \in \overline{T}$: If $\widehat{\theta_i} = 0$, we measure the second qubit of $\ket{\psi_i}$ in the Hadamard basis (instead of the  computational basis as defined in the previous hybrid). If $\widehat{\theta_i} = 1$, we measure the first qubit of $\ket{\psi_i}$ in computational basis (instead of the Hadamard basis as defined in the previous hybrid). But this doesn't affect the hybrid distribution because the values on these registers are not used anywhere in the hybrid and are eventually traced out.
    \end{proof}

    \begin{claim}
    $\Pr[\Hyb_4=1] = \Pr[\Hyb_5=1]$ 
    \end{claim}
    
    \begin{proof}
       The only difference between these experiments is the way in which we define the output of honest receiver. Assuming the correctness of $\cF_{\OT[\secp]}$, the two hybrids are identical. In $\Hyb_4$, the receiver's output is computed by the challenger as $\widehat{m_b} = ct_b \oplus h(s,||_{i \in \overline{T} \setminus U} \widehat{r_i}^{\widehat{\theta}_i})$ (where $||_{i \in G}x_i$ denotes the concatenation of $x_i$ for $i \in G$, in increasing order of $i$). In $\Hyb_5$, the receiver's output is derived via the OT ideal functionality which receives sender's input strings $\widehat{m_0} := ct_0 \oplus h(s,||_{i \in \overline{T} \setminus U} \widehat{r_i}^{\widehat{\theta}_i})$ and $\widehat{m_1} := ct_1 \oplus h(s,||_{i \in \overline{T} \setminus U} \widehat{r_i}^{\widehat{\theta}_i \oplus 1})$ from the challenger and receiver choice bit $b$. The OT ideal functionality sends $\widehat{m_b} = ct_b \oplus h(s,||_{i \in \overline{T} \setminus U} \widehat{r_i}^{\widehat{\theta}_i \oplus b})$ to the ideal receiver which it then outputs. Therefore for any fixing of the adversary's state and receiver choice bit, the two hybrids result in identical $\widehat{m_b}$. 
    \end{proof}
    
    Combining all the claims, we get that $|\Pr[\Hyb_0 = 1] - \Pr[\Hyb_5 = 1]| \le  \mu(\lambda, q(\lambda))$. Using Theorem \ref{thm:equcom} where we derived  $\mu(\secp,q, n_\com) = \frac{2qn_{\com}^{1/2}}{2^{\secp_\com/2}}$ and plugging $\lambda_\com = 4 \secp$, $n_\com = 3n$ (as we are committing to 3 bits at a time) where $n = 42\ 100\secp$ (this setting of $n$ is the same as that needed in the sender security part of the proof), we get $\frac{712q\sqrt{\secp}}{2^{2 \secp}}$ security against a malicious sender.
\end{proof}

    %
    %
    %
    %
    \noindent\underline{\textbf{Sender security}}
    
    Let $\SimExt = (\SimExt.\RO,\SimExt.\Ext)$ be the simulator for the extractable commitment scheme from \cref{sec:uccomm}. Let $+$ refer to the computational basis and $\times$ to the hadamard basis. Below we describe the simulator $\Sim[\sR^*]$ against a malicious receiver $\sR^*$ for \proref{fig:qot-3r-bb84}.
    
    \underline{$\Sim[\sR^*]$}:
    \begin{itemize}
        \item Initialize the on-the-fly random oracle simulator $\simro$ from \cref{impthm:extROSim}. Run $\sR^*$ answering its oracle queries to $\fsro$ using $\simro$ and queries to $\comro$ using $\SimExt.\RO$.
        \item Sample $2n$ EPR pairs on registers $\{(\cS_{i,b},\cR_{i,b})\}_{i \in [n], b \in \zo}$ (where each $\cS_{i,b},\cR_{i,b}$ is a 2-dimensional register). Send registers $\{\cR_{i,b}\}_{i \in [n],b \in \zo}$ to $\sR^*$.
        \item When $\sR^*$ outputs $\{c_i\}_{i \in [n]}, T, \{(\widehat{r_i}^0, \widehat{r_i}^1, \widehat{\theta}_i),u_i\}_{i \in T}, \{d_i\}_{i \in \overline{T}}$, run $\{(\widetilde{r_{i}}^0,\widetilde{r_{i}}^{1},\widetilde{\theta}_i)\}_{i \in [n]} \leftarrow \SimExt.\Ext(\{c_i\}_{i \in [n]})$. 
        \item Run the ``check receiver message" part of the honest sender strategy, except do the following in place of the third check. If any of the checks fail, send $\mathsf{abort}$ to the ideal functionality, output $\rcv^*$'s state and continue answering distinguisher's queries.
            \begin{itemize}
                \item Sample subset $U \subset [n]$ of size $k$ and for each $i \in [n]$, sample bit $b_i \in \zo$. 
                \item For each $i \in [n]$, do the following: 
                    \begin{itemize}
                        \item If $i \in U$, measure both registers $\cS_{i,0}$ and $\cS_{i,1}$ in basis $+$ when $b_i = 0$, and both in basis $\times$ when $b_i = 1$. Denote measurement outcomes from $\cS_{i,0}$ and $\cS_{i,1}$ by $r_i^0$ and $r_i^1$ respectively. 
                        \item If $i \notin U$, then measure $\cS_{i,0}$ in basis $+$ and $\cS_{i,1}$ in basis $\times$ and denote outcomes by outcomes $r_i^0,r_i^1$ respectively. 
                    \end{itemize}
                \item For each $i \in T\cap U$ such that $\widetilde{\theta}_i = b_i$, check that $\widetilde{r_i}^0 = r_i^0$ and $\widetilde{r_i}^1 = r_i^1$.
            \end{itemize}
        \item Set $b \coloneqq \maj\{\widetilde{\theta}_i \oplus d_i\}_{ i\in \overline{T}\setminus U}$ and send $b$ to $\cF_{\OT[\secp]}$ to obtain $m_b$.
        \item Compute the last message using the honest sender strategy except for using $m_{1-b} \coloneqq 0^\secp$. 
        \item Send $\sR^*$ this last message, output the final state of $\sR^*$ and terminate.
        \item Answer any queries of distinguisher to $\fsro$ and $\comro$ using $\simro$ and $\SimExt.\RO$ respectively.
    \end{itemize}
    Fix any distinguisher $\sD$ and let $q$ denote the total queries that $\sR^*,\sD$ make to $\fsro$ and $\comro$. Consider the following sequence of hybrids:
    \begin{itemize}
        \item $\Hyb_0$: This is the real world interaction between $\sR^*$ and $\sS$. Using the notation of \cref{def:secure-realization}, this is a distribution over $\zo$ denoted by $\Pi[\sR^*,\sD,(m_0,m_1)]$.
        \item $\Hyb_1$: This is the same as the previous hybrid, except the following are run instead to generate the first sender message: (1) Sample $2n$ EPR pairs on registers $\{(\cS_{i,b},\cR_{i,b})\}_{i \in [n], b \in \zo}$. (2) Run the following algorithm: 
        
            \underline{\textbf{Algorithm Measure-EPR}}:
            \begin{itemize}
                \item Sample subset $U \subset [n]$ of size $k$ and for each $i \in [n]$, sample bit $b_i \in \zo$. 
                \item For each $i \in [n]$, do the following: 
                    \begin{itemize}
                        \item If $i \in U$, measure registers $\cS_{i,0}$ and $\cS_{i,1}$ in basis $+$ when $b_i = 0$, and in basis $\times$ when $b_i = 1$, to get outcomes $r_i^0$ and $r_i^1$ respectively. 
                        \item If $i \notin U$, then measure $\cS_{i,0}$ in basis $+$ and $\cS_{i,1}$ in basis $\times$ to get outcomes $r_i^0,r_i^1$ respectively. 
                    \end{itemize}
            \end{itemize}
        Thereafter, registers $\{\cR_{i,b}\}_{i \in [n],b \in \zo}$ are sent over to $\sR^*$ and the rest of the experiment works as the previous hybrid.
        \item $\Hyb_2$: This is the same previous hybrid, except that the sender does not perform any measurements before sending registers $\{\cR_{i,b}\}_{i \in [n],b \in \zo}$ to $\sR^*$, and
        delays running the algorithm Measure-EPR to just before executing the third check in ``check receiver message" part of the honest sender strategy.
        \item $\Hyb_3$: This is the same as the previous hybrid, except for the following changes: queries of $\sR^*$ to $\comro$ are now answered using $\SimExt.\RO$. Once $\sR^*$ outputs its second message, run $\{(\widetilde{r_{i}}^0,\widetilde{r_{i}}^1,\widetilde{\theta}_i)\}_{i \in [n]} \leftarrow \SimExt.\Ext(\{c_i\}_{i \in [n]})$. Thereafter, $\{(\widetilde{r_i}^0,\widetilde{r_i}^1,\widetilde{\theta}_i)\}_{i \in T}$ are used for the third check in the ``check receiver part" of the honest sender strategy (instead of using $\{(\widehat{r_i}^0,\widehat{r_i}^1,\widehat{\theta}_i)\}_{i \in T}$).
        \item $\Hyb_4$: This is the result of the interaction between $\Sim[\sR^*]$, $\cF_{\OT[\secp]}$ and honest sender $\sS$. Using the notation of \cref{def:secure-realization}, this is denoted by $\widetilde{\Pi}_{\cF_{\OT[\secp]}}[\Sim[\sR^*],\sD,(m_0,m_1)]$.
    \end{itemize}
    We prove the indistinguishability between the hybrids using the following claims:
    \begin{claim} $\Pr[\Hyb_0=1] = \Pr[\Hyb_1=1]$
    \end{claim}
    \begin{proof}
        The only difference between the two hybrids is in how $\sS$ samples the state on the registers that it sends to $\rcv^*$. Denote the registers that $\sS$ sends to $\rcv^*$ in either hybrid by $\{\cR_{i,b}\}_{i \in [n], b \in \zo}$. 
        In $\Hyb_0$, each pair $(\cR_{i,0}, \cR_{i,1})$ contains state $(\ket{r_i^0}_{b_i}, \ket{r_i^1}_{b_i})$, for $i \in U$ and $b_i$ chosen uniformly from $\{+,\times\}$, and $(\ket{r_i^0}_{0}, \ket{r_i^1}_{1})$ for $i \notin U$, and for independently uniformly sampled bits $r_i^0, r_i^1$. 
        In $\Hyb_1$, the challenger prepares $2n$ EPR pairs on registers $\{(\cS_{i,b},\cR_{i,b})\}_{i \in [n],b\in \zo}$, then for every $i \in U$ measures the pair $\cS_{i,0}, \cS_{i,1}$ in basis $b_i$ that is uniformly sampled from $\{+,\times\}$, and for $i \notin U$ measures $\cS_{i,0}, \cS_{i,1}$ in basis $0,1$ respectively. 
        By elementary properties of EPR pairs, each register $\cR_{i,b}$ is in a state $\ket{r}$ for a uniformly independently sampled bit $r$ and in a basis that is chosen from the same distribution in both experiments.
    \end{proof}
    \begin{claim} $\Pr[\Hyb_1=1] = \Pr[\Hyb_2=1]$
    \end{claim}
    \begin{proof}
        In $\Hyb_1$, $\sS$ measures the registers $\{\cS_{i,b}\}_{i \in [n],b\in \zo}$ first, after which $\rcv^*$ operates on registers $\{\cR_{i,b}\}_{i \in [n],b\in \zo}$. In $\Hyb_2$, $\sS$ performs the same measurements, but after receiving the second round message from $\rcv^*$. Indistinguishability follows because measurements on disjoint sub-systems commute. 
    \end{proof}
    \begin{claim} 
    \[
        |\Pr[\Hyb_2=1] - \Pr[\Hyb_3=1]| \leq \frac{148(q + 3n + 1)^3 + 1}{2^{4\secp}} + \frac{24qn}{2^{2\secp}}.
    \]
    \end{claim}
    \begin{proof}
        This follows by a direct reduction to extractability of the commitment scheme (\cref{def:fextractable}). Indeed, let $\Adv_\mathsf{Commit}$ be the machine that runs $\Hyb_0$ until $\rcv^*$ outputs its message, which includes $\{c_i\}_{i \in [n]}$. Let $\Adv_\mathsf{Open}$ be the machine that takes as input the rest of the state of $\Hyb_0$, which includes $T$ and the openings $\{(\widehat{r_i}^0, \widehat{r_i}^1, \widehat{\theta}_i),u_i\}_{i \in T}$, and outputs $T$ and these openings. Let $\sD$ be the machine that runs the rest of $\Hyb_0$ and outputs a bit.

        Then, plugging in $\secp_\com = 4\secp$, \cref{def:fextractable} when applied to $(\Adv_{\mathsf{Commit}},\Adv_{\mathsf{Open}},\sD)$ implies that the hybrids cannot be distinguished except with probability \[\frac{148(q + 3n + 1)^3 + 1}{2^{4\secp}} + \frac{24qn}{2^{2\secp}},\] since we are committing to a total of $3n$ bits.
    \end{proof}
    \begin{claim}
    \label{claim:qot-3r-bb84-hyb34}
    For $A = 11\ 700, B=30\ 400$, and $q \geq 5$,
    \begin{align*}
        |\Pr[\Hyb_3=1] - \Pr[\Hyb_4=1]| \leq \frac{3\sqrt{10}q^{3/2}}{2^{\secp}}+\frac{1}{2^{5\secp}}
    \end{align*}
    \end{claim}
    \begin{proof}
    The only difference between $\Hyb_3$ and $\Hyb_4$ is that $m_{1-b} = 0^\secp$, where $b = \maj\{\widetilde{\theta}_i \oplus d_i\}_{i \in \overline{T}}$. In what follows, we show that $m_{1-b}$ is masked with a string that is (statistically close to) uniformly random from even given the view of $\sR^*$ in either hybrid, which implies the given claim. 
    \paragraph{Notation:} We setup some notation before proceeding.
    \begin{itemize}
    \item Let $\bc \coloneqq (c_1,\dots,c_n)$ be the classical commitments and $\bb \coloneqq (b_1,\dots b_n)$ be the bits sampled by the sender while executing its checks.
    \item Write the classical extracted values $\{(\widetilde{r_i}^0,\widetilde{r_i}^1,\widetilde{\theta}_i)\}_{i \in [n]}$ as matrices
    \[
        \widetilde{\bR} \coloneqq \begin{bmatrix}\widetilde{r_1}^0 \ \dots \ \widetilde{r_n}^0 \\ \widetilde{r_1}^1 \ \dots \ \widetilde{r_n}^1\end{bmatrix}, 
        \widetilde{\btheta} \coloneqq \begin{bmatrix}\widetilde{\theta}_{1} \ \dots \ \widetilde{\theta}_{n} 
        \end{bmatrix}.
    \]
    \item Given any $\bR \in \{0,1\}^{2 \times n}$, $\btheta \in \zo^n$, define $\ket{\bR_{\btheta}}$ as a state on $n$ $4$-dimensional registers, where register $i$ contains the state $\ket{\bR_{i,0},\bR_{i,1}}$ prepared in the $(\btheta_{i},\btheta_{i})$-basis. 
    \item Given $\bR,\widetilde{\bR} \in \{0,1\}^{2 \times n}$ and a subset $T \subseteq [n]$, define $\bR_T$ be the columns of $\bR$ indexed by $T$, and define $\Delta\left(\bR_T,\widetilde{\bR}_{T}\right)$ as the fraction of columns $i \in T$ such that $(\bR_{i,0},\bR_{i,1}) \neq (\widetilde{\bR}_{i,0},\widetilde{\bR}_{i,1})$.
    \item For $T \subset [n]$, let $\overline{T} \coloneqq [n] \setminus T$.
    \item Given $\widetilde{\bR} \in \{0,1\}^{2 \times n}, \widetilde{\btheta} \in \zo^n$, $T \subseteq [n]$, $U \subseteq [n]$, $\bb \in \zo^n$, and $\delta \in (0,1)$, define
    \[
        \Pi^{\widetilde{\bR},\widetilde{\btheta},T,U,\bb,\delta} \coloneqq \sum_{\substack{\bR \, : \, \bR_{S'} = \widetilde{\bR}_{S'}, \Delta\left(\bR_{\overline{T}\setminus U}, \widetilde{\bR}_{\overline{T}\setminus U}\right) \geq \delta\\\text{where }S' = \{j \, | \, j \in T \cap U \, \wedge \, \bb_j = \widetilde{\btheta}_{j}\}}}\ket{\bR_{\widetilde{\btheta}}}\bra{\bR_{\widetilde{\btheta}}}.
    \]
\end{itemize}

\noindent Now, consider the following projection, which has hard-coded the description of $\fsro$:

\[
    \Pi_{\mathsf{bad}}^\delta \coloneqq \sum_{\substack{\bc,\widetilde{\bR},\widetilde{\btheta},\bb,\\U \subseteq [n], |U|=k}}\ket{\bc}\bra{\bc}_{\cC} \otimes \ket{\widetilde{\bR},\widetilde{\btheta}}\bra{\widetilde{\bR},\widetilde{\btheta}}_{\cZ_1} \otimes \ket{U,\bb}\bra{U,\bb}_{\cZ_2} \otimes \Pi^{\widetilde{\bR},\widetilde{\btheta},\fsro(\bc),U,\bb,\delta}_{\cS},
\]
where $\cC$ is the register holding the classical commitments, $\cZ_1$ is the register holding the output of $\SimExt.\Ext$, $\cZ_2$ is the register holding the subset $U$ and bits $\bb$ sampled by sender, and $\cS$ denotes all the registers holding the sender's halves of EPR pairs.

\begin{subclaim}\label{subclaim:tau3r}
Let 
\[
    \tau \coloneqq \sum_{\bc,\widetilde{\bR},\widetilde{\btheta},U,\bb}p^{(\bc,\widetilde{\bR},\widetilde{\btheta},U,\bb)} \tau^{(\bc,\widetilde{\bR},\widetilde{\btheta},U,\bb)},
\] 
where 
\[
    \tau^{(\bc,\widetilde{\bR},\widetilde{\btheta},U,\bb)} = \ket{\bc}\bra{\bc}_{\cC} \otimes \ket{\widetilde{\bR},\widetilde{\btheta}}\bra{\widetilde{\bR},\widetilde{\btheta}}_{\cZ_1} \otimes \ket{U,\bb}\bra{U,\bb}_{\cZ_2} \otimes \rho^{(\bc,\widetilde{\bR},\widetilde{\btheta},U,\bb)}_{\cS,\cX}
\] 
is the entire state of $\Hyb_3$ (equivalently also $\Hyb_4$) immediately after $\rcv^*$ outputs its message (which includes $\bc$), $\SimExt.\Ext$ is run to get $\widetilde{\bR},\widetilde{\btheta}$, and sender samples the set $U \subseteq [n]$ of size $d$ and bits $\bb \in \zo^n$. Here, each $p^{(\bc,\widetilde{\bR},\widetilde{\btheta},U,\bb)}$ is the probability that the string $\bc,\widetilde{\bR},\widetilde{\btheta},U,\bb$ is contained in the registers $\cC,\cZ_1,\cZ_2$. Also, $\cS$ is the register holding the sender's halves of EPR pairs and $\cX$ is a register holding remaining state of the system, which includes the rest of the receiver's classical message and its private state. Then, for $A = 11\ 700, B=30\ 400$ and for $q\geq 5$,
\[
    \Tr(\Pi_{\mathsf{bad}}^{11/200} \tau) \leq \frac{45q^3}{2^{2\secp}}
\]
\end{subclaim}
\begin{proof}
Define $\Adv_{\rcv^*}^{\fsro}$ to be the oracle machine that runs $\Hyb_3$ until $\rcv^*$ outputs $\bc$ (and the rest of its message), then runs $\SimExt.\Ext$ to obtain $\ket{\widetilde{\bR},\widetilde{\btheta}}\bra{\widetilde{\bR},\widetilde{\btheta}}_{\cZ_1}$, followed by sampling the set $U \subseteq [n]$ of size $d$, and bits $\bb \in \zo^n$ in the register $\cZ_2$, and finally outputting the remaining state $\rho_{\cS,\cX}$. Consider running the measure-and-reprogram simulator $\Sim[\Adv_{\rcv^*}]$ from \cref{thm:measure-and-reprogram}, which simulates $\fsro$ queries, measures and outputs $\bc$, then receives a uniformly random subset $T \subset [n]$ of size $k$, and then continues to run $\Adv_{\rcv^*}$ until it outputs $\ket{\widetilde{\bR},\widetilde{\btheta}}\bra{\widetilde{\bR},\widetilde{\btheta}}_{\cZ_1} \otimes \ket{U,\bb}\bra{U,\bb}_{\cZ_2} \otimes \rho_{\cS,\cX}$. Letting 
\[
    \Pi_{\mathsf{bad}}^{\delta}[T] \coloneqq \sum_{\substack{\bc,\widetilde{\bR},\widetilde{\btheta},\bb \\ U \subseteq [n], |U|=d}}\ket{\bc}\bra{\bc}_{\cC} \otimes \ket{\widetilde{\bR},\widetilde{\btheta}}\bra{\widetilde{\bR},\widetilde{\btheta}}_{\cZ_1} \otimes \ket{U,\bb}\bra{U,\bb}_{\cZ_2} \otimes \Pi^{\widetilde{\bR},\widetilde{\btheta},T,U,\bb,\delta}_{\cS},
\] 
for $T \subset [n]$, \cref{thm:measure-and-reprogram} implies that $\Tr\left(\Pi_\mathsf{bad}^{\delta}\tau\right) \leq (2q+1)^2 \gamma$, where
\begin{align*}
    \gamma &= \E\left[\Tr\left(\Pi_{\mathsf{bad}}^{\delta}[T]\left(\ket{\bc}\bra{\bc}_\cC \otimes \ket{\widetilde{\bR},\widetilde{\btheta}}\bra{\widetilde{\bR},\widetilde{\btheta}}_{\cZ_1}\otimes \ket{U,\bb}\bra{U,\bb}_{\cZ_2} \otimes \rho_{\cS,\cX}\right)\right)\right]
\end{align*}
with expectation defined over the following experiment:
\begin{itemize}
    \item $(\bc,\state) \gets \Sim[\Adv_{\rcv^*}]$,
    \item $T \gets S_{n,k}$, the set of all subsets of $[n]$ of size $k$,
    \item $\left(\widetilde{\bR},\widetilde{\btheta},U,\\\bb,\rho_{\cS,\cX}\right) \gets \Sim[\Adv_{\rcv^*}](T,\state)$.
\end{itemize}
Now, recall that one of the last things that $\Adv_{\rcv^*}$ does in $\Hyb_3$ is run $\SimExt.\Ext$ on $\bc$ to obtain $(\widetilde{\bR},\widetilde{\btheta})$. Consider instead running $\SimExt.\Ext$ on $\bc$ immediately after $\Sim[\Adv_{\rcv^*}]$ outputs $\bc$. Note that $\SimExt.\Ext$ only operates on the register holding $\bc$ and its own private state used for simulating $\comro$, so since $\Com^{\comro}$ has a $\frac{8}{2^{{\secp_\com}/2}}$-commuting simulator (\cref{def:fextractable}), we have that, 

\begin{align}
\Tr\left(\Pi_\mathsf{bad}^{\delta}\tau\right) &\leq (2q+1)^2\left(\epsilon + \frac{8q}{2^{{\secp_\com}/2}}\right) \label{eq:taueq1}
\end{align}
where 
\begin{align}
    \label{eq:samplingrealexp}
    \epsilon \coloneqq \E\left[
    \Tr\left(\Pi_{\mathsf{bad}}^{\delta}[T]\left(\ket{\bc}\bra{\bc}_\cC  \otimes \ket{\widetilde{\bR},\widetilde{\btheta}}\bra{\widetilde{\bR},\widetilde{\btheta}}_{\cZ_1}\otimes \ket{U,\bb}\bra{U,\bb}_{\cZ_2} \otimes \rho_{\cS,\cX}\right)\right)
    \right]
\end{align}
over the randomness of the following experiment:
\begin{itemize}
    \item $(\bc,\state) \gets \Sim[\Adv_{\rcv^*}]$,
    \item $(\widetilde{\bR},\widetilde{\btheta}) \gets \SimExt.\Ext(\bc)$,
    \item $T \gets S_{n,k}$,
    \item $\left(\widetilde{\bR},\widetilde{\btheta},U,\bb,\rho_{\cS,\cX}\right) \gets \Sim[\Adv_{\rcv^*}](T,\state)$.
\end{itemize}
The sampling of $T,U \subseteq [n]$ each of size $k$ uniformly and independently at random in the experiment above is equivalent to the following sampling strategy.
First, sample the size of their intersection, i.e. sample and fix $s = |T \cap U|$, this fixes the size of $T \cup U$ to be $2k-s$, since $|T| = |U| = s$.
Next sample and fix a set $T'$ of size $2k-s$ (that will eventually represent the union $T \cup U$). Finally, sample a subset $S  \subset T'$ of size $s$ (which will eventually represent the intersection $T \cap U$), and then obtain $T$ and $U$ by paritioning $T' \setminus S$ into two random subsets each of size $k-s$, and then computing the union of each set with $S$.
This is described formally below.
\begin{itemize}
    \item Fix a state on register $\cS$ (and potentially other registers of arbitrary size), where $\cS$ is split into $n$ registers $\cS_1,\dots,\cS_n$ of dimension 4, and fix $\widetilde{\bR}\in \{0,1\}^{2 \times n},\widetilde{\btheta} \in \zo^n$.
    \item 
    Sample two independent and uniform subsets of $[n]$ each of size $k$. 
    Let $s$ denote the size of their intersection.
    Fix $s$, and discard the subsets themselves.
    \item Sample a random subset $T'$ of $[n]$, of size $2k-s$. 
    \item Sample subsets $T,U,S' \subseteq T'$ as follows:
    \begin{itemize}
        \item Sample and fix a random subset set of size $s$ of $T'$, call this subset $S$.
        \item Partition $T'\setminus S$ (note: this has size $2k-2s$) into two equal sets $W_1$ and $W_2$ of size $k-s$. 
        
        This can be done by first sampling a set $W_1$ of size $k-s$ uniformly at random from $T'\setminus S$ and setting $W_2 = (T'\setminus S) \setminus W_1$. 
        \item Let $T = W_1 \cup S$ and $U = W_2 \cup S$.
        \item Sample bits $\bb \in \zo^n$ and set $S' = \{j \, | \, j \in S \, \wedge \, \bb_j = \widetilde{\btheta}_j \}$.
    \end{itemize}
    \item For each $i \in S'$, measure the register $\cS_i$ in basis  $\widetilde{\btheta}_{i}$ to get $\bR_{S'} \in \zo^{2 \times |S'|}$. Output $\Delta(\bR_{S'},\widetilde{\bR}_{S'})$. 

\end{itemize}

The quantum error probability $\epsilon^{\delta}_{\mathsf{quantum}}$ (\cref{def:quantum-error-probability}) of the above game corresponds to the trace distance between the initial state on register $\cS$ and an ``ideal'' state (as defined in \cref{def:quantum-error-probability}). This ideal state is supported on vectors $\ket{\bR_{\widetilde{\btheta}}}$ such that $|\Delta(\bR_{\overline{T'}},\widetilde{\bR}_{\overline{T'}}) - \Delta(\bR_{S'},\widetilde{\bR}_{S'})| < \delta$. In particular, for any $\ket{\bR_{\widetilde{\btheta}}}$ with $\Delta(\bR_{S'},\widetilde{\bR}_{S'}) = 0$ in the support of the ideal state, it holds that $\Delta(\bR_{\overline{T'}},\widetilde{\bR}_{\overline{T'}}) < \delta$, or $\Delta(\bR_{\overline{T}\setminus U},\widetilde{\bR}_{\overline{T}\setminus U})<\delta$ (since $\overline{T'} = \overline{T}\setminus U$ in the sampling game above). Thus, this ideal state is orthogonal to the subspace $\Pi^{\widetilde{\bR},\widetilde{\btheta},T,U,\bb,\delta}_\cS$, and so it follows that $\epsilon$ is bounded by $\epsilon^{\delta}_{\mathsf{quantum}}$.

Thus, by \cref{impthm:error-probability}, $\epsilon$ is then bounded by $\sqrt{\epsilon^\delta_{\mathsf{classical}}}$, where $\epsilon^\delta_{\mathsf{classical}}$ is the \emph{classical} error probability (\cref{def:classical-error-probability}) in the corresponding classical sampling game, defined as follows:
\begin{itemize}
    \item Let $\bR, \widetilde{\bR} \in \{0,1\}^{2 \times n}$ s.t.\ $\bR$ is the matrix on which we are running the sampling and $\widetilde{\bR}$ is an arbitrary matrix.
    \item Sample two independent and uniform subsets of $[n]$ each of size $k$. 
    Let $s$ denote the size of their intersection.
    Fix $s$, and discard the subsets themselves.
    Sample a random subset $T'$ of $[n]$, of size $2k-s$.
    \item Sample subset $S' \subseteq T'$ as follows:
    \begin{itemize}
        \item Sample $S$ as a random subset of $T'$ of size $s$.
        \item Sample bits $\bb \in \zo^n$ and set $S' = \{j \, | \, j \in S \, \wedge \, \bb_j = \widetilde{\btheta}_j \}$.
    \end{itemize}
    \item Output $\Delta(\bR_{S'},\widetilde{\bR}_{S'})$.
\end{itemize}

We provide an analysis of this classical sampling game in \cref{appsubsec:3roundsamplingstrategy}. 
Using \cref{lemma:intersectionsampling} from the same appendix, we get that for $0 < \epsilon,\beta,\delta< 1$ and $0 < \gamma < \delta$,
\begin{align*}
    \epsilon_{\mathsf{classical}}^\delta &\leq 2\exp\left(-2\left(\frac{(n-k)^2 - 3\epsilon k^2}{(n-k)^2+(1-2\epsilon)k^2}\right)^2\gamma^2 (1-\epsilon)\frac{k^2}{n}\right)\\
    & + 2\exp\left(-(\delta-\gamma)^2(1-\beta)(1-\epsilon)\frac{k^2}{n}\right)\\
    & + \exp\left(-\frac{\beta^2(1-\epsilon)k^2}{2n}\right) + 2\exp\left(-\frac{2\epsilon^2k^3}{n^2}\right)
\end{align*}

Setting $\delta = 11/200, \epsilon = 0.03917, \beta = 0.04213, \gamma = 0.02456, k = A\secp, n = (A+B)\secp, A = 11\ 700, B=30\ 400$, we get each of the exp terms above is $\leq \frac{1}{2^{4\secp}}$. Thus, $\epsilon^\delta_{\mathsf{classical}} \leq \frac{7}{2^{4\secp}}$, giving us, $\epsilon \leq \epsilon^\delta_{\mathsf{quantum}} \leq \sqrt{\epsilon^\delta_{\mathsf{classical}}} \leq \frac{\sqrt{7}}{2^{2\secp}}$.

This gives using \cref{eq:taueq1} that:
\[
    \Tr(\Pi_{\mathsf{bad}}^{11/200} \tau) \leq (2q+1)^2 \left[\frac{\sqrt{7}}{2^{2\secp}} + \frac{8q}{2^{\secp_\com/2}} \right]
\]
For $\secp_\com = 4\secp$ and $q\geq 5$,
\[
    \Tr(\Pi_{\mathsf{bad}}^{11/200} \tau) \leq (2q+1)^2\left[\frac{\sqrt{7}}{2^{2\secp}} + \frac{8q}{2^{2\secp}} \right] \leq \frac{(2q+1)^2(8q+\sqrt{7})}{2^{2\secp}} \leq \frac{5q^2 \cdot 9q}{2^{2\secp}} \leq \frac{45q^3}{2^{2\secp}}
\]
\end{proof}
Thus, by gentle measurement (\cref{lemma:gentle-measurement}), the $\tau$ defined in \cref{subclaim:tau3r} is within trace distance $\frac{3\sqrt{10}q^{3/2}}{2^{\secp}}$ of a state $\tau_{\mathsf{good}}$ in the image of $\bbI - \Pi_\mathsf{bad}^{11/200}$.
The following sub-claim completes the proof of \cref{claim:qot-3r-bb84-hyb34}.
\begin{subclaim}
\label{subclaim:3rseededminentropy}
If $h : \zo^{m} \times \zo^{\leq A\secp} \rightarrow \zo^{\secp}$ is a universal family of hash functions, then conditioned on $\tau$ (defined in \cref{subclaim:tau}) being the image of $\bbI - \Pi_\mathsf{bad}^{11/200}$, and $A = 11\ 700, B=30\ 400$, it holds that 
\[
    |\Pr[\Hyb_3=1]-\Pr[\Hyb_4=1]| \leq \frac{1}{2^{5\secp}}
\]
where $h_b$ is the binary entropy function. 
\end{subclaim}
\begin{proof}
Note that 
\begin{align}
    \ifsubmission \else & \fi 
    \bbI - \Pi_{\mathsf{bad}}^\delta = \sum_{\substack{\bc,\widetilde{\bR},\widetilde{\btheta},\bb,\\U \subseteq [n], |U|=d}}\ket{\bc}\bra{\bc}_{\cC} \otimes 
    \ifsubmission & \else \fi 
    \ket{\widetilde{\bR},\widetilde{\btheta}}\bra{\widetilde{\bR},\widetilde{\btheta}}_{\cZ_1} \otimes \ket{U,\bb}\bra{U,\bb}_{\cZ_2} 
    \ifsubmission \nonumber \\& \else \fi
    \otimes \left(\sum_{\substack{\bR \, : \, \bR_{S'} \neq \widetilde{\bR}_{S'} \text{ or } \Delta\left(\bR_{\overline{T}\setminus U}, \widetilde{\bR}_{\overline{T}\setminus U}\right) < \delta\\\text{where }T = \fsro(\bc),\\S' = \{j \, | \, j \in T \cap U \, \wedge \, \bb_j = \widetilde{\btheta}_{j}\}}}\ket{\bR_{\widetilde{\btheta}}}\bra{\bR_{\widetilde{\btheta}}}_{\cS}\right) \label{eq:3rseeded1}
\end{align}
Since $\tau$ is in the image of $\bbI - \Pi_{\mathsf{bad}}^{11/200}$, by definition the state on register $\cS$ is in a superposition of states as in the summation above. However, note that if $T \neq \fsro(\bc)$ or if $\bR_{S'}\neq \widetilde{\bR}_{S'}$ (where $S' = \{j \, | \, j \in T \cap U \, \wedge \, \bb_j = \widetilde{\btheta}_j \}$), then the sender side check will fail and the two hybrids are perfectly indistinguishable. So, it suffices to analyze states $\tau$ where the register containing $T$ equals $\fsro(\bc)$ and where $\bR_{S'}=\widetilde{\bR}_{S'}$.
Thus, conditioned on the sender not aborting, 
the above equation implies that the register $\cS$ is in superposition of states $\ket{\bR_{\widetilde{\btheta}}}$ s.t.\ $\bR_{S'} = \widetilde{\bR}_{S'}, \Delta(\bR_{\overline{T}\setminus U},\widetilde{\bR}_{\overline{T}\setminus U})<11/200$, for $S'$ as defined above.

Recall that $\tau$ is the state of $\Hyb_3$ (equivalently also $\Hyb_4$) immediately after $\rcv^*$ outputs its message (which includes $\bc$), $\SimExt.\Ext$ is run to get $\widetilde{\bR},\widetilde{\btheta}$, and sender samples the set $U \subseteq [n]$ of size $d$ and bits $\bb \in \zo^n$
next the sender measures register $\cS$. 
Since measurements on different subsystems commute, we may assume that the sender measures the registers $\cS_{{T}\cup U}$ first (recall that we are trying to argue that the remaining registers have entropy). 
Then, by the argument in the previous paragraph, this leaves the remaining registers $\cS_{\overline{T} \setminus {U}}$ in a superposition of states $\ket{\left(\bR_{\overline{T}\setminus U}\right)_{\widetilde{\btheta}}}$ for $\bR_{\overline{T}\setminus U}$ s.t.\  $\Delta(\bR_{\overline{T}\setminus U},\widetilde{\bR}_{\overline{T}\setminus U}) < \frac{11}{200}$.

Next, to obtain $ct_{c}$ for $c \in \zo$, the sender measures registers $\cS_{i,d_i \oplus c}$ in basis $d_i \oplus c$ to obtain a string $\br'_{c} \in \{0,1\}^{|\overline{T}\setminus U|}$. Then, $ct_{c}$ is set as $m_{c} \oplus h(s,\br'_{c})$, where $s$ is uniformly sampled seed for a universal hash function $h$.
Recall, in addition, in $\Hyb_4$ the sender defines $b$ as $\maj \{\widetilde{\btheta}_i \oplus d_i\}_{i \in \overline{T}\setminus U}$. We now prove a lower bound on the quantum min-entropy of $\br'_{b \oplus 1}$, which by the Leftover Hash lemma (\cref{impthm:privacy-amplification}) would imply our claim. 

Consider the subset $W \subseteq \zo \times \overline{T}\setminus U$ defined as $W = \{i, d_i \oplus b \oplus 1\}_{i \in \overline{T}\setminus U}$. Consider again by the commuting property of measurements on different systems that registers $\cS_{\overline{T}\setminus (U \cup W)} = \{\cS_{i, d_i \oplus b}\}_{i \in \overline{T}\setminus U}$ are measured first, leaving the registers $\cS_W = \{\cS_{i, d_i \oplus b \oplus 1}\}_{i \in \overline{T}\setminus U}$ in a superposition of states $\ket{\br_{\widetilde{\btheta}_{W[1]}}}$, where $\Delta\left(\br, \widetilde{\bR}_{W}\right) < 11/200$, and  where $W[1]$ denotes the projection of $W$ on the second set, i.e.\ $W[1] = \{d_i \oplus b \oplus 1\}_{i \in \overline{T}\setminus U}$. Hence, since $\br'_{b \oplus 1}$ is obtained by measuring registers $\cS_W = \{\cS_{i,d_i \oplus b \oplus 1}\}_{i \in \overline{T}\setminus U}$ in basis $d_i \oplus b \oplus 1$, majority of the bits of $\br'_{b \oplus 1}$ are obtained by measuring $\cS_W$ in basis $\widetilde{\btheta}_i \oplus 1$ (since $b$ was defined as $\maj \{\widetilde{\btheta}_i \oplus d_i\}_{i \in \overline{T}\setminus U}$, this means in the majority of the places in $\overline{T} \setminus U$, the following holds: $b = \widetilde{\btheta}_i \oplus d_i \iff d_i \oplus b \oplus 1 = \widetilde{\btheta}_i \oplus 1$).

Therefore, registers $\cS_W$ are in a superposition of states $\ket{\br_{\widetilde{\btheta}_{W[1]}}}$, where $\Delta\left(\br, \widetilde{\bR}_{W}\right) < 11/200$. 
Recall that $\cS_W = \{\cS_{i, d_i \oplus b \oplus 1}\}_{i \in \overline{T}\setminus U}$, then the paragraph above implies that for a majority of $i \in \overline{T}\setminus U$, register $\cS_{i, \widetilde{\btheta}_i \oplus 1}$ is measured in basis $\widetilde{\btheta}_i \oplus 1$. Using \cref{impthm:small-superposition}, we get,
\begin{align*}
    \mathbf{H}_\infty(\br'_{b \oplus 1} \, | \, \cC,\cZ_1,\cZ_2,\cX) &\geq \frac{|\overline{T}\setminus U|}{2} - h_b\left(\frac{11}{200}\right)|\overline{T}\setminus U|\\
    &\geq \frac{n-2k}{2} - h_b\left(\frac{11}{200}\right)(n-k)\\
    &\geq \frac{n-2k}{2} - 0.3073(n-k)
\end{align*}
For $n=(A+B)\secp,k=A\secp,A = 11\ 700, B=30\ 400$, we get,
\begin{align*}
    \mathbf{H}_\infty(\br'_{b \oplus 1} \, | \, \cC,\cZ_1,\cZ_2,\cX) &\geq 9\secp.
\end{align*}
where $h_b$ is the binary entropy function, and we bound the number of strings of length $n$ with relative hamming weight at most $\delta$ by $h_b(\delta)n$. Hence, using the leftover hash lemma (\cref{impthm:privacy-amplification}), $(s,h(s,\br'_{b\oplus 1}))$ is $\frac{1}{2^{5\secp}}$.
\end{proof}
This completes the proof of the claim, as desired.
\end{proof}


\fi

\ifsubmission\vspace{-5mm}\else\fi
\addcontentsline{toc}{section}{References}

\ifsubmission
\bibliographystyle{splncs04}
\else 
\bibliographystyle{alpha}
\fi
\bibliography{abbrev3,custom,crypto,main} 

\ifsubmission\newpage\else\fi

\appendix
\ifsubmission\else\fi
\section{Security of the seedless extractors}\label{sec:adaptive-proof}

In this section, we show the security of the XOR and ROM extractors.

\subsection{XOR extractor}

\begin{theorem}
Let $\cX$ be an $n$-qubit register, and consider any state $\ket{\gamma}_{\cA,\cX}$ that can be written as \[\ket{\gamma} = \sum_{u: \cH\cW(u) < n/2} \ket{\psi_u}_{\cA} \otimes \ket{u}_\cX.\] Let $\rho_{\cA,\cP}$ be the mixed state that results from measuring $\cX$ in the Hadamard basis to produce $x$, and writing $\bigoplus_{i \in [n]}x_i$ into the single qubit register $\cP$. Then it holds that \[\rho_{\cA,\cP} = \Tr_\cX(\ket{\gamma}\bra{\gamma}) \otimes \left(\frac{1}{2}\ket{0}\bra{0} + \frac{1}{2}\ket{1}\bra{1}\right).\] 
\end{theorem}

\begin{proof}
First, write the state on $(\cA,\cX,\cP)$ that results from applying Hadamard to $\cX$ and writing the parity, denoted by $p(x) \coloneqq \bigoplus_{i \in [n]}x_i$, to $\cP$:
\[\frac{1}{2^{n/2}}\sum_{x \in \{0,1\}^n}\left(\sum_{u:\cH\cW(u) < n/2}(-1)^{u \cdot x}\ket{\psi_u}\right)\ket{x}\ket{p(x)}.\]
Then we have that 
\begin{align*}
    \rho_{\cA,\cP} &= \frac{1}{2^n}\sum_{x: p(x) = 0} \left(\sum_{u_1,u_2}(-1)^{(u_1 \oplus u_2)\cdot x}\ket{\psi_{u_1}}\bra{\psi_{u_2}}\right) \otimes \ket{0}\bra{0}\\ &\ \ \ + \frac{1}{2^n}\sum_{x: p(x) = 1} \left(\sum_{u_1,u_2}(-1)^{(u_1 \oplus u_2)\cdot x}\ket{\psi_{u_1}}\bra{\psi_{u_2}}\right) \otimes \ket{1}\bra{1} \\
    &= \frac{1}{2^n}\sum_{u_1,u_2}\ket{\psi_{u_1}}\bra{\psi_{u_2}} \otimes \left(\sum_{x:p(x)=0}(-1)^{(u_1 \oplus u_2)\cdot x}\ket{0}\bra{0} + \sum_{x:p(x)=1}(-1)^{(u_1 \oplus u_2)\cdot x}\ket{1}\bra{1}\right) \\ &= \frac{1}{2^n}\sum_{u_1,u_2}2^{n/2}\delta_{u_1 = u_2} \ket{\psi_{u_1}}\bra{\psi_{u_2}} \otimes \left(\ket{0}\bra{0} + \ket{1}\bra{1}\right) \\
    &= \frac{1}{2}\sum_{u:\cH\cW < n/2}\ket{\psi_u}\bra{\psi_u} \otimes \left(\ket{0}\bra{0} + \ket{1}\bra{1}\right) \\ &= \Tr_\cX(\ket{\gamma}\bra{\gamma}) \otimes \left(\frac{1}{2}\ket{0}\bra{0} + \frac{1}{2}\ket{1}\bra{1}\right),
\end{align*}

where the 3rd equality is due to the following claim, plus the observation that $u_1 \oplus u_2 \neq 1^n$ for any $u_1,u_2$ such that $\cH\cW(u_1),\cH\cW(u_2) < n/2$.

\begin{claim}
For any $u \in \{0,1\}^n$ such that $u \notin \{0^n,1^n\}$, it holds that \[\sum_{x:p(x)=0}(-1)^{u\cdot x} = \sum_{x:p(x)=1}(-1)^{u\cdot x} = 0.\]
\end{claim}

\begin{proof}
For any such $u \notin \{0^n,1^n\}$, define $S_0 = \{i : u_i = 0\}$ and $S_1 = \{i : u_i = 1\}$. Then, for any $y_0 \in \{0,1\}^{|S_0|}$ and $y_1 \in \{0,1\}^{|S_1|}$, define $x_{y_0,y_1} \in \{0,1\}^n$ to be the $n$-bit string that is equal to $y_0$ when restricted to indices in $S_0$ and equal to $y_1$ when restricted to indices in $S_1$. Then,

\begin{align*}
    &\sum_{x:p(x)=0}(-1)^{u \cdot x} = \sum_{y_1 \in \{0,1\}^{|S_1|}}\sum_{y_0 \in \{0,1\}^{|S_0|}: p(x_{y_0,y_1})=0} (-1)^{u \cdot x_{y_0,y_1}}\\ &= \sum_{y_1 \in \{0,1\}^{|S_1|}}2^{|S_0|-1}(-1)^{1^{|S_1|} \cdot y_1} = 2^{|S_0|-1}\sum_{y_1 \in \{0,1\}^{|S_1|}}(-1)^{p(y_1)} = 0,
\end{align*}
and the same sequence of equalities can be seen to hold for $x : p(x) = 1$.
\end{proof}

\end{proof}

\subsection{RO extractor}

\begin{theorem}
Let $H : \{0,1\}^n \to \{0,1\}^m$ be a uniformly random function, and let $q,C,k$ be integers. Consider a two-stage oracle algorithm $(A_1^H, A_2^H)$ that combined makes at most $q$ queries to $H$. Suppose that $A_1^H$ outputs classical strings $(T,\{x_i\}_{i \in T})$, and let $\ket{\gamma}_{\cA,\cX}$ be its left-over quantum state,\footnote{That is, consider sampling $H$, running a purified $A_1^H$, measuring at the end to obtain $(T,\{x_i\}_{i \in T})$, and then defining $\ket{\gamma}$ to be the left-over state on $\cA$'s remaining registers.} where $T \subset [n]$ is a set of size $n-k$, each $x_i \in \{0,1\}$, $\cA$ is a register of arbitary size, and $\cX$ is a register of $k$ qubits. Suppose further that with probability 1 over the sampling of $H$ and the execution of $A_1$, there exists a set $L \subset \{0,1\}^k$ of size at most $C$ such that $\ket{\gamma}$ may be written as follows: \[\ket{\gamma} = \sum_{u \in L}\ket{\psi_u}_\cA \otimes \ket{u}_\cX.\] Now consider the following two games.

\begin{itemize}
    \item $\mathsf{REAL}$: 
    \begin{itemize}
        \item $A_1^H$ outputs $T,\{x_i\}_{i \in T},\ket{\gamma}_{\cA,\cX}$.
        \item $\cX$ is measured in the Hadamard basis to produce a $k$-bit string which is parsed as $\{x_i\}_{i \in \overline{T}}$, and a left-over state $\ket{\gamma'}_{\cA}$ on register $\cA$. Define $x = (x_1,\dots,x_n)$. 
        \item $A_2^H$ is given $T, \{x_i\}_{i \in T},\ket{\gamma'}_{\cA}, H(x)$, and outputs a bit.
    \end{itemize}
    \item $\mathsf{IDEAL}$: 
    \begin{itemize}
        \item $A_1^H$ outputs $T,\{x_i\}_{i \in T},\ket{\gamma}_{\cA,\cX}$.
        \item $r \gets \{0,1\}^m$.
        \item $A_2^H$ is given $T, \{x_i\}_{i \in T},\Tr_\cX(\ket{\gamma}\bra{\gamma}),r$, and outputs a bit.
    \end{itemize}
\end{itemize}
Then, \[\left|\Pr[\mathsf{REAL} = 1] - \Pr[\mathsf{IDEAL} = 1]\right| \leq \frac{2\sqrt{q}C + 2q\sqrt{C}}{2^{k/2}} < \frac{4qC}{2^{k/2}}.\]
\end{theorem}

\begin{proof}

The proof follows via two steps. First, we define a HYBRID distribution where we re-program the random oracle at input $x$ to a uniformly random string $r$, and argue that the adversary cannot notice, even given $x$. Intuitively, this is establishing that $H(x)$ must have been quite close to uniformly random from the adversary's perspective at the point that $x$ is measured (on average over $x$). This requires a new ``adaptive re-programming'' lemma for the QROM, where the point $x$ that is adaptively re-programmed may be sampled from a \emph{quantum} source of entropy. As mentioned in the introduction, all previous adaptive re-programming lemmas have only handled classical entropy sources. Second, we ``undo'' the re-programming of $H(x)$, but still output (uniformly random) $r$ as the extracted string. Indistinguishability of these two games, on the other hand, can be established via a one-way-to-hiding lemma, since in the final game, the adversary is given no information at all about the measured string $x$. In particular, it suffices to use the fact that $x$ has high quantum min-entropy conditioned on the adversary's state to argue that the adversary cannot guess $x$ and thus cannot notice whether or not $H(x)$ was re-programmed.

Now we formalize this strategy. Consider the following hybrid game.
\begin{itemize}
    \item $\mathsf{HYBRID}$: 
    \begin{itemize}
        \item $A_1^H$ outputs $T,\{x_i\}_{i \in T},\ket{\gamma}_{\cA,\cX}$.
        \item $\cX$ is measured in the Hadamard basis to produce a $k$-bit string which is parsed as $\{x_i\}_{i \in \overline{T}}$, and a left-over state $\ket{\gamma'}_{\cA}$ on register $\cA$. Define $x = (x_1,\dots,x_n)$. Sample $r \gets \{0,1\}^m$, and re-program $H(x)$ to $r$.
        \item $A_2^H$ is given $T, \{x_i\}_{i \in T},\ket{\gamma'}_{\cA}, r$, and outputs a bit.
    \end{itemize}
\end{itemize}

The theorem follows by combining the two following claims.

\begin{claim}
\[|\Pr[\mathsf{REAL} = 1] - \Pr[\mathsf{HYBRID} = 1]| \leq \frac{2\sqrt{q}C}{2^{k/2}}.\]
\end{claim}

\begin{proof}
Consider purifying the random oracle $H$ on register $\cR$, and let $\ket{\widehat{\gamma}}_{\cR,\cA,\cX}$ be the left-over state of $A_1$ and the random oracle in $\mathsf{REAL}$ or $\mathsf{HYBRID}$ after $A_1$ outputs $(T,\{x_i\}_{i \in T})$. By \cref{impthm:superposition-oracle}, the state $\ket{\widehat{\gamma}}_{\cR,\cA,\cX}$ satisfies the premise of \cref{lemma:adaptive-reprogramming} below, where $\cF$ is the set of $2^k$ sub-registers of $\cR$ corresponding to each $x' \in \{0,1\}^n$ such that $x'_i = x_i$ for all $i \in T$.

Now, consider a reduction that receives the state $\rho^{\mathsf{REAL}}$ or $\rho^{\mathsf{REPROG}}$ from \cref{lemma:adaptive-reprogramming}, measures $\cF_{x}$ in the computational basis to obtain $H(x)$, and then continues to run $A_2$ on input $H(x)$ (along with $A_1$'s state on $\cA$ and $T$, $\{x_i\}_{i \in T}$). In the case of $\rho^{\mathsf{REAL}}$, this exactly matches the $\mathsf{REAL}$ game and in the case of $\rho^{\mathsf{REPROG}}$, this exactly matches the $\mathsf{HYBRID}$ game, using the fact that $\ket{\phi_0}$ (as defined in \cref{impthm:superposition-oracle}) is the uniform superposition state.

\end{proof}

\begin{claim}
\[|\Pr[\mathsf{HYBRID} = 1] - \Pr[\mathsf{IDEAL} = 1]| \leq \frac{2q\sqrt{C}}{2^{k/2}}.\]
\end{claim}

\begin{proof}
This follows from an invocation of \cref{impthm:ow2h}. Consider the distribution over $(S,O_1,O_2,\ket{\psi})$ that results from the following. 
\begin{itemize}
    \item Sample $O_1$ as a random oracle $H$,
    \item run $A_1^H$ to obtain $T,\{x_i\}_{i \in T},\ket{\gamma}_{\cA,\cX}$,
    \item measure to obtain $\{ x_i \}_{ i\in \overline{T}}$ as in the $\mathsf{HYBRID}$ game, define $x = (x_1,\dots,x_n)$, and define $S = \{x\}$, 
    \item sample $r \gets \{0,1\}^m$, and let $O_2$ be the same as $O_1$, except that $O_2(x) = r$,
    \item let $\ket{\psi}$ be the resulting state on register $\cA$ along with the classical information $(T,\{x_i\}_{i \in T},r)$.
\end{itemize}

Then, $\Pr[\mathsf{IDEAL} = 1] = P_{\text{left}}$ and $\Pr[\mathsf{HYBRID} = 1] = P_{\text{right}}$, so it suffices to bound $P_{\text{guess}}$. By \cref{impthm:conditional-min-entropy}, $P_{\text{guess}}$ is upper bounded by $1/2^{\ell}$, where $\ell$ is the quantum conditional min-entropy of $\{x_i\}_{i \in \overline{T}}$ given register $\cA$. By \cref{impthm:small-superposition}, and the fact that measuring an unentangled $k$-bit standard basis vector in the Hadamard basis gives $k$ bits of quantum conditional min-entropy, $P_\text{guess}$ is upper bounded by $\frac{C}{2^{k}}$. Thus, \cref{impthm:ow2h} gives the final bound of $\frac{2q\sqrt{C}}{2^{k/2}}$.


\end{proof}

\end{proof}

\subsection{The superposition oracle}

Following \cite{C:Zhandry19,10.1007/978-3-030-92062-3_22}, we will use the fact that a quantum accesible random oracle $H: \{0,1\}^n \to \{0,1\}^m$ can be implemented as follows.
\begin{itemize}
    \item Let $\cF$ be a $(m \cdot 2^n)$-qubit register split into $2^n$ subregisters $\{\cF_x\}_{x \in \{0,1\}^n}$ of size $m$. Let $\ket{\phi_0}$ be the uniform superposition state. Prepare an initial state \[\ket{\Psi}_\cF = \bigotimes_{x \in \{0,1\}^n} \ket{\phi_0}_{{\cF_x}}.\]
    \item A query on registers $\cX,\cY$ is answered with a unitary $O_{\cX,\cY,\cF}$ such that \[O_{\cX,\cY,\cF}\ket{x}\bra{x}_\cX = \ket{x}\bra{x}_\cX \otimes (\mathsf{CNOT}^{\otimes m})_{\cF_x : \cY}.\]
    \item Register $\cF$ is measured to obtain a random function $H$.
\end{itemize}

\begin{importedtheorem}[\cite{EC:AMRS20,10.1007/978-3-030-92062-3_22}]\label{impthm:superposition-oracle}
Let $\ket{\psi_q}_{\cA,\cF}$ be the joint adversary-oracle state state after an adversary has made $q$ queries to the superposition oracle on register $\cF$. Then this state can be written as \[\ket{\psi_q}_{\cA,\cF} = \sum_{S \subset \{0,1\}^n, |S| \leq q} \ket{\psi_{q,S}}_{\cA,\cF_S} \otimes \left(\ket{\phi_0}^{\otimes (2^n - |S|)}\right)_{\cF_{\overline{S}}},\] where $\ket{\psi_{q,S}}$ are such that $\bra{\phi_0}_{\cF_x}\ket{\psi_{q,S}}_{\cA,\cF_S} = 0$ for all $x \in S$.
\end{importedtheorem}

\subsection{Re-programming} In this section, we prove the following lemma.

\begin{lemma}\label{lemma:adaptive-reprogramming}
Let $\ket{\phi_0}$ be an $m$-qubit unit vector, and let $\cF$ be a $(m \cdot 2^k)$-qubit register split into $2^k$ sub-registers $\{\cF_x\}_{x \in \{0,1\}^k}$ of $m$ qubits. Let $\cA$ be an arbitrary register and $\cX$ be an $k$-qubit register. Consider any state $\ket{\gamma}_{\cF,\cA,\cX}$, set $L \subseteq \{0,1\}^k$, and integer $q \in \bbN$, such that $\ket{\gamma}$ can be written as

\[\ket{\gamma} = \sum_{u \in L} \ket{\widehat{\psi}_u}_{\cF,\cA} \otimes \ket{u}_\cX,\]

where each 

\[\frac{\ket{\widehat{\psi}_u}}{\|\ket{\widehat{\psi}_u}\|} = \sum_{S \subset \{0,1\}^n : |S| \leq q} \ket{\psi_{u,S}}_{\cA,\cF_S} \otimes \left(\ket{\phi_0}^{(2^k-|S|)}\right)_{\cF_{\overline{S}}},\] and each $\ket{\psi_{u,S}}$ is orthogonal to $\ket{\phi_0}_{\cF_x}$ for all $x \in S$. Let

\begin{itemize}
    \item $\rho^{\mathsf{REAL}}_{\cF,\cA,\cX}$ be the mixed state that results from measuring $\cX$ in the Hadamard basis to produce $x \in \{0,1\}^k$ and a left-over state $\ket{\gamma_x}_{\cF,\cA}$, and outputting $\ket{\gamma_x}\bra{\gamma_x} \otimes \ket{x}\bra{x}$, and
    \item $\rho^{\mathsf{REPROG}}_{\cF,\cA,\cX}$ be the mixed state that results from measuring $\cX$ in the Hadamard basis to produce $x \in \{0,1\}^k$ and a left-over state $\ket{\gamma_x}_{\cF,\cA}$, and outputting $\Tr_{\cF_x}\left(\ket{\gamma_x}\bra{\gamma_x}\right) \otimes \ket{\phi_0}\bra{\phi_0}_{\cF_x} \otimes \ket{x}\bra{x}$.
\end{itemize}

Then, \[\TD\left(\rho^{\mathsf{REAL}},\rho^{\mathsf{REPROG}}\right) \leq \frac{2\sqrt{q}|L|}{2^{k/2}}.\]


\end{lemma}


\begin{proof}
 
For each $u \in L$, let $a_u \coloneqq \|\ket{\widehat{\psi}_u}\|$, and $\ket{\psi_u} \coloneqq \ket{\widehat{\psi}_u}/a_u$. Consider applying the Hadamard transform to register $\cX$ of $\ket{\gamma}$, producing

\[\sum_{x \in \{0,1\}^k,u \in L}\frac{(-1)^{u \cdot x}a_u}{2^{k/2}}\ket{\psi_u}_{\cF,\cA} \otimes \ket{x}_{\cX} \coloneqq \sum_{x \in \{0,1\}^k}\ket{\gamma_x}_{\cF,\cA} \otimes \ket{x}_{\cX}\] and then measuring in the computational basis to produce $x$ and left-over state $\ket{\gamma_x}\bra{\gamma_x}_{\cF,\cA}.$

The lemma asks to bound the following quantity.


\begin{align*}
\frac{1}{2}&\Bigg\|\sum_{x \in \{0,1\}^k}\ket{\gamma_x}\bra{\gamma_x}_{\cF,\cA} \otimes \ket{x}\bra{x} - \sum_{x \in \{0,1\}^k}\Tr_{\cF_{x}}\left(\ket{\gamma_x}\bra{\gamma_x}_{\cF,\cA}\right) \otimes \ket{\phi_0}\bra{\phi_0}_{\cF_{x}} \otimes \ket{x}\bra{x} \Bigg\|_1 \\
&= \frac{1}{2}\Bigg\|\frac{1}{2^k}\sum_{x \in \{0,1\}^k}\sum_{u_1,u_2 \in L}(-1)^{(u_1 \oplus u_2) \cdot x}a_{u_1}a_{u_2}\ket{\psi_{u_1}}\bra{\psi_{u_2}} \otimes \ket{x}\bra{x} \\ 
&\ \ \ - \frac{1}{2^k}\sum_{x \in \{0,1\}^k}\Tr_{\cF_x}\left(\sum_{u_1,u_2 \in L}(-1)^{(u_1 \oplus u_2) \cdot x}a_{u_1}a_{u_2}\ket{\psi_{u_1}}\bra{\psi_{u_2}}\right) \otimes \ket{\phi_0}\bra{\phi_0}_{\cF_x} \otimes \ket{x}\bra{x}\Bigg\|_1 \\ 
&\leq \frac{1}{2^{k+1}}\sum_{u_1,u_2 \in L}a_{u_1}a_{u_2}\left\|\sum_{x \in \{0,1\}^k}\left(\ket{\psi_{u_1}}\bra{\psi_{u_2}} \otimes \ket{x}\bra{x}\right) - \left(\Tr_{\cF_x}\left(\ket{\psi_{u_1}}\bra{\psi_{u_2}}\right)\otimes \ket{\phi_0}\bra{\phi_0}_{\cF_x} \otimes \ket{x}\bra{x}\right)\right\|_1 \\ 
 &\leq \frac{1}{2^{k+1}}\sum_{u_1,u_2 \in L}a_{u_1}a_{u_2}\sum_{x \in \{0,1\}^n}\left\|\ket{\psi_{u_1}}\bra{\psi_{u_2}} - \Tr_{\cF_x}\left(\ket{\psi_{u_1}}\bra{\psi_{u_2}}\right)\otimes \ket{\phi_0}\bra{\phi_0}_{\cF_x}\right\|_1,\end{align*}

where the inequalities follow from the triangle inequality. Now, following the proof of \cite[Theorem 6]{10.1007/978-3-030-92062-3_22}, for any $u_1,u_2,$ and $x$, we can write \begin{align*}\ket{\psi_{u_1}}\bra{\psi_{u_2}} = &\bra{\phi_0}_{\cF_x}\ket{\psi_{u_1}}\bra{\psi_{u_2}}\ket{\phi_0}_{\cF_x} \otimes \ket{\phi_0}\bra{\phi_0}_{\cF_x} + \ket{\psi_{u_1}}\bra{\psi_{u_2}}(\bbI - \ket{\phi_0}\bra{\phi_0}_{\cF_x})\\ &+ (\bbI - \ket{\phi_0}\bra{\phi_0}_{\cF_x})\ket{\psi_{u_1}}\bra{\psi_{u_2}}\ket{\phi_0}\bra{\phi_0}_{\cF_x}\end{align*} and 
\begin{align*}
    \Tr_{\cF_x}\left(\ket{\psi_{u_1}}\bra{\psi_{u_2}}\right)\otimes \ket{\phi_0}\bra{\phi_0}_{\cF_x} = &\bra{\phi_0}_{\cF_x}\ket{\psi_{u_1}}\bra{\psi_{u_2}}\ket{\phi_0}_{\cF_x} \otimes \ket{\phi_0}\bra{\phi_0}_{\cF_x}\\ &+ \Tr_{\cF_x}\left((\bbI - \ket{\phi_0}\bra{\phi_0}_{\cF_x})\ket{\psi_{u_1}}\bra{\psi_{u_2}}\right) \otimes \ket{\phi_0}\bra{\phi_0}_{\cF_x}
\end{align*}
so
\begin{align*}
    &\left\|\ket{\psi_{u_1}}\bra{\psi_{u_2}} - \Tr_{\cF_x}\left(\ket{\psi_{u_1}}\bra{\psi_{u_2}}\right)\otimes \ket{\phi_0}\bra{\phi_0}_{\cF_x}\right\|_1 \leq \|\ket{\psi_{u_1}}\bra{\psi_{u_2}}(\bbI - \ket{\phi_0}\bra{\phi_0}_{\cF_x})\|_1 \\ &+ \|(\bbI - \ket{\phi_0}\bra{\phi_0}_{\cF_x})\ket{\psi_{u_1}}\bra{\psi_{u_2}}\ket{\phi_0}\bra{\phi_0}_{\cF_x}\|_1 + \|\Tr_{\cF_x}\left((\bbI - \ket{\phi_0}\bra{\phi_0}_{\cF_x})\ket{\psi_{u_1}}\bra{\psi_{u_2}}\right) \otimes \ket{\phi_0}\bra{\phi_0}_{\cF_x}\|_1.
\end{align*}


Now, for each $x,u$, define $\alpha_{x,u} = \|\bra{\phi_0}_{\cF_x}\ket{\psi_{u}}\|$. The first term above simplifies as 
\begin{align*}
    \|\ket{\psi_{u_1}}\bra{\psi_{u_2}}(\bbI - \ket{\phi_0}\bra{\phi_0}_{\cF_x})\|_1 = \|\bra{\psi_{u_2}}(\bbI - \ket{\phi_0}\bra{\phi_0}_{\cF_x})\| = \sqrt{1-\bra{\psi_{u_2}}\ket{\phi_0}\bra{\phi_0}\ket{\psi_{u_2}}} = \sqrt{1-\alpha_{x,u_2,}^2}.
\end{align*}

The second term above simplifies as 

\begin{align*}
    &\|(\bbI - \ket{\phi_0}\bra{\phi_0}_{\cF_x})\ket{\psi_{u_1}}\bra{\psi_{u_2}}\ket{\phi_0}\bra{\phi_0}_{\cF_x}\|_1 \leq \|(\bbI - \ket{\phi_0}\bra{\phi_0}_{\cF_x})\ket{\psi_{u_1}}\bra{\psi_{u_2}}\|_1 \\ &= \|(\bbI - \ket{\phi_0}\bra{\phi_0}_{\cF_x})\ket{\psi_{u_1}}\| = \sqrt{1-\bra{\psi_{u_1}}\ket{\phi_0}\bra{\phi_0}\ket{\psi_{u_1}}} = \sqrt{1-\alpha_{x,u_1}^2},
\end{align*}

where the inequality is Holder's inequality. The third term simplifies as

\begin{align*}
    &\left\|\Tr_{\cF_x}\left((\bbI - \ket{\phi_0}\bra{\phi_0}_{\cF_x})\ket{\psi_{u_1}}\bra{\psi_{u_2}}\right) \otimes \ket{\phi_0}\bra{\phi_0}_{\cF_x}\right\|_1 =  \left\|\Tr_{\cF_x}\left((\bbI - \ket{\phi_0}\bra{\phi_0}_{\cF_x})\ket{\psi_{u_1}}\bra{\psi_{u_2}}\right)\right\|_1 \\ &\leq \left\|(\bbI - \ket{\phi_0}\bra{\phi_0}_{\cF_x})\ket{\psi_{u_1}}\bra{\psi_{u_2}}\right\|_1 = \sqrt{1-\alpha^2_{x,u_1}},
\end{align*}

where the inequality is the following fact from \cite{Holevo+2019}: for any bounded operator $T$ on $\cA \otimes \cB$, $\|\Tr_\cB(T)\|_1 \leq \|T\|_1$.

Thus, the distinguishing advantage can be bounded by 

\begin{align*}
    &\frac{1}{2^{k+1}}\sum_{u_1,u_2 \in L}a_{u_1}a_{u_2}\sum_{x \in \{0,1\}^k}2\sqrt{1-\alpha^2_{x,u_1}} + \sqrt{1-\alpha^2_{x,u_2}}  \\ &\leq \frac{1}{2^{k}}\sum_{u_1,u_2 \in L}a_{u_1}a_{u_2}\left(\sum_{x \in \{0,1\}^k}\sqrt{1-\alpha^2_{x,u_1}} + \sum_{x \in \{0,1\}^k}\sqrt{1-\alpha^2_{x,u_2}}\right)
    \\ &\leq \frac{1}{2^{k}}\sum_{u_1,u_2 \in L}a_{u_1}a_{u_2}\left(\sqrt{2^k\left(2^k - \sum_{x \in \{0,1\}^k}\alpha^2_{x,u_1}\right)}+\sqrt{2^k\left(2^k - \sum_{x \in \{0,1\}^k}\alpha^2_{x,u_2}\right)}\right),
\end{align*}

    


where the second inequality follows from Cauchy-Schwartz. 



Now, for any $u$,

\begin{align*}
    &\sum_{x \in \{0,1\}^k}\alpha^2_{x,u} = \sum_{x \in \{0,1\}^k} \| \bra{\phi_0}_{\cF_x}\ket{\psi_{u}}\|^2 \|\\ &= \sum_{x \in \{0,1\}^k}\bigg\|\sum_{S \subset \{0,1\}^k : |S| \leq q} \bra{\phi_0}_{\cF_x}\ket{\psi_{u,S}}_{\cA,\cF_S} \otimes \left(\ket{\phi_0}^{(2^k-|S|)}\right)_{\cF_{\overline{S}}}\bigg\|^2 \\ &= \sum_{x \in \{0,1\}^k}\bigg\|\sum_{S \not\owns x} \ket{\psi_{u,S}}_{\cA,\cF_S} \otimes \left(\ket{\phi_0}^{(2^k-|S|)}\right)_{\cF_{\overline{S}}}\bigg\|^2,
\end{align*}

where the last equality follows since $\ket{\phi_0}_{\cF_x}$ is orthogonal to $\ket{\psi_{u,S}}$ for all $x \in S$, and $\ket{\phi_0}$ is normalized. Now, since each of the summands in the inner summation are pairwise orthogonal (so that we can move the summation outside of the norm), we can write

\begin{align*} 
    &\sum_{x \in \{0,1\}^k}\bigg\|\sum_{S \not\owns x} \ket{\psi_{u,S}}_{\cA,\cF_S} \otimes \left(\ket{\phi_0}^{(2^k-|S|)}\right)_{\cF_{\overline{S}}}\bigg\|^2 = \sum_{S}\sum_{x \notin S}\bigg\|\ket{\psi_{u,S}}_{\cA,\cF_{S}} \otimes \left(\ket{\phi_0}^{(2^k-|S|)}\right)_{\cF_{\overline{S}}}\bigg\|^2 \\ &\geq (2^k-q)\sum_{S}\bigg\|\ket{\psi_{u,S}}_{\cA,\cF_{S}} \otimes \left(\ket{\phi_0}^{(2^k-|S|)}\right)_{\cF_{\overline{S}}}\bigg\|^2 = (2^k-q)\bigg\|\sum_S \ket{\psi_{u,S}}_{\cA,\cF_{S}} \otimes \left(\ket{\phi_0}^{(2^k-|S|)}\right)_{\cF_{\overline{S}}}\bigg\|^2 \\ &= (2^k - q)\|\ket{\psi_{u}}\|^2 = 2^k - q.
\end{align*}

Thus, the distinguishing advantage can be bounded by

\begin{align*}
    \frac{2}{2^{k}}\sqrt{2^k \cdot q}\sum_{u_1,u_2 \in L}a_{u_1}a_{u_2} \leq \frac{2\sqrt{q}|L|}{2^{k/2}},
\end{align*}

since, by Cauchy-Schwartz and the fact that $\sum_{u \in L} a_u^2 = 1$, we can bound $\sum_{u \in L}a_u \leq \sqrt{|L|}$.

\end{proof}
\ifsubmission\ifsubmission\section{Commitments}\label{sec:com-proofs}\else \fi
\subsection{Construction} \label{sec:eecom-construction}

\protocol
{\proref{fig:uccomm-protocol}}
{Extractable and equivocal commitment scheme}
{fig:uccomm-protocol}
{
Parameters: security parameter $\secp$, number of commitments $n = n(\lambda)$\\
Random oracle: $H:\zo^{\secp+1} \rightarrow \zo^{\secp+1}$.

\begin{itemize}
    \item {$\Com^{H}(1^\secp,\{b_i\}_{i \in [n]})$}: For all $i \in [n]$, sample $r_i \leftarrow \zo^\secp$ and set $\com_i = H(b_i||r_i)$. Set $\state = \{b_i,r_i\}_{i \in [n]}$ and output $(\state, \{\com_i\}_{i \in [n]})$.
    
    \item $\Open^{H}(\state, T)$: Parse $\state$ as $\{b_i,r_i\}_{i \in [n]}$ and output $\{b_i,r_i\}_{i \in T}$. 
    
    \item $\Rec^{H}(\{\com_i\}_{i \in [n]},T,\{b_i,r_i\}_{i \in T})$: Output $\bot$ if there exists $i \in T$ s.t.\ $H(b_i||r_i) \neq \com_i$. Otherwise output $\{b_i\}_{i \in T}$.
\end{itemize}
}

We construct extractable and equivocal bit commitments in the QROM in \cref{fig:uccomm-protocol}. Without loss of generality, a committer can commit to strings of length $>1$ by committing to each bit in the string one by one, and sending all commitments in parallel.\else\fi
\ifsubmission\else\fi
\ifsubmission\else\fi
\ifsubmission\else\fi
\ifsubmission\else\fi
\ifsubmission\section{The random basis framework}
\label{sec:ot-eecom}




In this section, we obtain three round OT realizing $\cF_{\SROT}$, and we provide a modification that yields four round chosen input $\cF_{\OT[\secp]}$. The constructions make use of standard BB84 states, therefore we refer to this as the random basis framework.

\begin{theorem}[Three round random-sender-input OT.]\label{thm:3-round-ROT}
Instantiate \proref{fig:qot-4r-const} with any non-interactive commitment scheme that is \emph{extractable} (\cref{def:fextractable}) and \emph{equivocal} (\cref{def:fequivocal}). Then the following hold.
\begin{itemize}
    \item When instantiated with the XOR extractor, there exist constants $A,B$ such that \proref{fig:qot-4r-const} securely realizes (\cref{def:secure-realization}) $\cF_{\SROT[1]}$.
    \item When instantiated with the ROM extractor, there exist constants $A,B$ such that \proref{fig:qot-4r-const} securely realizes (\cref{def:secure-realization}) $\cF_{\SROT[\secp]}$.
\end{itemize}

Furthermore, letting $\secp$ be the security parameter, $q$ be an upper bound on the total number of random oracle queries made by the adversary, and using the commitment scheme from \cref{sec:eecom-construction} with security parameter $\secp_\com = 2\secp$, the following hold. 
\begin{itemize}
    \item When instantiatied with the XOR extractor and constants $A = 1100$, $B= 500$, \proref{fig:qot-4r-const} securely realizes $\cF_{\SROT[1]}$ with $\mu_{\sR^*}$-security against a malicious receiver and $\mu_{\sS^*}$-security against a malicious sender, where 
    \[
        \mu_{\sR^*} = \frac{\sqrt{5}+1}{2^{\secp}}+\frac{148(q + 3200\secp + 1)^3 + 1}{2^{2\secp}} + \frac{25600q\secp}{2^{\secp}},\quad \mu_{\sS^*} = \frac{114q\sqrt{\secp}}{2^{ \secp}}.
    \]
    This requires a total of $(A+B)\secp = 1600\secp$ BB84 states.
    \item When instantiated with the ROM extractor and constants $A = 11\ 000, B = 12\ 000$, \proref{fig:qot-4r-const} securely realizes $\cF_{\SROT[\secp]}$ with $\mu_{\sR^*}$-security against a malicious receiver and $\mu_{\sS^*}$-security against a malicious sender, where 
    \[
        \mu_{\sR^*} = \frac{\sqrt{5}}{2^{\secp}}+ \frac{4q}{2^{18\secp}}+\frac{148(q + 46000\secp + 1)^3 + 1}{2^{2\secp}} + \frac{368000q\secp}{2^{\secp}}, \quad \mu_{\sS^*} =\frac{430q\sqrt{\secp}}{2^{\secp}}.
    \] 
    This requires a total of $(A+B)\secp = 23\ 000\secp$ BB84 states.
\end{itemize}
\end{theorem}

\begin{theorem}[Four Round chosen input string OT.]
\label{thm:4-round-chosenOT}
Instantiate \proref{fig:qot-4r-actual} with any non-interactive commitment scheme that is \emph{extractable} (\cref{def:fextractable}) and \emph{equivocal} (\cref{def:fequivocal}). Then there exist constants $A,B$ such that \proref{fig:qot-4r-actual} securely realizes (\cref{def:secure-realization}) $\cF_{\OT[\secp]}$. 
    
Furthermore, letting $\secp$ be the security parameter, $q$ be an upper bound on the total number of random oracle queries made by the adversary, and using the commitment scheme from \cref{sec:eecom-construction} with security parameter $\secp_\com = 2\secp$, for constants $A = 5300, B = 5000$, \proref{fig:qot-4r-actual} securely realizes $\cF_{\OT[\secp]}$ with $\mu_{\sR^*}$-security against a malicious receiver and $\mu_{\sS^*}$-security against a malicious sender, where 
\[
    \mu_{\sR^*} = \frac{\sqrt{5}}{2^\secp} + \frac{1}{2^{9\secp}} + \frac{148(q + 2n + 1)^3 + 1}{2^{2\secp}} + \frac{16qn}{2^{\secp}}, \quad \mu_{\sS^*} = \frac{288q\sqrt{\secp}}{2^{\secp}}.
\]
This requires a total of $(A+B)\secp = 10\ 300 \secp$ BB84 states.
\end{theorem}

\protocol
{\proref{fig:qot-4r-const}}
{Three-round OT protocol realizing $\cF_{\SROT[1]}$.} 
{fig:qot-4r-const}
{
{\bf Ingredients, parameters and notation.}
\vspace{-2mm}
\begin{itemize}
\item Security parameter $\secp$ and constants $A,B$. Let $n = (A+B)\secp$ and $k = A\secp$.
\item For classical bits $(x,\theta)$, let $\ket{x}_\theta$ denote $\ket{x}$ if $\theta = 0$, and $(\ket{0} + (-1)^x\ket{1})/\sqrt{2}$ if $\theta = 1$.
\item A non-interactive extractable and equivocal commitment scheme $(\Com,\allowbreak\Open,\allowbreak\Rec)$, where commitments to 2 bits have size $\ell \coloneqq \ell(\secp)$.
\item An extractor $E$ with domain $\{0,1\}^{n-k}$ which is either
\begin{itemize}
    \item The XOR function, so $E(r_1,\dots,r_{n-k}) = \bigoplus_{i \in [n-k]} r_i$.
    \item A random oracle $\extro : \{0,1\}^{n-k} \to \{0,1\}^\secp$.
\end{itemize}
\end{itemize}
{\bf Receiver input}: $b \in \{0,1\}, m \in \{0, 1\}^z$, where $z$ is the output length of the extractor.
\begin{enumerate}
\item {\bf Receiver message.} $\sR$ performs the following steps.
    \begin{enumerate}
        \item Choose $x \leftarrow \{0, 1\}^n$, $\theta \leftarrow \{0, 1\}^n$ and prepare the states $\{\ket{x_i}_{\theta_i}\}_{i \in [n]}$. 
        \item Compute $\left(\state,\{c_i\}_{i \in [n]}\right) \gets \Com\left(\{(x_i, \theta_i)\}_{i \in [n]}\right)$, and send $\{\ket{x_i}_{\theta_i},c_i\}_{i \in [n]}$ to $\sS$.
    \end{enumerate}

\item {\bf Sender message.} \label{itm:4rqot-recv-second-msg} $\sS$ performs the following steps.
    \begin{enumerate}
        \item Choose $\widehat{\theta} \leftarrow \{0, 1\}^n$. For all $i \in [n]$, measure $\ket{x_i}_{\theta_i}$ in basis $\widehat{\theta_i}$ to get outcome $\widehat{x_i}$. 
        \item Sample a random subset $T$ of $[n]$ of size $k$.
        Send $T, \{\widehat{\theta}_i\}_{i \in \overline{T}}$ to $\sR$, where $\overline{T} \coloneqq [n]\setminus T$.
    \end{enumerate}
    
\item {\bf Receiver message.} $\sR$ performs the following steps.
    \begin{enumerate}
        \item Divide $\overline{T}$ into 2 disjoint subsets $S_0, S_1$ as follows. Set $I_b = S_0$ and $I_{1-b} = S_1$, where ${S_0 = \{ i \ | \ i \in \overline{T} \wedge \theta_i = \widehat{\theta}_i \}}, {S_1 = \{ i \ | \ i \in \overline{T} \wedge \theta_i \ne \widehat{\theta}_i \}}$. 
        \item Compute $\{(x_i, \theta_i), u_i\}_{i \in T} \gets \Open(\state, T)$.
        
        \item Compute $X$ as the concatenation of $\{x_i\}_{i \in I_b}$, set $r_b := E(X) \oplus m$, and sample $r_{1-b} \gets \zo^z$, where $z$ is the output length of extractor.
        
        \item \label{step:qot4r-recv-openings} Send $I_0,I_1,\{(x_i, \theta_i), u_i\}_{i \in T}, (r_0, r_1)$ to $\sS$.
    \end{enumerate}
    


        



\item {\bf Output computation} $\sS$ does the following: 
            \begin{enumerate}
                \item Abort if $\Rec(\{c_i\}_{i \in T}, \{(x_i, \theta_i), u_i\}_{i \in T}) = \bot$ or if $\exists i \in T$ s.t.\ $\theta_i = \widehat{\theta}_i$ but $x_i \neq \widehat{x}_i$. 
                \item Compute $\widehat{X}_0, \widehat{X}_1$ as the concatenation of $\{\widehat{x_i}\}_{i \in I_0},\{\widehat{x_i}\}_{i \in I_1}$ respectively. Output $m_0 := E(\widehat{X}_0) \oplus r_0$ and $m_1 := E(\widehat{X}_1) \oplus r_1$.
            \end{enumerate}
\end{enumerate}
}

\protocol
{\proref{fig:qot-4r-actual}}
{Four-round OT protocol realizing $\cF_{\OT[\secp]}$. Parts in \nblue{blue} are different than the 3-round OT protocol realizing $\cF_{\SROT[1]}$ in \proref{fig:qot-4r-const}.}
{fig:qot-4r-actual}
{
{\bf Ingredients, parameters and notation.}
\vspace{-2mm}
\begin{itemize}
\item Security parameter $\secp$ and constants $A,B$. Let $n = (A+B)\secp$ and $k = A\secp$.
\item For classical bits $(x,\theta)$, let $\ket{x}_\theta$ denote $\ket{x}$ if $\theta = 0$, and $(\ket{0} + (-1)^x\ket{1})/\sqrt{2}$ if $\theta = 1$.
\item A non-interactive extractable and equivocal commitment scheme $(\Com,\allowbreak\Open,\allowbreak\Rec)$, where commitments to 2 bits have size $\ell \coloneqq \ell(\secp)$.
\item \nblue{A universal hash function $h : \zo^{p(\secp)} \times \zo^{\leq B\secp} \rightarrow \zo^\secp$}.

\end{itemize}
\nblue{{\bf Sender input}: $m_0,m_1 \in \zo^\secp$, {\bf Receiver input}: $b \in \{0,1\}$}.
\begin{enumerate}
\item {\bf Receiver message.} $\sR$ performs the following steps.
\vspace{-2mm}
    \begin{enumerate}
        \item Choose $x \leftarrow \{0, 1\}^n$, $\theta \leftarrow \{0, 1\}^n$ and prepare the states $\{\ket{x_i}_{\theta_i}\}_{i \in [n]}$. 
        \item Compute $\left(\state,\{c_i\}_{i \in [n]}\right) \gets \Com\left(\{(x_i, \theta_i)\}_{i \in [n]}\right)$.
        \item \label{step:qot4r-recv-first-msg} 
        Send $\{\ket{x_i}_{\theta_i},c_i\}_{i \in [n]}$ to $\sS$.
    \end{enumerate}

\item {\bf Sender message.} \label{itm:4rqot-recv-second-msg} $\sS$ performs the following steps.
\vspace{-2mm}
    \begin{enumerate}
        \item Choose $\widehat{\theta} \leftarrow \{0, 1\}^n$. For all $i \in [n]$, measure $\ket{x_i}_{\theta_i}$ in basis $\widehat{\theta_i}$ to get outcome $\widehat{x_i}$. 
        \item Sample a random subset $T$ of $[n]$ of size $k$. Send $T, \{\widehat{\theta}_i\}_{i \in \overline{T}}$ to $\sR$, where $\overline{T} \coloneqq [n]\setminus T$.
    \end{enumerate}
    
\item {\bf Receiver message.} $\sR$ performs the following steps.
\vspace{-2mm}
    \begin{enumerate}
        \item Divide $\overline{T}$ into 2 disjoint subsets $S_0, S_1$ as follows. Set $I_b = S_0$ and $I_{1-b} = S_1$, where ${S_0 = \{ i \ | \ i \in \overline{T} \wedge \theta_i = \widehat{\theta}_i \}}, {S_1 = \{ i \ | \ i \in \overline{T} \wedge \theta_i \ne \widehat{\theta}_i \}}$. 
        
        \item Compute $\{(x_i, \theta_i), u_i\}_{i \in T} \gets \Open(\state, T)$.

        \item \sout{\nblue{Compute $X$ as the concatenation of $\{x_i\}_{i \in I_b}$, set $r_b := E(X) \oplus m$, and}} \sout{\nblue{sample $r_{1-b} \gets \zo^z$, where $z$ is the output length of extractor.}}
        
        \item \label{step:qot4r-recv-openings} Send $I_0,I_1,\{(x_i, \theta_i), u_i\}_{i \in T}$, \sout{\nblue{$(r_0, r_1)$}} to $\sS$.
    \end{enumerate}


\item \nblue{{\bf Sender message.} $\sS$ and $\sR$ do the following:
    \begin{itemize}
        \item $\sS$ does the following:
            \begin{itemize}
                \item Abort if $\Rec(\{c_i\}_{i \in T}, \{(x_i, \theta_i), u_i\}_{i \in T}) = \bot$ or if $\exists i \in T$ s.t.\ $\theta_i = \widehat{\theta}_i$ but $x_i \neq \widehat{x}_i$.
                \item Compute $\widehat{X}_0, \widehat{X}_1$ as the concatenation of $\{\widehat{x_i}\}_{i \in I_0},\{\widehat{x_i}\}_{i \in I_1}$ respectively.
                \item Sample $s \leftarrow \zo^{p(\secp)}$, send $(s,ct_0 = m_0 \oplus h(s,\widehat{X}_0),ct_1 = m_1 \oplus h(s,\widehat{X}_1))$ to $\sR$.
            \end{itemize}
        \item $\sR$ computes $X$ as the concatenation of $\{x_i\}_{i \in I_b}$ and outputs $ct_b \oplus h(s,X)$.
    \end{itemize}
}
\end{enumerate}
}

\subsection{Three-round random-input OT}
In this section, we prove \cref{thm:3-round-ROT}.
\paragraph{Sender Security}



Let $\SimExt = (\SimExt.\RO,\SimExt.\Ext)$ be the simulator for the extractable commitment scheme $(\Com^{\comro},\Open^{\comro},\Rec^{\comro})$ (according to \cref{def:fextractable}). Let $\epsilon = 0.053$ (for XOR extractor) or $0.010748$ (for ROM extractor) be a constant. We describe the simulator $\Sim[\rcv^*]$ against a malicious receiver $\rcv^*$. 

\noindent\underline{\textbf{$\Sim[\rcv^*]$}}:
\begin{itemize}
    \item Initialize $\rcv^*$ and answer its oracle queries to $\comro$ using $\SimExt.\RO$ and in case of ROM extractor queries to $\extro$ using the efficient on-the-fly random oracle simulator (\cref{impthm:extROSim}). Wait to receive $n$ qubits and commitments $\{c_i\}_{i \in [n]}$ from $\rcv^*$.
                
    \item $\{({x}^*_i,{\theta}^*_i)\}_{i \in [n]} \leftarrow \SimExt.\Ext(\{c_i\}_{i \in [n]})$.
    
    \item Choose $\widehat{\theta} \leftarrow \zo^{n}$ and measure all the received qubits in bases $\widehat{\theta}$ to get measurement outcomes $\widehat{x}$. Also, sample a random subset $T \leftarrow [n]$ s.t. $|T|=k$ and send $T,\{\widehat{\theta}_i\}_{i \in \overline{T}}$ to $\rcv^*$, where $\overline{T} \coloneqq [n]\setminus T$.
    
    \item Wait to receive sets $I_0,I_1$, where $I_0 \subseteq \overline{T}$, $I_1 = \overline{T} \setminus I_0$ and openings $\{(x_i,\theta_i),u_i\}_{i \in T}$.
    
    \item Check if $\Rec(\{c_i\}_{i \in T}, \{(x_i, \theta_i), u_i\}_{i \in T}) = \bot$ or if there exists $ i \in T$ s.t.\ $\widehat{\theta}_i = {\theta}^*_i$ but $\widehat{x}_i \neq {x}^*_i$. If any of the checks fail, send $\mathsf{abort}$ to the ideal functionality, output $\rcv^*$'s state and continue answering distinguisher's queries. 
    

    \item \label{itm:4rqot-simurecv-bset}
    Compute the set $Q = \{i \ | \  i \in I_0 \wedge \theta_i^* \neq \widehat{\theta}_i\}$. If $|Q| \geq \frac{(1-\epsilon)(n-k)}{4}$, then set $b=1$, else set $b=0$. Compute $m_{b} = \{\widehat{x}_i\}_{i \in I_{b}}$, send $(b,m_b)$ to the ideal functionality and output $\rcv^*$'s state.
    
    \item Continue answering distinguisher's queries to $\comro$ (and $\extro$) using $\SimExt.\RO$ (and the efficient on-the-fly random oracle simulator).
\end{itemize}

Consider a distinguisher $(\sD,\bsigma)$ such that $\rcv^*,\sD$ make a total of $q$ queries combined to $\comro$ (and $\extro$). Consider the following sequence of hybrids:

\begin{itemize}
    \item $\Hyb_0$: This is the real world interaction between $\rcv^*,\sS$. Using the notation of \cref{def:secure-realization}, this is a distribution over $\zo$ denoted by $\Pi[\rcv^*,\sD,\top]$.
    
    \item $\Hyb_1$: This is identical to the previous hybrid, except that queries to $\comro$ are answered using $\SimExt.\RO$, and $\{({x}^*_i,{\theta}^*_i)\}_{i \in [n]} \leftarrow \SimExt.\Ext(\{c_i\}_{i \in [n]})$ is run after $\rcv^*$ outputs its first message. After $\rcv^*$ sends its openings in the third round, the sender performs the following checks: check if $\Rec(\{c_i\}_{i \in T}, \allowbreak \{(x_i, \theta_i), u_i\}_{i \in T}) = \bot$ or if there exists $ i \in T$ s.t.\ $\theta^*_i = \widehat{\theta}_i$ but $x^*_i \neq \widehat{x}_i$ (note that it uses $x_i^*,\theta_i^*$ for its second check, rather than ${x}_i,{\theta}_i$ as in the honest sender strategy). It then continues with the rest of the protocol as in the honest sender strategy.
    
    \item $\Hyb_2$: This is the result of $\Sim[\rcv^*]$ interacting in $\widetilde{\Pi}_{\SROT[1]}[\Sim[\rcv^*],\sD,\top]$ (or $\widetilde{\Pi}_{\SROT[\secp]}[\Sim[\rcv^*],\sD,\top]$). 
\end{itemize}

\begin{claim}
    $|\Pr[\Hyb_0=1]-\Pr[\Hyb_1=1]| \leq \frac{148(q + 2n + 1)^3 + 1}{2^{2\secp}} + \frac{16qn}{2^{\secp}}$.
\end{claim}
\begin{proof}
    
    This follows by a direct reduction to extractability of the commitment scheme (\cref{def:fextractable}). Indeed, let $\Adv_\mathsf{Commit}$ be the machine that runs $\Hyb_0$ until $\rcv^*$ outputs its message, which includes $\{c_i\}_{i \in [n]}$. Let $\Adv_\mathsf{Open}$ be the machine that takes as input the rest of the state of $\Hyb_0$, and runs it till the third round to get set $T$ and openings $\{(x_i,\theta_i),u_i\}_{i \in T}$, and outputs $T$ and these openings. Let $\sD$ be the machine that runs the rest of $\Hyb_0$ and outputs a bit.

    Then, plugging in $\secp_\com = 2\secp$, \cref{def:fextractable} when applied to $(\Adv_{\mathsf{Commit}},\Adv_{\mathsf{Open}},\sD)$ implies that the hybrids cannot be distinguished except with probability 
    \[
        \frac{148(q + 2n + 1)^3 + 1}{2^{2\secp}} + \frac{16qn}{2^{\secp}},
    \] since we are committing to a total of $2n$ bits. 
\end{proof}

\begin{claim}
    When instantiated with the XOR extractor and constants $A = 1100, B = 500$, we have,
        \[
            |\Pr[\Hyb_1=1]-\Pr[\Hyb_2=1]| \leq \frac{\sqrt{5}+1}{2^{\secp}}.
        \]
    And when instantiated with the ROM extractor and constants $A = 11\ 000, B = 12\ 000$, we have,
        \[
            |\Pr[\Hyb_1=1]-\Pr[\Hyb_2=1]| \leq \frac{\sqrt{5}}{2^{\secp}} + \frac{4q}{2^{18\secp}}.
        \]
\end{claim}
\begin{proof}
    In $\Hyb_1$, while both $R_0,R_1$ were set according to the honest sender strategy, in $\Hyb_2$, for $b$ as defined by the simulator, $R_{1-b}$ (output by the honest sender) is set as a uniformly random string. In the following, we show that in either hybrid $R_{1-b}$ is statistically close to a uniformly random string given $\rcv^*$'s view which would imply this claim. We setup some notation before proceeding:
    \begin{itemize}
        \item Let $\cX=\{\cX_i\}_{i \in [n]}$ denote the $n$ registers, each holding a single qubit, sent by $\rcv^*$ in its first round.
        \item For a vector $\bx = (x_1,\dots x_n)$, and a set $S\subseteq [n]$, let $\bx_{S}$ denotes the values of $\bx$ indexed by $S$.
        \item For vectors $\bx,\btheta \in \zo^n$, let $\ket{\bx_{\btheta}}$ denote the state on $n$ single-qubit registers, where register $i$ contains the state $\ket{\bx_i}$ prepared in the $\btheta_i$ basis. 
        \item For a set $T \subseteq [n]$, let $\overline{T} \coloneqq [n]\setminus T$.
        \item Using notation as defined in \cref{sec:prelims}, for a subset $S \subseteq [n]$ and two vectors $\bx,\by \in \zo^n$, $\Delta(\bx_S,\by_S)$ denotes the fraction of values $\bx_i, i \in S$ s.t.\ $\bx_i \neq \by_i$.
    \end{itemize}
    Consider the following quantum sampling game (defined as in \cref{subsec:quantum-sampling}):
    \begin{itemize}
        \item Fix some state on register $\cX$ and some strings $\bx^*,\btheta^* \in \zo^n$. 
        \item Sample $T \subseteq [n]$ as a uniform random subset of size $k$. 
        \item Sample $\widehat{\btheta} \leftarrow \zo^n$, and let $S = \{i \, | \, i \in T \wedge \widehat{\btheta}_i = \btheta_i^*\}$. For each $i \in S$, measure register $\cX_i$ in basis $\btheta_i^*$ to get outcome $x_i$. 
        \item Let $\bx_S$ be the concatenation of $\{x_i\}_{i \in S}$. Output $\Delta(\bx_S,\bx_S^*)$.
    \end{itemize}
    
    This quantum sampling game corresponds to the execution in either hybrid, where register $\cX$ is the register sent by $\rcv^*$ in its first message, $(\bx^*,\btheta^*)$ represent the values extracted by running $\SimExt.\Ext$, and $\widehat{\btheta}$ respresents the bases sampled by the sender. 
    By \cref{def:quantum-error-probability}, the quantum error probability $\epsilon^\delta_{\mathsf{quantum}}$ of the above game corresponds to the trace distance between the initial state on the register $\cX$ and an ``ideal'' state, where it holds with certainty that register $\cX$ is in a superposition of states $\ket{\bx_{\btheta^*}}$ for $\bx$ s.t.\ $|\Delta(\bx_{\overline{T}},\bx^*_{\overline{T}}) - \Delta(\bx_{S},\bx^*_S)| < \delta$. 
    In the following we find a bound $\epsilon_{\mathsf{quantum}}^\delta$ and show that given the state on register $\cX$ is in the ideal state described above, the two hybrids are statistically indistinguishable.
    
    \begin{subclaim}
        \label{subclaim:randombasis-quantumerror}
        $\epsilon^\delta_{\mathsf{quantum}} \leq \frac{\sqrt{5}}{2^{\secp}}$ when instantiated with the XOR extractor (with $\delta = 0.1183$) or with the ROM extractor (with $\delta = 0.0267$).
    \end{subclaim}
    \begin{proof}
        Using \cref{impthm:error-probability}, $\epsilon^\delta_{\mathsf{quantum}}$ can be bound by the square root of the classical error probability, $\epsilon_{\mathsf{classical}}^\delta$, of the corresponding classical sampling game, described as follows:
        \begin{itemize}
            \item Given a string $\bq \in \zo^n$, sample $T \subseteq [n]$ as a uniform random subset of size $k$. 
            \item Sample a subset $S \subseteq T$ as follows: sample bits $\bb \leftarrow \zo^n$. Let set $S = \{i \, | \, i \in T \wedge \bb_i = 1\}$. 
            \item Output $\omega(\bq_S)$.
        \end{itemize}
        Since setting $S$ as above is equivalent to choosing a random subset of $T$ (chosen uniformly among all possible subsets of $T$), using the analysis in \cref{appsubsec:CK88sampling}, we get, for $0 < \beta < 1$ and $0 < \eta < \delta$, 
        \[
            \epsilon_{\mathsf{classical}}^\delta \leq 2\exp\left(-2\left(1-\frac{k}{n}\right)^2\eta^2 k\right)+2\exp\left(-(\delta-\eta)^2(1-\beta){k}\right)+ \exp\left(-\frac{\beta^2 k}{2}\right).
        \]
        For the case of XOR extractor, for $\delta = 0.1183, \beta = 0.051, \eta = 0.081$, we have each of the expressions inside the exp terms above bounded by $\frac{1}{2^{2\secp}}$, giving us $\epsilon_{\mathsf{classical}}^\delta \leq \frac{5}{2^{2\secp}}$, which means $\epsilon_{\mathsf{quantum}}^\delta \leq \frac{\sqrt{5}}{2^{\secp}}$. 
        
        And for the case of ROM extractor, for $\delta = 0.0267, \beta = 0.01588, \eta = 0.01538$, we achieve the same bounds and get $\epsilon_{\mathsf{quantum}}^\delta \leq \frac{\sqrt{5}}{2^{\secp}}$. 
    \end{proof}
    
    \begin{subclaim}
        \label{subclaim:randombasis-minentropy}
        Given the state on register $\cX$ is in a superposition of states $\ket{\bx_{\btheta^*}}$ s.t.\ $|\Delta(\bx_{\overline{T}},\bx^*_{\overline{T}}) - \Delta(\bx_{S},\bx^*_S)| < \delta$, where $S = \{i \, | \, i \in T \ \wedge \ \widehat{\theta}_i = \theta_i^*\}$, the following holds: 
        \begin{itemize}
            \item When instantiated with the XOR extractor and $\delta = 0.1183, A = 1100, B = 500$, 
            \[    
                |\Pr[\Hyb_1=1]-\Pr[\Hyb_2=1]|\leq \frac{1}{2^\secp}
            \]
            \item When instantiated with ROM extractor and $\delta = 0.0267, A = 11\ 000, B = 12\ 000$,
            \[
                |\Pr[\Hyb_1=1]-\Pr[\Hyb_2=1]|\leq \frac{4q}{2^{18\secp}}
            \]
        \end{itemize}
            
    \end{subclaim}
    \begin{proof}
        Note that the checks performed by the sender once $\rcv^*$ sends the openings in the third round correspond to checking if $\Delta(\bx_{S},\bx^*_S)=0$. If this check fails, the sender aborts and the hybrids are perfectly indistinguishable. So it suffices to analyze the states on $\cX$ that are in a superposition of states $\ket{\bx_{\btheta^*}}$ s.t.\ $\Delta(\bx_{S},\bx^*_S)=0$ and  $\Delta(\bx_{\overline{T}},\bx^*_{\overline{T}}) < \delta$. 
        
        Recall that in either hybrid, for $i \in \overline{T}$, the sender chooses bit $\widehat{\theta}_i \leftarrow \zo$ and measures register $\cX_i$ in basis $\widehat{\theta}_i$. Using Hoeffding's inequality (stated in \cref{app:classampling}), the number of positions $i \in \overline{T}$ s.t.\ $\widehat{\theta}_i \neq \theta_i^*$ is at least $\frac{(1-\epsilon)(n-k)}{2}$ except with probability $\exp\left(-\frac{\epsilon^2(n-k)}{2}\right)$. Hence, given any partition $(I_0,I_1)$ of $\overline{T}$ that $\rcv^*$ provides in the third round, it holds that there exists a bit $b$ and partition $I_{1-b}$ s.t.\ there are at least $\frac{(1-\epsilon)(n-k)}{4}$ positions $i$ with $\widehat{\theta}_i \neq \theta_i^*$ except with probability $\exp\left(-\frac{\epsilon^2(n-k)}{2}\right)$. Call this subset of positions in $I_{1-b}$ as $M$. Also, note that in $\Hyb_3$, the bit $b$ used by $\Sim[\rcv^*]$ is the same bit used above.
        \begin{itemize}
            \item {\it XOR extractor}: For the case of XOR extractor, for $\epsilon=0.053, n-k=B\secp, B = 500$, this probability above $\exp\left(-\frac{\epsilon^2(n-k)}{2}\right)<\frac{1}{2^\secp}$. 

            Combining the two parts above, we have that register $\cX_M = \{\cX_i\}_{i \in M}$ is in a superposition of states $\ket{\bx_{\btheta^*_M}}$ s.t.\ $\Delta(\bx,\bx^*_M) < \frac{\delta(n-k)}{(1-\epsilon)(n-k)/4}= \frac{4\delta}{1-\epsilon} \approx 0.49968 < \frac{1}{2}$ (for $\delta = 0.1183, \epsilon = 0.053$). Using \cref{lemma:XOR-extractor}, it then follows that $m_{1-b}$ is uniformly random string, hence proving the given claim.
            
            \item {\it ROM extractor}: For the case of ROM extractor, for $\epsilon=0.01013, n-k=B\secp, B = 13\, 500$, this probability above $\exp\left(-\frac{\epsilon^2(n-k)}{2}\right)<\frac{1}{2^\secp}$. 

            In a similar way as above, we have that register $\cX_M = \{\cX_i\}_{i \in M}$ is in a superposition of states $\ket{\bx_{\btheta^*_M}}$ s.t.\ $\Delta(\bx,\bx^*_M) < \frac{\delta(n-k)}{(1-\epsilon)(n-k)/4}= \frac{4\delta}{1-\epsilon} < 0.10796$ (for $\delta = 0.0267, \epsilon = 0.010748$).
            We now apply \cref{thm:ROM-extractor} with random oracle input size $n-k$, register $\cX$ of size $|M|$, and $|L|\leq 2^{h_b(0.10796)(1-\epsilon)(n-k)/4}$. Note that, when applying this theorem, we are fixing the outcome of the $n-k- 
            |M|$ bits of the random oracle input that are measured in basis ${\theta}^*_i$, and setting $\cX$ to contain the $|M|$ registers are measured in basis $\widehat{\theta}_i = \theta^*_i\oplus 1$. This gives a bound of 
            \begin{align*}
                 \frac{4\cdot q\cdot 2^{h_b(0.10796)(1-\epsilon)(n-k)/4}}{2^{(1-\epsilon)(n-k)/8}} &= \frac{4q}{2^{\left(\frac{1}{2} - h_b(0.10796)\right)(1-\epsilon)(n-k)/4}}\\
                 = \frac{4q}{2^{\left(\frac{1}{2} - h_b(0.10796)\right)(1-\epsilon)B\secp/4}} &\leq \frac{4q}{2^{18\secp}}
            \end{align*}
            for $B = 12\ 000, \epsilon = 0.010748$.
        \end{itemize}
        
    \end{proof}

\end{proof}



\paragraph{Receiver Security.}
Let $\SimEqu = (\SimEqu.\RO,\SimEqu.\Com,\SimEqu.\Open)$ be the simulator for the equivocal commitment scheme $(\Com^{\comro},\Open^{\comro},\Rec^{\comro})$ (according to \cref{def:fextractable}). We describe the simulator $\Sim[\sS^*]$ against a malicious receiver $\sS^*$. $\simu$ will answer random oracle queries to $\comro$ using $\SimEqu.\RO$.
Additionally, if randomness extractor $E$ in the protocol is $\extro$, then its simulation is accomplished via queries to an on-the-fly random oracle simulator $\simro.\RO$ as mentioned in Imported Theorem \ref{impthm:extROSim}.

\paragraph{The Simulator.}
\begin{enumerate}
    \item Prepare $n$ EPR pairs on registers $\{\cS_i, \cR_i\}_{i \in [n]}$
    
    \item Compute the commitments strings $\{ c_i \}_{i \in [n]}$ by calling $\SimEqu.\Com$ for the underlying commitment scheme.
    
    \item Send $\{\cS_i\}_{i \in [n]}$ and $\{ c_i \}_{i \in [n]}$ to $S^*$.
    
    \item Receive $T, \{ \widehat{\theta_i} \}_{i \in \overline{T}}$ from $S^*$.
    
    \item For all $i \in T$, sample $\theta_i \leftarrow \{+, \times \}$ and measure $\cR_i$ in the basis $\theta_i$ to obtain outcome $x_i$. For all $i \in \overline{T}$, measure $\cR_i$  in the basis $\widehat{\theta_i}$ to obtain outcome $\widetilde{x_i}$.
    
    \item Call $\SimEqu.\Open(\{ x_i, \theta_i \}_{i \in [T]})$ of the underlying commitment scheme to obtain openings $\{ (x_i, \theta_i), u_i \}_{i \in [T]}$.

    \item Generate a partition $I_0, I_1$ of $\overline{T}$ as follows: for every $i \in \overline{T}$, flip a random bit $b_i$ and place $i \in I_{b_i}$. 
    
    \item Receives $m_0, m_1$ from $\cF_{\SROT}$ functionality.
    
    \item Set $r_0 = E(\{ \widetilde{x_i} \}_{i \in I_0}) \oplus m_0$,  $r_1 = E(\{ \widetilde{x_i} \}_{i \in I_1}) \oplus m_1$.
    
    \item Send $I_0, I_1, \{ (x_i, \theta_i), u_i \}_{i \in T}, (r_0, r_1)$ to $S^*$.

\end{enumerate}

\paragraph{Analysis.}
Fix any adversary $\{ \sS^*_\secp,\sD_\secp,(b_\secp, m_\secp)\}_{\secp \in \bbN}$, where $\sS^*_\secp$ is a QIOM that corrupts the sender, $\sD_\secp$ is a QOM, and $(b_\secp, m_\secp)$ is the input of the honest receiver. 
For any receiver input $b_\secp \in \{0,1\}, m_\secp \in \zo^v$, where $v$ is the output lenght of extractor, consider the random variables $\Pi[ \sS^*_\secp,\sD_\secp,(b_\secp, m_\secp)]$ and 
$\widetilde{\Pi}_{\cF^{\mathsf{OT}}}[ \Sim_\secp,\sD_\secp,(b_\secp, m_\secp)]$
according to Definition \ref{functionalities-and-protocols-in-qrom} for the protocol in Figure \ref{fig:qot-4r-const}.
Let $q(\cdot)$ denote an upper bound on the number of queries of $\sS^*_\secp$ and $\sD_\secp$.
We will show that :
\[\bigg|\Pr[\Pi[ \sS^*_\secp,\sD_\secp,(b_\secp, m_\secp)] = 1] - \Pr[\widetilde{\Pi}_\cF[ \Sim_\secp,\sD_\secp,(b_\secp, m_\secp)] = 1]\bigg| = \mu(\secp,q(\secp))\]
where the term on the right corresponds to the security error in the equivocal commitment.

This is done via a sequence of hybrids, as follows \footnote{If randomness extractor $E$ in the protocol is $\extro$, then there will be an additonal hybrid between $\Hyb_0$ and $\Hyb_1$ where we switch from using $\extro$ to simulating it using an efficient on-the-fly random oracle simulator $\simro.\RO$ as mentioned in Imported Theorem \ref{impthm:extROSim}. This hybrid's (perfect) indistinguishability will follow directly from the indistinguishable simulation property as mentioned in Imported Theorem \ref{impthm:extROSim}.}:

\begin{itemize}{
    \item $\Hyb_0:$ The output of this hybrid is the {\em real} distribution $\Pi[ \sS^*_\secp,\sD_\secp,(b_\secp, m_\secp)]$.



\item $\Hyb_1:$ 
This is the same as the previous hybrid except that instead of running the $\Com$ and  $\Open$ algorithm, as in Figure \ref{fig:qot-4r-const}, the challenger now answers random oracle queries to $\comro$ using $\SimEqu.\RO$, the random oracle simulator for the commitment scheme $(\Com^\comro,\allowbreak\Open^\comro,\allowbreak\Rec^\comro)$. Additionally, it performs the following modified steps on behalf of $\sR$:
\begin{itemize}
    \item Prepare the commitments by calling $\SimEqu.\Com$ for the underlying commitment protocol.
    
    \item Prepare the opening by calling $\SimEqu.\Open(\{x_i, \theta_i\}_{i \in [n]})$, where $\{x_i, \theta_i\}_{i \in [n]}$ are as defined in the previous hybrid.
\end{itemize}

    
    
    \item $\Hyb_2:$ 
    This is the same as the previous hybrid except the following change: in protocol round 1, the challenger calls 
    Algorithm EPR-to-BB84$(i)$ for every $i \in [n]$ to obtain $\ket{x_i}_{\theta_i}$. 
    
    \textbf{Algorithm EPR-to-BB84 $(i)$:}
    \begin{enumerate}
        \item Sample EPR pair on registers $\cS_i, \cR_i$.
        \item Randomly sample a basis $\theta_i \leftarrow \{+, \times\}$
        \item Measure $\cR_i$ in the basis $\theta_i$ and let the outcome be $x_i$
        \item Use $\cS_i$ as a BB84 state $\ket{x_i}_{\theta_i}$
    \end{enumerate}
    
    
    
        

    \item $\Hyb_3$:
    This is the same as the previous hybrid, except that in protocol round 1, the challenger sends halves of $n$ EPR pairs $\{\cS_i\}_{i \in [n]}$, prepared by executing Step 1 of the algorithm EPR-to-BB84, to $\sS^*$ while retaining $\{ \cR \}_{i \in [n]}$ with itself.
    After round 2, the challenger runs Steps 2 and 3 of the Algorithm EPR-to-BB84 for every $i \in [n]$ to obtain $\{x_i, \theta_i \}_{i \in [n]}$. The resulting values $\{x_i, \theta_i \}_{i \in T}$ are used as inputs to $\SimEqu.\Open$ to prepare openings in round 3. Step 4 of the Algorithm EPR-to-BB84$(i)$ is not relevant in this hybrid.

    \item $\Hyb_4$:
    This is the same as the previous hybrid, except the following changes. After round 2, the challenger runs the Steps 2-3 of algorithm EPR-to-BB84 for every $i \in [T]$, leaving $\{\cR_i\}_{i \in \overline{T}}$ unmeasured. 
    It generates a partition $I_0, I_1$ of $\overline{T}$ as follows: for every $i \in \overline{T}$, flip a random bit $b_i$ and place $i \in I_{b_i}$.
    
    
    For all $i \in \overline{T}$, the challenger measures $\{\cR_i\}_{i \in \overline{T}}$  in the basis $\{ \widehat{\theta_i} \}_{i \in \overline{T}}$ where $\{ \widehat{\theta_i} \}_{i \in \overline{T}}$ was obtained from $\sS^*$ in round $2$. Denote measurement outcomes by $\{ \widetilde{x_i} \}_{i \in \overline{T}}$. 
    Using the resulting outcomes, the challenger sets $r_b := E(\{\widetilde{x}_i\}_{i \in I_b}) \oplus m$.
    
    
    

    \item $\Hyb_5:$ This is the same as the previous hybrid, except that in Round 3, the challenger sets $r_0 = E(\{ \widetilde{x_i} \}_{i \in I_0}) \oplus m_0$,  $r_1 = E(\{ \widetilde{x_i} \}_{i \in I_1}) \oplus m_1$  where $m_0, m_1$ are received from $\cF_{\SROT}$

The output of this experiment is identical to the {\em ideal} distribution $\widetilde{\Pi}_{\cF_{\SROT}}[ \allowbreak \Sim_\secp,\sD_\secp,(b_\secp, m_\secp)]$. 

}
\end{itemize}

\noindent We show 
that $|\Pr[\Hyb_5=1]-\Pr[\Hyb_0=1]| \leq \mu(\secp,q(\secp))$, where $(\Com^\comro, \allowbreak \Open^\comro, \allowbreak \Rec^\comro)$ is a $\mu(\secp,q(\secp))$-equivocal bit commitment scheme, where $(\Com^\comro,\allowbreak\Open^\comro,\allowbreak\Rec^\comro)$ is a $\mu(\secp,q(\secp))$-equivocal bit commitment scheme, where $\mu(\secp,q, n_\com) = \frac{2qn_{\com}^{1/2}}{2^{\secp_\com/2}}$ for the specific commitment scheme that we construct in Section \ref{sec:eecom-construction}, where $n_{\com}$ is the number of bit commitments and $\secp_\com$ is the security parameter for the commitment scheme. Later, we will set $n_\com = c_1 \secp$ and $\secp_\com = c_2 \secp$ for some fixed constants $c_1, c_2$. Thus $\mu$ will indeed be a function of $\secp$ and $q$. We now procced with the proof by arguing the computational indistinguishability of each pair of consecutive hybrids in the above sequence.

\begin{claim}
    $|\Pr[\Hyb_0=1]-\Pr[\Hyb_1=1]|\leq \mu(\lambda, q(\lambda))$. 
\end{claim}
\begin{proof}
    Suppose there exists an adversary $\Adv_\secp$ corrupting $\sS$, a distinguisher $\sD_\secp$, quantuam states $\rho_\secp, \sigma_\secp$ and a bit $b$ such that,
    
    \[
     \bigg|\Pr[\Hyb_0=1]-\Pr[\Hyb_1=1]\bigg| > \mu(\lambda, q(\lambda))
    \]
    
    We will construct a reduction $\{\Adv^*_{\secp} = (\Adv_{\RCommit,\secp},\allowbreak\Adv_{\ROpen,\secp},\allowbreak\sD^*_{\secp})\}_{\secp \in \bbN}$ that makes at most $q(\secp)$ queries to the random oracle, and contradicts the $\mu$-equivocality of the underlying commitment scheme $(\Com^\comro,\allowbreak\Open^\comro,\allowbreak\Rec^\comro)$ as defined in Definition \ref{def:fequivocal}. In the following reduction, all random oracle queries to $\comro$ will be answered by the equivocal commitment challenger.\\
    
    $\Adv_{\RCommit,\secp}(\rho_\secp)$:
    \begin{itemize}
        \item Initalize the OT protocol with between honest receiver $\sR$ and $\Adv(\rho_\secp)$ corrupting $\sS$. 
        
        \item After $\sR$ samples $\{(x_i,  \theta_i)\}_{i \in [n]}$, output the intermediate state $\rho^*_{\secp, 1}$ representing the joint state of $\sS$ and $\sR$ along with $\{(x_i,  \theta_i)\}_{i \in [n]}$.
    \end{itemize}        

    The commitment challenger obtains $\{(x_i,  \theta_i)\}_{i \in [n]}$ and returns a set of commitments $\{ \com_i \}_{i \in [n]}$.\\
    
    $\Adv_{\ROpen,\secp}(\rho^*_{\secp, 1}, \{ \com_i \}_{i \in [n]})$: Use $\rho^*_{\secp, 1}$ to initialize the joint state of $\sS$ and $\sR$, and $\{ \com_i \}_{i \in [n]}$ as commitments of $\sR$ in the protocol. Output the new joint state $\rho^*_{\secp, 2}$ after $\sR$ has computed $I_0$ and $I_1$.\\
    
    The challenger returns $\{u_i\}_{i \in [n]}$ which is then fed to the following distinguisher (along with the information $\{ \com_i, (x_i, \theta_i) \}_{i \in [n]}$ from the aforementioned execution).\\
    
     $\sD^*_\secp(\rho^*_{\secp, 2}, \{ \com_i, (x_i, \theta_i), u_i \}_{i \in [n]}):$
    
    \begin{itemize}
        \item Use $\rho^*_{\secp, 2}$ to initialize the joint state of $\sS$ and $\sR$. Run it until completion using $\{ (x_i, u_i) \}_{i \in T}$ as openings of $\sR$.
        
        \item Let $\tau_\secp^*$ be the final state of $\Adv$ and $y^*$ be the output of $\sR$. Run $\sD_\secp(\sigma_\secp, \tau^*_\secp, y^*)$ and output the bit $b$ returned by it.
    \end{itemize}
    
    By construction, when the challenger executes $(\Com^\comro,\allowbreak\Open^\comro)$, the reduction will generate a distribution identical to $\Hyb_0$. Similarly, when the challenger executes $(\SimEqu.\Com, \SimEqu.\Open)$ algorithms, the reduction will generate a distribution identical to $\Hyb_1$. Therefore, the reduction  directly contradicts the $\mu$-equivocality of the underlying commitment scheme $(\Com^\comro,\allowbreak\Open^\comro,\allowbreak\Rec^\comro)$ according to Definition \ref{def:fequivocal}, as desired.
\end{proof}

\begin{claim}
    $\Pr[\Hyb_1=1] = \Pr[\Hyb_2=1]$
\end{claim}
\begin{proof}
The only difference between the two hybrids is a syntactic change in the way BB84 states are sampled in round $1$. The distribution $(x_i, \theta_i, \ket{x_i}_{\theta_i})_{i \in [16\secp]}$ resulting from these syntactically different sampling strategies is identical in both hybrids.
\end{proof}

\begin{claim}
    $\Pr[\Hyb_2=1] = \Pr[\Hyb_3=1]$
\end{claim}
\begin{proof}
$\Hyb_3$ constitutes a purification of the receiver's strategy in round $1$ and since actions on disjoint subsystems commute, this does not affect the joint distribution of the sender's view and receiver output. 
\end{proof}


\begin{claim} $\Pr[\Hyb_3=1] = \Pr[\Hyb_4=1]$
\end{claim}
\begin{proof}

$\Hyb_4$ constitutes a purification of the receiver's strategy in round $3$ and since actions on disjoint subsystems commute, this does not affect the joint distribution of the sender's view and receiver output. 
\end{proof}

\begin{claim}
    $\Pr[\Hyb_4=1] = \Pr[\Hyb_5=1]$
\end{claim}

\begin{proof}
    Assuming correctness of $\cF_{\SROT}$, the two hybrids are identical.Suppose ideal world receiver's input is $(b = 0, m)$. In this case, $\cF_{\SROT}$ sends $m_0 = m, m_1$ to the challenger where $m_1 \gets \zo^v$. In $\Hyb_5$, this would lead to $r_0 = E(\{ \widetilde{x_i} \}_{i \in I_0}) \oplus m$ (which is same as $\Hyb_4$) and $r_1 = E(\{ \widetilde{x_i} \}_{i \in I_1}) \oplus m_1$ (which is uniformly random as in $\Hyb_4$). Moreover, the output on sender side in $\Hyb_5$ is $(E(\{ \widetilde{x_i} \}_{i \in I_0}) \oplus r_0, E(\{ \widetilde{x_i} \}_{i \in I_1}) \oplus r_1)$ = $(m, m_1)$ as desired.
    The case when ideal world receiver bit $b$ is $1$ can be proved in a similar way. Therefore for any fixing of the adversary's state and receiver's input, the two hybrids result in identical distributions.
\end{proof}

Combining all the claims, we get that $|\Pr[\Hyb_0 = 1] - \Pr[\Hyb_5 = 1]| \le  \mu(\lambda, q(\lambda))$. In Theorem \ref{thm:equcom}, we derived  $\mu(\secp,q, n_\com) = \frac{2qn_{\com}^{1/2}}{2^{\secp_\com/2}}$. We will now state the parameters for $n_\com$ and $\lambda_\com$.

\begin{itemize}
    \item {\it XOR extractor}: Plugging $\lambda_\com = 2 \secp$, $n_\com = 2n$ (as we are committing to 2 bits at a time) where $n = 1600\secp$ (this setting of $n$ is the same as that needed in the sender security part of the proof), we get $\frac{114q\sqrt{\secp}}{2^{\secp}}$ security against a malicious sender.
    
    \item {\it ROM extractor}: Plugging $\lambda_\com = 2 \secp$, $n_\com = 2n$ (as we are committing to 2 bits at a time) where $n = 23\ 000\secp$ (this setting of $n$ is the same as that needed in the sender security part of the proof), we get $\frac{430q\sqrt{\secp}}{2^{\secp}}$ security against a malicious sender.
\end{itemize}

\subsection{Four-round chosen-input OT}
In this section, we prove \cref{thm:4-round-chosenOT}.
\paragraph{Sender Security}
The proof of this follows along a similar line as the proof of security against a malicious receiver for the 3-round protocol described before. We only describe the changes to the corresponding proof from before over here. 

The only change to the simulator (compared to $\Sim[\rcv^*]$ for the 3-round protocol described earlier) is that after computing $b$ at the end of third round, it sends $b$ to $\cF_{\OT[\secp]}$ to receive back $m_b$, sets $m_{1-b} \coloneqq 0^\secp$, and thereafter completes the protocol as in the honest sender strategy. Once it outputs $\rcv^*$'s state, it continues answering distinguisher's queries using $\SimExt.\RO$.

The hybrids ($\Hyb_0,\Hyb_1,\Hyb_2$) also remain same as in the proof before, and the indistinguishability between $\Hyb_0,\Hyb_1$ proceeds as before. The indistiguishability between $\Hyb_1,\Hyb_2$ follows using a slightly modified analysis of \cref{subclaim:randombasis-quantumerror} and \cref{subclaim:randombasis-minentropy}. Specifically, for the proof of \cref{subclaim:randombasis-quantumerror}, we use $\delta = 0.04, \beta = 0.023, \eta = 0.0236, A = 5300, B = 5000$, and obtain the same result of $\epsilon_{\mathsf{quantum}}^\delta \leq \frac{\sqrt{5}}{2^\secp}$. 

For the proof of \cref{subclaim:randombasis-quantumerror}, we use a different analysis as follows: set $\epsilon = 0.017$. By assumption of the subclaim and using a similar analysis as the proof of \cref{subclaim:randombasis-quantumerror}, the state on $\cX$ is in a superposition of states $\ket{\bx_{\btheta^*}}$ s.t.\ $\Delta(\bx_{S},\bx^*_S)=0$ and  $\Delta(\bx_{\overline{T}},\bx^*_{\overline{T}}) < \delta$. Using Hoeffding's inequality, the number of positions $i \in \overline{T}$ s.t.\ $\widehat{\theta}_i \neq \theta_i^*$ is at least $\frac{(1-\epsilon)(n-k)}{2}$ except with probability $\exp\left(-\frac{\epsilon^2(n-k)}{2}\right)$. For $\epsilon=0.053, n-k=B\secp, B = 5000$, this probability is $<\frac{1}{2^\secp}$. Next, as before,  given any partition $(I_0,I_1)$ of $\overline{T}$ that $\rcv^*$ sends in the third round, it holds that there exists a bit $b$ and partition $I_{1-b}$ s.t.\ there are at least $\frac{(1-\epsilon)(n-k)}{4}$ positions $i$ with $\widehat{\theta}_i \neq \theta_i^*$. 

Hence, $\cX_{I_{1-b}}$ is in a superposition of states $\ket{\bx_{\left(\btheta^*_{I_{1-b}}\right)}}$ s.t.\ $\Delta(\bx,\bx^*_{I_{1-b}}) < \delta$ and at least $\frac{(1-\epsilon)(n-k)}{4}$ positions of it are measured in basis $\widehat{\theta}_i \neq \theta_i^*$. Let $\widehat{X}_{1-b}$ be the string obtained by concatenating the measurement outcomes of $I_{1-b}$. Also, let $\cC$ denote the register for the complete system (including the private state of $\rcv^*$), but excluding register $\cX$. Using \cref{impthm:small-superposition}, we get, 
\begin{align*}
    \mathbf{H}_\infty(\widehat{X}_{1-b} \, | \, \cX_{I_b}, \cC) &\geq \frac{(1-\epsilon)(n-k)}{4} - h_b(\delta)|I_{1-b}|\\
    &\geq \frac{(1-\epsilon)(n-k)}{4} - h_b(\delta)(n-k)
\end{align*}
For $\epsilon = 0.017, \delta = 0.04, n-k = B\secp, B = 5000$, we get, $\mathbf{H}_\infty(\widehat{X}_{1-b} \, | \, \cX_{I_b}, \cC)\geq 17\secp$, and also, that $\mathbf{H}_\infty(\widehat{X}_{1-b} \, | \, \widehat{X}_{I_b}, \cC)\geq 17\secp$. Using \cref{impthm:privacy-amplification}, we then get that $(s,h(s,\widehat{X}_{I_{1-b}}))$ is $\frac{1}{2^{9\secp}}$ statistically close to uniformly random string.
Hence, $|\Pr[\Hyb_1=1]-\Pr[\Hyb_2=1]|\leq \frac{\sqrt{5}}{2^{\secp}}+ \frac{1}{2^{9\secp}}$.

\paragraph{Receiver Security} The proof of this is similar to the proof of receiver security for the 3 round random basis protocol described before. We only describe the changes here. The only change to the simulator is the following: Instead of executing Steps 8-10, it computes $\widetilde{X_0}, \widetilde{X_1}$ as the concatenation of $\{\widetilde{x_i}\}_{i \in I_0}, \{\widetilde{x_i}\}_{i \in I_1}$ respectively. It sends $I_0, I_1, \{(x_i, \theta_i), u_i\}_{i \in T}$ to $\sS$. On receiving $s, ct_0, ct_1$ from $\sS$ in Round 4, it extracts $m_0 := ct_0 \oplus h(s, \widetilde{X_0})$, $m_1 := ct_1 \oplus h(s, \widetilde{X_1})$, and sends $m_0, m_1$ to $\cF_\OT$.

The hybrids $\Hyb_0, \Hyb_1, \Hyb_2, \Hyb_3$ remain same as before. In $\Hyb_4$, instead of setting $r_b$, the challenger just outputs $m_b := ct_b \oplus h(s, \{\widetilde{x_i}\}_{i \in I_b})$ after Round 4. In $\Hyb_5$, instead of setting $r_0, r_1$, the challenger extracts $m_0 := ct_0 \oplus h(s, \{x_i\}_{i \in I_0})$, $m_1 := ct_1 \oplus h(s, \{x_i\}_{i \in I_1})$ after Round 4, and sends $m_0, m_1$ to $\cF_\OT$. The proof of indistinguishability between each pair of hybrids is similar to the prior proof.

The only security loss in the proof is between $\Hyb_0$ and $\Hyb_1$ (when we invoke the equivocality of the underlying commitment scheme).  Using Theorem \ref{thm:equcom} where we derived  $\mu(\secp,q, n_\com) = \frac{2qn_{\com}^{1/2}}{2^{\secp_\com/2}}$ and plugging $\lambda_\com = 2 \secp$, $n_\com = 2n$ (as we are committing to 2 bits at a time) where $n = 10\ 300\secp$ (this setting of $n$ is the same as that needed in the sender security part of the proof), we get $\frac{288q\sqrt{\secp}}{2^{\secp}}$ security against a malicious sender.\else
\fi
\section{Three round chosen input bit OT  via the XOR extractor}
In this section, we derive parameters required when using a seedless XOR extractor in place of a universal hash function, in \proref{fig:qot-3r-bb84}.
\label{sec:fixed-basis-3r-xor}
\begin{theorem}[Three round chosen input bit OT.]
\label{thm:4-round-chosenbitOT}
Consider \proref{fig:qot-3r-bb84} and modify it to use the XOR extractor in place of the universal hash function. In addition, instantiate the protocol with any non-interactive commitment scheme that is \emph{extractable} (\cref{def:fextractable}) and \emph{equivocal} (\cref{def:fequivocal}).  Then there exist constants $A,B$ such that \proref{fig:qot-3r-bb84} (modified to use XOR extractor) securely realizes (\cref{def:secure-realization}) $\cF_{\OT[1]}$. 

Furthermore, letting $\secp$ be the security parameter, $q$ be an upper bound on the total number of random oracle queries made by the adversary, and using the commitment scheme from \cref{sec:eecom-construction} with security parameter $\secp_\com = 4\secp$, for constants $A = 800, B = 800$, \proref{fig:qot-3r-bb84} (modified to use XOR extractor) securely realizes $\cF_{\OT[1]}$ with $\mu_{\sR^*}$-security against a malicious receiver and $\mu_{\sS^*}$-security against a malicious sender, where 
\[
    \mu_{\sR^*} = \frac{3\sqrt{10}q^{3/2}}{2^{\secp}}+\frac{148(q + 4800\secp + 1)^3 + 1}{2^{4\secp}} + \frac{38400q\secp}{2^{2\secp}}, \quad \mu_{\sS^*} = \frac{80\sqrt{3}q\secp}{2^{2 \secp}}.
\]
This requires a total of $2(A+B)\secp = 3200 \secp$ BB84 states.
\end{theorem}
\begin{proof}
The proof of this proceeds along the same line as that of \proref{fig:qot-3r-bb84}. We only describe the changes here.

\paragraph{Sender security} We define the same hybrids as used in the proof of sender security of \proref{fig:qot-3r-bb84}, and the proof of the indistinguishability between $\Hyb_0,\Hyb_1,\Hyb_2,\Hyb_3$ proceeds along the same way. For the proof of indistinguishability between $\Hyb_3$ and $\Hyb_4$ as well, the proof proceeds similarly except that the proof of some subclaims change. Specifically, \cref{subclaim:tau3r} now proves that for $A = 800, B = 800, q\geq 5$, $\Tr\left( \Pi_{\mathsf{bad}}^{0.245}\tau \right) \leq \frac{45q^3}{2^{2\secp}}$. In particular, in the proof of \cref{subclaim:tau3r} we get $\epsilon_{\mathsf{classical}}^\delta \leq \frac{7}{2^{\secp}}$ assuming $\delta = 0.245, \epsilon = 0.08326, \beta = 0.123, \gamma = 0.152, k = A\secp, n = (A+B)\secp, A = 800, B=800$. 

As in that proof then, by gentle measurement (\cref{lemma:gentle-measurement}), the $\tau$ defined in \cref{subclaim:tau3r} is within trace distance $\frac{3\sqrt{10}q^{3/2}}{2^{\secp}}$ of a state $\tau_{\mathsf{good}}$ in the image of $\bbI - \Pi_\mathsf{bad}^{0.245}$. And now conditioned on $\tau$ being in the image of $\bbI - \Pi_\mathsf{bad}^{0.245}$, and $A=800,B=800$, we show that $\Pr[\Hyb_3=1]=\Pr[\Hyb_4=1]$. 

To prove this, as in the proof of \cref{subclaim:3rseededminentropy}, we have registers $\cS_W$ are in a superposition of states $\ket{\br_{\widetilde{\btheta}_{W[1]}}}$, where $\Delta\left(\br, \widetilde{\bR}_{W}\right) < 0.245$. Recalling that $\cS_W = \{\cS_{i, d_i \oplus b \oplus 1}\}_{i \in \overline{T}\setminus U}$, we have, for a majority of $i \in \overline{T}\setminus U$, register $\cS_{i, \widetilde{\btheta}_i \oplus 1}$ is measured in basis $\widetilde{\btheta}_i \oplus 1$. Call these set of registers that are measured in basis $\widetilde{\btheta}_i \oplus 1$ as $M$. We then have that registers $\cS_M$ are in superposition of states $\ket{\br_{\widetilde{\btheta}_{M}}}$, where $\Delta\left(\br, \widetilde{\bR}_{M}\right) \leq \frac{0.245 \cdot |\overline{T}\setminus U|}{|M|} \leq \frac{0.245 \cdot |\overline{T}\setminus U|}{|\overline{T}\setminus U|/2} = 2\cdot 0.245 = 0.49 < \frac{1}{2}$.

Hence, using \cref{lemma:XOR-extractor} it then follows that the measured bit is a uniformly random bit.

\paragraph{Receiver Security} This proceeds in the same way as the proof of receiver security. The only difference is the security loss incurred during in the indistinguishability between $\Hyb_1$ and $\Hyb_2$. As before, $|\Pr[\Hyb_1=1]- \Pr[\Hyb_2=1]|\leq 
\frac{2qn_{\com}^{1/2}}{2^{\secp_\com/2}}$. Plugging $\lambda_\com = 4 \secp$, $n_\com = 3n$ (as we are committing to 3 bits at a time) where $n = 1600\secp$ (this setting of $n$ is the same as that needed in the sender security part of the proof), we get $\frac{80\sqrt{3}q\secp}{2^{2 \secp}}$ security against a malicious sender.
\end{proof}
\section{Classical sampling strategies} 
\label{app:classampling}

We analyze some common sampling strategies to find their classical error probability, $\epsilon_{\mathsf{classical}}^\delta$ in this section. 
Before doing so, we recall Hoeffding's inequality, which we make extensive use of below.
\paragraph{Hoeffding's inequality}
Let $X_1,\dots X_n$ be independent bounded random variables with $X_i \in [a,b]$ for all $i$, where $-\infty < a \leq b < \infty$. Let $X = \sum_{i\in[n]}X_i$. Then, for $\epsilon>0$,
\[
    \Pr[X \geq \E[X] + \epsilon] \leq \exp\left(-\frac{2\epsilon^2}{n(b-a)^2}\right), \, \Pr[X \leq \E[X] - \epsilon] \leq \exp\left(-\frac{2\epsilon^2}{n(b-a)^2}\right)
\]

\subsection{Random subset without replacement}
\label{appsubsec:randsampling}
    This corresponds to sampling $T \subseteq [n]$ of size $k$ uniformly at random without replacement and outputting $\omega(\bq_{T})$. Then, for $0 < \delta < 1$, $\epsilon_{\mathsf{classical}}^\delta \leq 2\exp\left(-2\left(1-\frac{k}{n}\right)^2\delta^2 k\right)$ \cite[Appendix B.1]{C:BouFeh10}.
    
    
\subsection{Random subset without replacement, using only part of the sample} 
\label{appsubsec:CK88sampling}
    This corresponds to sampling a set $T\subseteq [n]$ of size $k$ without replacement, then sampling $S \subseteq T$ uniformly at random among all possible subsets of $T$ and outputting $\omega(\bq_S)$. We provide a tighter analysis of this compared to \cite[Appendix B.4]{C:BouFeh10}. 
    
    \begin{lemma}
        For $0 < \delta,\beta < 1$ and $0 < \eta < \delta$, 
        \[
            \epsilon_{\mathsf{classical}}^\delta \leq 2\exp\left(-2\left(1-\frac{k}{n}\right)^2\eta^2 k\right)+2\exp\left(-(\delta-\eta)^2(1-\beta){k}\right)+ \exp\left(-\frac{\beta^2 k}{2}\right).
        \]
    \end{lemma}
    \begin{proof}
        \begin{align*}
            \epsilon_{\mathsf{classical}}^\delta 
            &= \max_{\bq} \Pr_{T,S}[|\omega(\bq_{\overline{T}}) - \omega(\bq_S)|\geq \delta]
        \end{align*}
        We have using the sampling strategy above, for $0 < \eta < 1$,
        \begin{align}
            \label{eq:classampeqT}
            \max_{\bq} \Pr_{T}[|\omega(\bq_{\overline{T}}) - \omega(\bq_T)|\geq \eta] &\leq 2\exp\left(-2\left(1-\frac{k}{n}\right)^2\eta^2 k\right)
        \end{align}
        In the following, given a string $\bq$, we find a bound on $\Pr_S[|\omega(\bq_{T}) - \omega(\bq_S)|\geq \gamma]$. Relating this to the above, we will get the final bound. Conditioning on the size of $S$ being $s$, the sampling of $S$ corresponds to sampling a uniform subset of size $s$. 
        We have the following subclaim:
        \begin{subclaim}
            \[
                \Pr_S\left[|\omega(\bq_T) - \omega(\bq_S)|\geq \gamma \, \big| \, |S|=s\right] \leq 2\exp\left(-2\gamma^2s\right)
            \]
        \end{subclaim}
        \begin{proof}
            We find a bound using Hoeffding's inequality applied to sampling $S$ with replacement (sampling $S$ without replacement will only be tighter). For each $i \in [s]$, let $Y_i=1$ if the $i^{th}$ drawn element of $S$ is $1$. Let $Y = \sum_{i \in [s]}Y_i, \overline{Y} = Y/s$. Then, using Hoeffding's inequality, since $Y_i$ are independent bounded random variables, for $\gamma'>0$, $\Pr[|Y - \E[Y]| \geq \gamma'] \leq 2\exp\left(-\frac{2\gamma'^2}{s}\right)$ or $\Pr[|\overline{Y} - \E[\overline{Y}]| \geq \gamma] \leq 2\exp\left(-{2\gamma^2s}\right)$ for $\gamma = \gamma'/s$. Since, $\E[\overline{Y}] = \omega(\bq_T)$, we have for $0 < \gamma < 1$, $\Pr\left[|\omega(\bq_S) - \omega(\bq_T)| \geq \gamma \, | \, |S|=s\right] \leq 2\exp\left(-{2\gamma^2s}\right)$.
        \end{proof}
        Using the distribution of $|S|$, we have,
        \begin{subclaim}
            For $0 < \gamma,\beta < 1$,
            $\Pr_S\left[|\omega(\bq_T) - \omega(\bq_S)|\geq \gamma \right] \leq 2\exp\left(-2\gamma^2(1-\beta)\frac{k}{2}\right) + \exp\left(-\frac{\beta^2 k}{2}\right)$.
        \end{subclaim}
        \begin{proof}
            To find the distribution of $|S|$, note that sampling $S$ corresponds to choosing each element of $T$ at random with probability $1/2$. For $i \in [k]$, let $X_i=1$ if $i^{th}$ element is chosen to be part of the set, and let $X = \sum_{i\in [k]} X_i$. Then, $\Pr[X_i=1]=1/2$ and each $X_i$ is an independent bounded random variable. Using Hoeffding's inequality, for $\beta' > 0$, $\Pr[X \leq \E[X] - \beta'] \leq \exp\left(-\frac{2\beta'^2}{k}\right)$. Setting $\beta' = \E[X]\beta$, we get, $\Pr[X \leq (1-\beta)\E[X]] \leq \exp\left(-\frac{2\beta^2 (\E[X])^2}{k}\right)$. Since $\beta'>0$, we have $\beta>0$. In particular, for $0 < \beta < 1$, we have, $\Pr_S[|S| \leq (1-\beta)\frac{k}{2}] \leq \exp\left(-\frac{\beta^2k}{2}\right)$. 
        
            Therefore, 
            \begin{align*}
                \Pr\left[|\omega(\bq_{S}) - \omega(\bq_T)|\geq \gamma \right] &= \sum_{s \leq (1-\beta)k/2} \Pr\left[|\omega(\bq_{S}) - \omega(\bq_T)|\geq \gamma \, \big| \, |S|=s\right] \Pr[|S|=s]\\
                &+ \sum_{s > (1-\beta)k/2} \Pr\left[|\omega(\bq_{S}) - \omega(\bq_T)|\geq \gamma \, \big| \, |S|=s\right] \Pr[|S|=s]\\
                &\leq \exp\left(-\frac{\beta^2 k}{2}\right) + \sum_{s > (1-\beta)k/2} 2\exp\left(-2\gamma^2s\right)\\
                &\leq \exp\left(-\frac{\beta^2 k}{2}\right) + 2\exp\left(-2\gamma^2(1-\beta)\frac{k}{2}\right)
            \end{align*}
        \end{proof}
        
        Combining the above with \cref{eq:classampeqT}, we get, for any string $\bq \in \zo^n$, 
        \begin{align*}
            \Pr\left[|\omega(\bq_{\overline{T}}) - \omega(\bq_S)|\geq \eta + \gamma \right] &\leq 2\exp\left(-2\left(1-\frac{k}{n}\right)^2\eta^2 k\right)\\
            &+ 2\exp\left(-2\gamma^2(1-\beta)\frac{k}{2}\right) + \exp\left(-\frac{\beta^2 k}{2}\right).
        \end{align*}
        For $\delta = \eta + \gamma$, we get, for $0 < \delta < 1, 0 < \eta < \delta$ and $0 < \beta < 1$,
        \begin{align*}
            \Pr\left[|\omega(\bq_{\overline{T}}) - \omega(\bq_S)|\geq \delta \right] &\leq 2\exp\left(-2\left(1-\frac{k}{n}\right)^2\eta^2 k\right)\\
            &+ 2\exp\left(-(\delta - \eta)^2(1-\beta){k}\right) + \exp\left(-\frac{\beta^2 k}{2}\right).
        \end{align*}
    \end{proof}
    
\subsection{Intersection of two uniform subsets and then using part of the sample}
\label{appsubsec:3roundsamplingstrategy}
    This corresponds to sampling two independent uniform subsets $T,U \subseteq [n]$, each of size $k$, setting $S = T \cap U$, and then taking a random subset of $S$ (among all possible subsets of $S$). This is the strategy followed in our 3 round chosen input OT protocol (\cref{sec:bb84-3r-ot}). In terms of the sampling strategy definition in \cref{subsec:quantum-sampling}, the above strategy can be thought of as the following sampling strategy $\Psi$: \fflater{in future it would be interesting to extend BF10 to more general strategies}
    \begin{itemize}
        \item $P_{T'}$: Sample two independent and uniform subsets of $[n]$ each of size $k$. Let $s$ denote their intersection size. Fix $s$, and discard the subsets themselves. 
        Sample a random subset $T'$ of $[n]$, of size $2k-s$, and output $T'$.
        \item $P_{S'}$: Given $T'$, reverse calculate $s$ as $2k-|T'|$. Sample a uniformly random subset $S$ of $T'$, of size $s$. Sample a uniformly random subset $S'$ of $S$ (among all possible subsets of $S$). Output $S'$.
        \item $f(T',\bq_{T'},S')$: Output $\omega(\bq_{S'})$.
    \end{itemize}
    We prove then the following:
    \begin{lemma}
    \label{lemma:intersectionsampling}
        For $0 < \epsilon,\beta,\delta< 1$ and $0 < \gamma < \delta$,
        \begin{align*}
            \epsilon_{\mathsf{classical}}^\delta(\Psi) \leq & 2\exp\left(-2\left(\frac{(n-k)^2 - 3\epsilon k^2}{(n-k)^2+(1-2\epsilon)k^2}\right)^2\gamma^2 (1-\epsilon)\frac{k^2}{n}\right)\\
            & + 2\exp\left(-(\delta-\gamma)^2(1-\beta)(1-\epsilon)\frac{k^2}{n}\right)\\ 
            &+ \exp\left(-\frac{\beta^2(1-\epsilon)k^2}{2n}\right) + 2\exp\left(-\frac{2\epsilon^2k^3}{n^2}\right)
        \end{align*}
    \end{lemma}
    \begin{proof}
        \begin{align*}
            \epsilon^\delta_{\mathsf{classical}}(\Psi) &= \max_\strq \Pr_{T' \leftarrow P_{T'}, S' \leftarrow P_{S'}}\left[\strq \notin B_{T',S'}^\delta\right]\\
            &= \max_\strq \Pr_{T'
            \leftarrow P_{T'}, S' \leftarrow P_{S'}}\left[ |\omega(\strq_{\overline{T'}}) - \omega(\strq_{S'})| \geq \delta \right]
        \end{align*}
    To relate $\omega(\bq_{\overline{T'}})$ with $\omega(\bq_{S'})$, consider the following equivalent sampling strategy $\Psi'$: 
    \begin{itemize}
        \item Sample two independent and uniform subsets of $[n]$ each of size $k$. Let $s$ denote their intersection size. Fix $s$, and discard the subsets themselves. 
        Sample a random subset $R$ of $[n]$, of size $n-2(k-s)$. 
        \item Sample a uniformly random subset $S$ of $R$, of size $s$. Sample a uniformly random subset - $S'$, of $S$ (among all possible subsets of $S$). Output $\omega(\bq_{S'})$.
    \end{itemize}
    In intuitive terms, compared to the original sampling strategy where $T,U \subseteq [n]$ of size $k$ each were sampled, in sampling strategy $\Psi$, set $T'$ corresponds to sampling $T \cup U$ first, and then sampling $S = T \cap U$. And in the above sampling strategy, $\Psi'$, sampling $R$ corresponds to sampling $\overline{T \cup U} \cup (T \cap U)$, or in terms of sampling strategy $\Psi$ it corresponds to sampling $\overline{T'} \cup S$. Therefore, $\omega(\strq_{\overline{T'}})$ in sampling strategy $\Psi$ is equivalent to $\omega(\strq_{R \setminus S})$ in $\Psi'$.
    Therefore, 
    \begin{align}
        \epsilon^\delta_{\mathsf{classical}}(\Psi) &= \max_\strq \Pr_{T'
        \leftarrow P_{T'}, S' \leftarrow P_{S'}}\left[ |\omega(\strq_{\overline{T'}}) - \omega(\strq_{S'})| \geq \delta \right] \nonumber\\
        &= \max_\strq \Pr_{s,R,S,S'}\left[ |\omega(\strq_{R\setminus S}) - \omega(\strq_{S'})| \geq \delta \right] \label{eq:classamplingeq1}
    \end{align}
    But note that given $s$, sampling $S$ and $S'$ from $R$ corresponds exactly the sampling analyzed in \cref{appsubsec:CK88sampling}.
    Therefore, using the same result, we get, for $0 < \delta,\beta <1$ and $0 < \gamma < \delta$,
    \begin{align*}
        \max_\strq \Pr_{S,S'}\left[ |\omega(\strq_{R\setminus S}) - \omega(\strq_{S'})| \geq \delta \, \big| \, |S|=s \right]
        &\leq 2\exp\left(-2\left(1-\frac{s}{n-2(k-s)}\right)^2\gamma^2 s\right)\\
        &+2\exp\left(-(\delta-\gamma)^2(1-\beta){s}\right)+ \exp\left(-\frac{\beta^2 s}{2}\right).
    \end{align*}
    We now factor in the distribution of $s$, which we analyze using Hoeffding's inequality applied to sampling with replacement. Consider the following experiment - sample subset $T \subseteq [n]$ size $k$ uniformly at random. Now sample $k$ elements from $[n]$ with replacement and call that set $U$. Set $s = |T \cap U|$. Let for all $i \in [k]$, $X_i=1$ iff the $i^{th}$ drawn element for $U$ is drawn from $T$. Then, $\Pr[X_i=1] = k/n$. Let $X = \sum_{i \in [k]}X_i$ represent $s=T \cap U$. $\E[X] = k^2/n$. Since $X_i$ are independently drawn binary random variables, applying Hoeffding's inequality, for $\epsilon' > 0$,
        \begin{align*}
            \Pr[|X - \E[X]|\geq \epsilon'] &\leq  2\exp(\frac{-2\epsilon'^2}{k})  \\
            \implies \Pr\left[\bigg|s - \frac{k^2}{n}\bigg|\geq \epsilon'\right] &\leq 2\exp(\frac{-2\epsilon'^2}{k})\\
            \implies \Pr\left[\bigg|s - \frac{k^2}{n}\bigg|\geq \frac{\epsilon k^2}{n}\right] &\leq 2\exp(\frac{-2\epsilon^2k^3}{n^2})
        \end{align*}
        where we substituted $\epsilon' = \frac{\epsilon k^2}{n}$. Therefore, $s \in \left(\frac{(1-\epsilon)k^2}{n}, \frac{(1+\epsilon)k^2}{n}\right)$ except with probability $2\exp\left(\frac{-2\epsilon^2k^3}{n^2}\right)$ for $0<\epsilon < 1$.
        We have then for $0 < \epsilon,\beta,\delta< 1$ and $0 < \gamma < \delta$,
        \begin{align*}
            &\Pr_{s,R,S,S'}\left[|\omega(\strq_{R\setminus S}) - \omega(\strq_{S'})| \geq \delta \right]\\
            &= \sum_{s_0:\, |s_0 - {k^2}/{n}|\geq \frac{\epsilon k^2}{n}} \Pr_{R,S,S'}\left[|\omega(\strq_{R\setminus S}) - \omega(\strq_{S'})| \geq \delta \, | \, s=s_0 \right] \Pr_s\left[s=s_0\right] \\
            &+ 
            \sum_{s_0:\, |s_0 - {k^2}/{n}|< \frac{\epsilon k^2}{n}} \Pr_{R,S,S'}\left[|\omega(\strq_{R\setminus S}) - \omega(\strq_{S'})| \geq \delta \, | \, s=s_0 \right] \Pr_s\left[s=s_0\right]\\
            &\leq \sum_{s_0:\, |s_0 - {k^2}/{n}|\geq \frac{\epsilon k^2}{n}}\Pr_s\left[s=s_0\right] + \sum_{s_0:\, |s_0 - {k^2}/{n}|< \epsilon} 2\exp\left(-2\left(1-\frac{s}{n-2(k-s)}\right)^2\gamma^2 s\right)\\
            &\qquad \qquad \qquad \qquad \qquad \qquad +2\exp\left(-(\delta-\gamma)^2(1-\beta){s}\right)+ \exp\left(-\frac{\beta^2 s}{2}\right)\\
            &\leq 2\exp\left(-2\left(\frac{(n-k)^2 - 3\epsilon k^2}{(n-k)^2+(1-2\epsilon)k^2}\right)^2\gamma^2 (1-\epsilon)\frac{k^2}{n}\right) \\
            &+ 2\exp\left(-(\delta-\gamma)^2(1-\beta)(1-\epsilon)\frac{k^2}{n}\right) + \exp\left(-\frac{\beta^2(1-\epsilon)k^2}{2n}\right) + 2\exp\left(-\frac{2\epsilon^2k^3}{n^2}\right)
        \end{align*}
        where we substituted the upper bound and lower bound of $s$ to get the last inequality. Since this is true for any string $\bq$, it is also true for $\max_{\bq}$ and hence using \cref{eq:classamplingeq1}, we get $\epsilon_{\mathsf{classical}}^\delta$ is bounded by the quantity above.
    \end{proof}

\newpage

\clearpage

\end{document}